\documentclass{amsart}[11pt,oneside]

\usepackage{array, multirow}
\usepackage{amssymb}
\usepackage{amsfonts}
\usepackage{amsmath,amsthm}
\usepackage{caption}
\usepackage{subcaption}
\usepackage{setspace}
\usepackage{color}
\usepackage{xcolor}
\usepackage{breakcites}
\usepackage{amsaddr}
\usepackage{graphicx}
\usepackage[utf8]{inputenc}
\usepackage[english]{babel}
\usepackage{pifont}
\usepackage{changepage}

\usepackage{tabularx, booktabs}

\makeatletter
\def\@tocline#1#2#3#4#5#6#7{\relax
  \ifnum #1>\c@tocdepth % then omit
  \else
    \par \addpenalty\@secpenalty\addvspace{#2}%
    \begingroup \hyphenpenalty\@M
    \@ifempty{#4}{%
      \@tempdima\csname r@tocindent\number#1\endcsname\relax
    }{%
      \@tempdima#4\relax
    }%

        \rightskip\@pnumwidth plus4em 
          \parfillskip-\@pnumwidth
    #5\leavevmode\hskip-\@tempdima
      \ifcase #1
       \or\or \hskip 1em \or \hskip 2em \else \hskip 3em \fi%
      #6\nobreak\relax
    \dotfill\hbox to\@pnumwidth{\@tocpagenum{#7}}\par
    \nobreak
    \endgroup
  \fi}
\makeatother

\newtheorem{theorem}{Theorem}[section]

\newtheorem{lemma}[theorem]{Lemma}
\date{\today}
\newtheorem{corollary}[theorem]{Corollary}
\newtheorem{remark}[theorem]{Remark}
\newtheorem{assumption}[theorem]{Assumption}
\newtheorem{example}[theorem]{Example}

\newcommand{\ul}{\ensuremath{\lfloor t/\Delta_n\rfloor}}

\newcommand{\ulT}{\ensuremath{\lfloor T/\Delta_n\rfloor}}

\newcommand{\indicator}{\ensuremath{\mathbb{I}}}

\allowdisplaybreaks

\newcommand{\dennis}[1]{\textcolor{red}{#1}}

\title[Dynamically Consistent Analysis of Realized  Covariations in Term Structure Models]{Dynamically Consistent Analysis of Realized  Covariations in Term Structure Models
}

\calclayout
\author{Dennis Schroers}
\address{Institute of Finance and Statistics and Hausdorff Center for Mathematics\\University of Bonn}
\usepackage{mathptmx}

\begin{document}

\maketitle
\vspace{-1cm}
\begin{scriptsize}
\begin{abstract}
%We introduce a method for measuring covariations in a general  arbitrage-free term structure setting. In particular, we describe how to  determine the number of statistically relevant random drivers in bond-market dynamics. Unlike traditional dimension reduction techniques, our approach does not assume finite-dimensional dynamics, stationarity, or the existence of moments. Using an asymptotic theory, we identify scaling limits of empirical covariances from simple trading strategies, enabling principal component analysis based on difference returns.
%We show that these scaling limits correspond to quadratic covariations of latent drivers, facilitating the inference of random driver numbers from discrete bond price data. Additionally, we address the impact of outliers, such as financial crises, by separating continuous and jump components of quadratic variation and employing a truncation technique for robust estimation.

%Empirical analysis of bond market data from 1990 to 2022 reveals significant temporal variation in driver dimensionality, with consistently high dimensions needed to explain bond market variations. Our framework offers a flexible, robust approach to term structure analysis, validated through simulations and practical applications.

In this article we show how to analyze the covariation of bond prices nonparametrically and robustly,  staying consistent with  a general no-arbitrage setting.  This is, in particular,  motivated by the problem of identifying the number of statistically relevant  factors in the bond market under minimal conditions.  We apply this method in an empirical study which suggests that a high number of factors is needed to describe the term structure evolution and that the term structure of volatility varies over time.
\end{abstract}
%\tableofcontents


\end{scriptsize}

%\newpage
% Dimension reduction for term structure data is often done in a preprocessing step based on an estimator of an empirical covariance of yields, discount curves or forward rate curves. However,  empirical covariances are only a sensible pivot to base ...

\section{Introduction}

We present a nonparametric method to measure covariations in a general arbitrage-free term structure setting in the spirit of \cite{bjork1999}.  
A motivation is to determine the number of statistically relevant random drivers needed to describe bond market dynamics under minimal assumptions and in a dynamically consistent manner.  That is, by working coherently in the abstract setting of \cite{bjork1999}, we circumenvent the well-known consistency problems of finding arbitrage-free finite-dimensional dynamics that reflect the empirical observations. %Customary dimension reduction techniques  are conducted   independently of a subsequent dynamical analysis and implicitly impose assumptions such as consistency of the data with finite-dimensional arbitrage-free dynamics,  stationarity and the existence of moments.  In contrast,  we allow for an economically meaniningful measurement of covariations and dimensionality analyses indepentently of such assumptions.  

In the bond market  the notion of a zero coupon bond is fundamental. A zero coupon bond guarantees its holder at time $t$ a fixed amount of money at some time $t+x$ in the future. The cost $P_t(x)$ of entering this contract at time $t$ is the price of this bond, which depends on the time to maturity $x$, implying  a price curve $x\mapsto P_t(x)$  at each time $t\geq 0$ called the discount curve. 
We assume to observe bond price or yield curve data, potentially derived by smoothing as in \cite{FPY2022} or \cite{liu2021},  such that we can recover log bond prices  
$$P_{i,j}^n:=\log P_{i\Delta_n}(j\Delta_n)\quad j=0,1,...,\lfloor M/\Delta_n\rfloor, \,\, i=0,1,...,\lfloor T\Delta_n\rfloor$$
for a resolution $\Delta_n=1/n$ and
with different maturities
$j\Delta_n $ and on different time points $i\Delta_n $.
Here $M>0$ is some maximal time to maturity (e.g. $M=10$ or $M=30$ when time is measured in years) and $T$ is the time until which the data are observed or considered.  
% Econometric studies are often based on implied interest rate curves, such as the yield curve, which for $x,t\geq 0$ is defined as $y_t(x)=-\log(P_t(x))/x$.  

Often,  risk factor analyses are conducted on the basis of transformations of the discount curve, such as yield differences or excess returns,
  %$P_{i+1,j+1}^n- P_{i,j}^n-P_{i+1,1}^n- P_{i,0}^n$, 
  which typically suggests that three factors explain a large amount of variation in bond market dynamics, c.f. \cite{Litterman1991}. 
  Recently \cite{Crump2022},  raised the concern that  these low-dimensional factor structures are obtained irrespectively of the data generating process due to the high correlation of bond prices with close maturities.  
%  that dimensionality reduction methods based on these yield differences or di tend  to suggest a low cross-sectional factor structure irrespectively of the data generating process due to the mechanically induced high local correlation in yield or discount curves.  
To remedy this effect,  dimension reduction could be based on difference returns, which are the returns of the trading strategy of buying an $x+\Delta$-maturity bond and shorting an $x$-maturity bond. Precisely,  difference returns
are defined  for $i=1,...,\ulT-1$ and $j=1,...,\lfloor M/\Delta_n\rfloor-1$
by
\begin{align}\label{Difference returns}
d_i^n(j) :=  P_{i+1,j}^n- P_{i,j+1}^n- P_{i+1,j-1}^n+ P_{i,j}^n.
 % \tilde \Delta_t P (x)-\tilde \Delta_t P (x-\Delta)%=   (0\log P_{t}(x)-\log P_{t}(x+\Delta))-( \log P_{t+\Delta}(x-\Delta)-\log P_{t+\Delta}(x)).
\end{align}
In this article we develop an asymptotic econometric theory for the realized covariations of these difference returns, that is,  for $j_1,j_2=1,...,\lfloor M/\Delta_n\rfloor$ we analyse the covariations
$$\hat q_T^n(j_1,j_2):=
\sum_{i=1}^{\ulT} d_i(j_1)d_i(j_2).$$

Importantly, while $\ulT^{-1} q_T^n$ is the empirical covariance of the data $d_1^n,...d_{\ulT-1}^n$,  assuming w.l.o.g.  $\mathbb E[d_i^n]=0$, we do not consider it as an estimator of the population covariance of difference returns.  Such an interpretation 
%Dimension reduction based on difference returns, which are economically meaningful and less affected by  mechanical overlaps, can be conducted as follows:
%Assuming for the moment that w.l.o.g. $\mathbb E[d_i^n]=0$ ,  one calculates the empirical covariance$$\hat C_{\Delta_n}(j_1,j_2)=\frac 1{\ulT}\sum_{i=1}^{\ulT} d_i(j_1)d_i(j_2),$$
% to derive its  eigenstructure to analyse the major modes of variations in the data.   If   $d_1^n,...d_{\ulT-1}^n$ are i.i.d.  the empirical covariance converges to the population covariance when $T\to \infty$. 
%From such an estimate  by analysing the eigenstructure of such an estimator. 
% However,  this procedure is not free of conceptual challenges. For instance,  
  is
 not invariant with respect to the resolution $\Delta_n$ and requires the restrictive assumption that difference returns are i.i.d. or at least covariance stationary and ergodic. More importantly, it is not clear how dimension reduction can be conducted without entailing arbitrage opportunities.
General  arbitrage-free term structure models in the sense of \cite{bjork1997} and \cite{FilipovicTappeTeichmann2010b} require that forward rates
$ f_t(x)=-\partial_x\log(P_t(x))$ for $ x,t\geq 0,$
  satisfy dynamics of the form 
 \begin{align}\label{SPDE}
    df_t= \partial_x f_t dt+dX_t, \quad t\geq 0.
 \end{align}
 where the equation holds in an appropriate function space.
The latent process
$X$ is a possibly infinite-dimensional   It{\^o} semimartingale 
$$X_t:= \int_0^t \alpha_s ds+\int_0^t\sigma_s dW_s+J_t.
    $$
where $\alpha$ is a curve-valued drift process, $\sigma$ is the (in general operator-valued) volatility, $W$ is an (in principle infinite-dimensional) Wiener process and $J$ is a jump process that we assume to model rare extreme events (c.f.  Section \ref{Sec:  semimartingales in Hilbert spaces} for the details).   Under a risk-neutral measure, the drift $\alpha$ can in addition be characterized as a deterministic function of the volatility $\sigma$ and characteristics of the jump process $J$ (c.f.  \cite{bjork1997}, \cite{FilipovicTappeTeichmann2010b}).   Parametrizations of forward cuves are then required to be viable in the dynamic setting \eqref{SPDE} to avoid the introduction of arbitrage opportunities to the model.  This is known to be an intricate problem  in term structure modelling (see e.g.  \cite{bjork1999}, \cite{bjork2001}, %\cite{bjork1999b},
%\cite{bjork2002},
\cite{filipovic2003}, %\cite{filipovic2002}, 
\cite{Filipovic2000}, \cite{filipovic2000b}%\cite{filipovic2004}
) and some frequently employed parametrizations of forward curves are incompatible with arbitrage-free dynamics or induce restrictive additional conditions (see e.g. \cite{filipovic1999}).
The concerns on realized covariations of difference returns can be resolved when we consider infill asymptotics ($n\to \infty$).  Precisely,
we show without imposing further assumptions and interpreting $\hat q_T^n$ as a piecewise constant kernel 
that 
\begin{equation}\label{Covscalinglimit}
 \lim_{n\to \infty}\Delta_n^{-2}\hat q_T^n(\lfloor x/\Delta_n\rfloor,\lfloor y/\Delta_n\rfloor)=  \lim_{n\to\infty} \sum_{i=1}^{\ulT} \Delta_i^n X(x)\Delta_i^n X(y),
\end{equation}
where the limits hold in $L^2([0,M]^2)$. The right hand side describes the quadratic covariation of the latent driver $X$, which always exists (see e.g.   \cite{Schroers2024}) and can equivalently be described as the limit of covariance operators in $L^2(\mathbb R_+)$ by
 \begin{equation}\label{Quadratic variation abstract probability limit}
[X,X]_t:=\lim_{\Delta_n \downarrow 0}\sum_{i=1}^{\ul} \langle \Delta_i^n X,  \cdot \rangle \Delta_i^n X,
\end{equation}
where $\langle \cdot, \cdot\rangle$ is the inner product of $L^2(\mathbb R_+)$ and $\Delta_i^n X:= X_{i\Delta_n}-X_{(i-1)\Delta_n}$ denotes the $i$'th increment of $X$. 
  %The quadratic variation $[X,X]$ is, hence, an attractive object for dimension reduction, since it has has a practially relevant interpretation as the covariation of difference returns and since $X$ is the only source of randomness in the model.

%The interpretation of the realized variations $\hat q_T^n$ as estimators of the covariations of the abstract latent driver $X$ resolves the concerns on dimension reductions via covariances raised before.
The limit \eqref{Covscalinglimit} implies that it is possible to infer on the number of random drivers in the bond market on the basis of discrete bond price data without further assumptions on moments, stationarity and ergodicity. 
 In fact,  if the quadratic variation of $X$ is $d$-dimensional for $d\in \mathbb N$,  then $X$ is $d$-dimensional and its state space (until time $T$) is spanned by the eigenvectors corresponding to the nonzero eigenvalues of $[X,X]_T$.  The eigenvalues of $[X,X]_T$ indicate the amount of variation that the corresponding random factor explained of $X$ up to time $T$.  %Thus, the realized covariation of difference returns on the left side of \eqref{Covscalinglimit} is informative on the number of random drivers. 
This is in contrast to the  explained variation of factors derived from a covariance of $f$ or its increments,  which is  a priori not informative on the number of random drivers.  In fact,  it is  possible that $f$ is infinite-dimensional, while  $X$ is a one-dimensional  process (c.f.  Example 4.1 in \cite{Schroers2024}). 
The most striking advantage of the  interpretation of the realized covariation of difference returns in \eqref{Covscalinglimit} is, however, that exchanging $X$ by an arbitrary finite-dimensional semimartingale (with the correct form of the drift) in the formulation of the dynamics in \eqref{SPDE} does not affect the capability of the model to be free of arbitrage, such that an investigation of the number of relevant factors can be conducted independently of further consistency conditions.  This underlines that our abstract infinite-dimensional setting relaxes the analysis of the term structure when compared to models in which state spaces are assumed to be finite-dimensional a priori.  An additional advantage of quadratic variations is that they are naturally interpreted as time-varying objects enabling their temporal analysis.

In practice, the distortion of the measurements due to ouliers can bias the analysis of the relevant factors.
%Despite the theoretical advantages of quadratic variations,   the distortion of the measurements due to ouliers remains as a practical issue.
For instance,  in the context  of sudden interest rate movements during an economic crisis it is possible that a single outlying difference return impacts the measurement of covariations and, thus, the measured dimensionality  of the driver $X$ substantially.  
For this reason, besides $[X,X]$, 
the continuous part $[X^C,X^C]$ of the quadratic variation  where $X^C_t= \int_0^t \alpha_s ds+\int_0^t \sigma_s dW_s$ is central for the task of identifying the statistically relevant number of random drivers,  considering the jump part to model outlying events.  
  We describe how estimation of the continuous quadratic covariation is possible by a truncation technique, which sets outlying difference returns in the realized covariation on the left of \eqref{Covscalinglimit} to $0$. 
We derive rates of convergence and a central limit theorem for these estimators and also show how the long-time limit of $[X^C,X^C]_T/T$ as $T\to \infty$ can be estimated, if it exists. 
Our limit theory holds under weak assumptions, which mainly reflect those for finite-dimensional semimartingales, although, $f$ does not need to be a semimartingale.

We conduct an emprical study on the relevant drivers in the market via this covariation estimates based on real bond market data.  The procedure is numerically equivalent to a principal component analysis based on the (truncated) empirical covariance of difference returns with a daily resolution and mean zero.
However, the classification of jumps takes into account the abstract setting \eqref{SPDE}.  
By investigation of a truncated version of the realized covariations $q_j^*(x,y):=\Delta_n^{-2}
\sum_{i=\lfloor ( j-1)/\Delta_n \rfloor}^{\lfloor j/\Delta_n \rfloor} d_i(\lfloor x/\Delta_n\rfloor)d_i(\lfloor y/\Delta_n\rfloor)$ for any year $j$ from 1990 to 2022,  we find evidence for the dimension of the driver to vary from year to year but  also to be consistently high (in each year more than $8$ drivers are needed to explain at least $99\% $  of the variation). %This confirms empirical observations in \cite{Crump2022} on the basis of raw bond market data. 
We further observe that quadratic variations vary in shape over time and not just their level.
We provide Monte-Carlo evidence for the validity of the limit theory in the context of  sparse and noisy data.

Formal validity of our method is guaranteed  by relating difference returns via cross-sectional and temproal discretization to the abstract setting \eqref{SPDE} and then apply the
  results from the article \cite{Schroers2024}.  %In the continuous case ($X=X^C$),  this theory was initialized in \cite{Benth2022} and \cite{BSV2022}. 
The truncation procedure is inspired from the truncated realized variation estimators of \cite{Mancini2001, Mancini2004, Mancini2009} and \cite{Jacod2008} for finite-dimensional semimartingales.  %The generality of our limit theory reflects the role of quadratic variations in high-frequency limits of quadratic variations for semimartingales in the context of finite-dimensional stock market the infinite-dimensional nature of fixed income markets and term structures.  
Here,  we also provide a data-driven variant of the truncation rule, to account for functional outliers in a similar way as the trimmed least squares method in \cite{ren2017}.
%Importantly,  the process $f$, and neither the bond price curves $(P_t)_{t\geq 0}$ follow infinite-dimensional semimartingales, which underlines adequacy of difference returns for analysing covariations in the bond market.  
%The difference to the finite-dimensional setting can be found in the rates of convergence and the additional regularity assumptions that need to be imposed to derive asymptotic normality.
Naturally, our asymptotic theory also allows for nonparametric estimation of characteristics of infinite-dimensional volatility models in continuous time employed for term structure modeling (c.f.  \cite{BenthRudigerSuss2018}, \cite{BenthSimonsen2018}, \cite{BenthSgarra2021}, \cite{Cox2020}, \cite{Cox2021}, \cite{Petersson2022} and \cite{cox2023}).

%This question can just be answered in an infinite-dimensional context as ours,  due to the intricate consistency problems term structure models are prone to.Our application to  bond market data underlines the importance to consider jump robust estimators and reveals that a high number of random drivers might be necessary to describe the evolution of the bond market accurately.  \\
%As a leading example and motivation, we here provide an extensive treatment of the analysis of quadratic variations for arbitrage-free term structure models in the spirit of \cite{bjork1997}.

%At this point, we did not specify any particular discrete sampling scheme for the process $f$,  which has to be considered separately for any application.  However,  we describe in detail in Section \ref{Sec: term structure models theory} how the theory can be applied to  discrete bond price data:
%In the context of bond markets, our limit theory also reveals that quadratic variations of the drivers of forward rate curves coincide with scaling limits of empirical covariances corresponding to certain short-term trading strategies. 

%which are, yet,  mostly approached from a theoretical angle in the context of mathematical finance.

The article is structured as follows: Section \ref{Sec: bond market primer} describes the general bond market setting that we consider for this article.  Section
\ref{Sec: term structure models theory}  presents the estimation theory for the central application of term structure models.  Identification for the quadratic variations of $X$, $X^C$ and $J$ on the basis of difference return variations can be found in Section \ref{Sec: Identification of the quadratic variaiton}, while rates of convergence for estimating $[X^C,X^C]$ and a central limit theorem can be found in Section \ref{Sec: Error Bounds}. Section \ref{Sec: Long time Volatility estimation} discusses long-time asynmptotics for estimation of a stationary mean of $[X^C,X^C]_T/T$.  Section 
\ref{Sec: Practical Considerations} contains practical considerations on smoothing of discrete bond price data and presents a data-driven truncation rule for robust estimation. 
 Section \ref{Sec: Simulation Study} provides a simulation study. %investigating the impact of sparse, noisy, and irregular samples. 
 Finally, we apply our theory to bond market data in Section \ref{Sec: Empirical Study}.
Technical proofs of our results along with further remarks on the simulation scheme and additional empirical results can be found in the appendix. % supplementary material to this article.

\subsection{Technical preliminaries and notation}
Let $I$ be an interval in $\mathbb R$.  
We write $\langle h,g\rangle= \int_I h(x)g(x) dx$ for the $L^2$-scalar product of two elements $h,g\in L^2(I)$ as well as $\|h\|=\sqrt{\langle h,h\rangle}$ for the norm.
%If $H$ is a separable Hilbert space, the corresponding inner product and norm are denoted by $\langle \cdot,\cdot\rangle_H$ and $\|\cdot\|_H$ and the identity operator on $H$ by $I_H.$
%If $G$ is another separable Hilbert space,  $h\in H$ and $g\in G$,  $L(G,H)$ denotes the space of bounded linear operators from $G$ to $H$ and $L(H):=L(H,H)$.  
We write $L_{\text{HS}}(L^2(I))$ for the Hilbert space of Hilbert-Schmidt operators from $L^2(I)$ into itself and $\|\mathcal T_k\|_{\text{HS}}$ for the Hilbert-Schmidt norm of an $\mathcal T_k\in L_{\text{HS}}(L^2(I))$.
%$\Vert B\Vert_{L_{\text{HS}}(U,H)}^2:=\sum_{n=1}^{\infty}\| B e_n\|_H^2<\infty,$
%for an orthonormal basis $(e_n)_{n\in\mathbb N}$ of $G$. When $G=H$, we write $L_{\text{HS}}(H):=L_{\text{HS}}(H,H)$. 
Recall, that 
a Hilbert-Schmidt operator $\mathcal  T_k:L^2(I)\to L^2(I)$ can be uniquely associated to a kernel $k\in L^2(I^2)$ such that $\|\mathcal T_k\|_{\text{HS}}=\|k\|_{L^2(I^2)}$ and
\begin{equation}\label{Hilbert-Schmidt kernel equivalence}
   \mathcal  T_kf(x)=\int_I k(x,y)f(y)dy\quad \forall f\in L^2(I).
\end{equation}
Importantly, the operator $h\otimes g:= \langle h,\cdot\rangle g$ is Hilbert-Schmidt for two elements $h,g \in L^2(I)$. %Recall that $B$ is nuclear, if  $\sum_{n=1}^{\infty}\| B e_n\|_H<\infty$ for some orthonormal basis $(e_n)_{n\in\mathbb N}$ of $G$. M
We shortly write $h^{\otimes 2}=h\otimes h$.
Finally, for $L^2(I)$-valued processes $X^n, n\in \mathbb N$, $X$, we write $X^n\stackrel{u.c.p.}{\longrightarrow}{X}\quad \text{ as }n\to\infty$ for the convergence uniformly on compacts in probability, i.e. it is $\mathbb P[\sup_{t\in [0,T]} \|X^n(t)-X(t)\|>\epsilon]\to 0$ for all $\epsilon, T>0$.

\section{General Arbitrgae-Free Bond Market-Dynamics}\label{Sec: bond market primer}
Let $(\Omega,\mathcal F,(\mathcal F_t)_{t\geq 0},\mathbb P)$ be a filtered probability space with right-continuous filtration.
From here on,  we assume that the forward rate process $(f_t)_{t\geq 0}$ is an $L^2(\mathbb R_+)$-valued stochastic process that is the mild solution to the stochastic partial differential equation \eqref{SPDE}, defined on $(\Omega,\mathcal F,(\mathcal F_t)_{t\geq 0},\mathbb P)$. 
That is, \begin{align}\label{mild Ito process}
f_t %=&\mathcal S(t)f_0+\int_0^t \mathcal S(t-s)dX_s\\
= &\mathcal S(t) f_0+\int_0^t \mathcal S(t-s)\alpha_s ds+\int_0^t \mathcal S(t-s)\sigma_sdW_s+\int_0^t\mathcal S(t-s)d J_t.
\end{align}
%As a solution to such an equation exists, the coefficients $\alpha, \sigma$ and $\gamma$ can be dependent on $f$ themselves, allowing for many nonlinear dependence patterns in the dynamics.  A comprehensive discussion on conditions ensuring the existence of such a mild solution in this case is given in \cite{FilipovicTappeTeichmann2010}.
where $ \mathcal S(t)f(x)=f(x+t)$ for $t\geq 0 $ and $f\in L^2(\mathbb R_+)$ defines the left-shift operator semigroup.
%, $\alpha$ is an almost surely integrable adapted process, $\sigma$ is the Hilbert-Schmidt operator-valued volatility process, $W$ is a cylindrical Wiener process with identity covariance and $J$ is a jump process.
 We relegate all further technical discussions on $X$,  $\alpha$,  $\sigma$, $W$ and $J$ and
 various related technical assumptions that we need in for the validity of our limit theory to  Sections \ref{Sec:  semimartingales in Hilbert spaces}  in the appendix.  We remark, that under arbitrage-free dynamics, that is, under an equivalent local martingale measure  the drift is necessarily a deterministic function of $\sigma$  and $\gamma$ (c.f. \cite{bjork1997}), %that is
%\begin{equation}\label{HJM drift condition}
%\alpha_t(x)=\alpha_{HJM}(\sigma_t,\gamma_t)(x)= \sum_{j=1}^{\infty} \sigma_t e_j (x)\int_0^x\sigma_t e_j (y)dy-\int_{H\setminus\{0\}} \gamma_s(z)\left(e^{-\int_0^x \gamma_s(z)(y)dy}-1\right)F(dz)
%\end{equation}
which was in the continuous case the original inside leading to the popular Heath-Jarrow-Morton framework of \cite{HJMoriginal} for pricing bonds and interest rate sensitive contingent claims. However, this will not be of particular importance for our purposes, as the drift later vanishes asymptotically in our limit theory. We will discuss however some important examples subsequently. Before that,  we make a remark on the choice $L^2(\mathbb R_+)$ as the state space of forward rate curves.

\begin{remark}[On the forward curve space]
There are other choices for the state space of $f$ than $L^2(\mathbb R_+)$ such as the forward curve space of \cite{Filipovic2000}. 
We choose, however, to work in an $L^2$-setting because in that way we do not impose further regularity assumptions on the forward curves.
   A supposed restriction of the state space $L^2(\mathbb R_+)$  is that the so-called long-rates $\lim_{x\to \infty}f_t(x)$ are equal to $0$.   This is undesirable from a financial point of view and could easily be fixed in several ways. For instance, we could consider the  Hilbert space 
$H:=\mathbb R \oplus L^2(\mathbb R_+)=\{f:\mathbb R_+\to \mathbb R: f(x)=a+h(x), a\in\mathbb R, h\in L^2(\mathbb R_+) \}$, for which the first component models the long-rate. In this case, the forward curve space of \cite{Filipovic2000} would be contained as a subspace.
We then might just assume that the state spaces of $X^C$ and $J$  belong to $\{0\}\times L^2(\mathbb R_+)\equiv L^2(\mathbb R_+)$, which 
is in line with the assumptions on the volatilities in \cite{FilipovicTappeTeichmann2010b}. As $[X,X]$, $[X^C,X^C]$ and $[J,J]$ do not depend on the drift and the initial condition, the respective limit theory  would be exactly the same. Another reason that justifies our choice is that in practice, our asymptotic analysis just
takes into account bond price data $P_t(x)$ with $(t,x)\in [0,T]\times [0,M]\subset \mathbb R_+^2$ for some maximal time to maturity $M<\infty$ and the behavior of the forward curves for $x\to \infty$ is of minor importance.
To relax the  notation, we stick to the state space $L^2(\mathbb R_+)$ without loss of generality.
\end{remark}

Let us now discuss some important simple examples.
\begin{example}\label{Ex: CPP in Hilbert space}[Sum of $Q$-Wiener and compound Poisson process]
As a simple example assume $\alpha_t=\alpha \in L^2(\mathbb R_+)$ and $\sigma_t=\sigma \in L_{\text{HS}}(\mathbb R_+)$ to be constant. Then, $X^C_t$ is a Gaussian random variable in $L^2(\mathbb R_+)$ with mean $t\alpha$ and covariance $tQ:=t\sigma\sigma^*$, where $\sigma^*$ is the Hilbert space adjoint of $\sigma$. equivalently,  $X_t^C$ has covariance kernel $q$ given by $\int_{\mathbb R_+} q(x,y)f(y) dy=(Q f)(x)$ for $x\geq 0$. Since we want use the jump process to model rare outliers,  a reasonale model would be a compound Poisson process 
%A very simple example for the driving semimartingale is the sum of an $L^2(\mathbb R_+)$-valued Wiener process and an $L^2(\mathbb R_+)$-valued Poisson random measure
%$$X_t= at+W^Q_t+ J_t.$$
%The drift $A_t=a t$ is linear for $a\in H$ and Here, $a\in H$ and the $Q$-Wiener process $W_t\sim N(0,tQ)$ for a covariance operator $Q$ in $H$ has independent stationary increments and forms a continuous martingale.  If $Q^{\frac 12}$ denotes the positive  square root of $Q$, we have in distribution that $W^Q_{\cdot}=\int_0^{\cdot} Q^{\frac 12} dW_t$ for a cylindrical Wiener process $W$ and we have $M_t^C=W_t^Q$ in this case. 
%The jumps correspond to a compound Poisson processes
    $$J_t:=\sum_{i=1}^{N_t} \chi_i,$$
    for an i.i.d. sequence $(\chi_i)_{i\in \mathbb N}$ of random variables in $L^2(\mathbb R_+)$ with law $F$ and finite second moment ($\mathbb E[\|\chi_i\|^2]<\infty$) and a Poisson process $N$ with intensity $\lambda>0$. 
     Since in the technical Section \ref{Sec:  semimartingales in Hilbert spaces} we require the jump process to be a martingale,  this is not immediately a valid choice, but we can rewrite the dynamics accordingly (c.f. Example \ref{Ex: Compensate CPP} in the appendix).
 The quadratic covariation of this semimartingale is then
  $$[X,X]_t=[X^C,X^C]_t+[J,J]_t= t Q + \sum_{i=1}^{N_t} \chi^{\otimes 2}_i.$$
     \end{example}

The term structure setting \eqref{SPDE} contains the vast majority of existing arbitrage-free term structure models considered in the literature. Among them is the class of affine term structure models, which are widely  appreciated for their parsimony. 

\begin{example}[Affine term structure models]\label{Ex: Affine Term Structure Models}
    In an affine term structure model, the state space for forward curves is spanned by a finite amount of factors, such that
    \begin{align}\label{affine forward curves}
        f_t(x)=g_0(x)+g_1(x) x_t^1+...+g_d(x)x_t^d,
    \end{align}
    where $g:\mathbb R_+\to\mathbb R$ are some particularly suitable functions and the process $x=(x^1,...,x^d)$ is an affine process (c.f. \cite{Duffie2003}). To guarantee that a factor structure like \eqref{affine forward curves} can be in line with 
    general no arbitrage dynamics of the form \eqref{SPDE} one has to impose restrictions on both the functions $g_0,g_1,...,g_d$ as well as the multivariate semimartingale $x$ (c.f. e.g. \cite{Filipovic2000}, Section 7.4 for a description in a continuous setting). 
   % Though the restrictions on the factors of this models rule out many possible dynamical structures, they are widely appreciated for their parsimony. 
   %One can show (c.f. Example ... below) that at least for continuous affine term structure models the quadratic variation of the forward curve process $f$ and the latent driving semimartingale $X$ must coincide, and 
   If $[x^c,x^c]$ and $[x^d,x^d]$ denote the continuous and discontinuous part of the multivariate quadratic variation of $x$, it is
    $$[X^C,X^C]_tt= \sum_{i=1}^d[x_i^c,x_j^c]_t g_i\otimes g_j,\qquad [J,J]_t= \sum_{i=1}^d[x_i^d,x_j^d]_t g_i\otimes g_j.$$
If the respective bond market data are in line with an affine model such as \eqref{affine forward curves}, our theory in Section \ref{Sec: term structure models theory} identifies this structure asymptotically.
  
   Let us outline two classical special cases when $d=1$ and when there are no jumps:
    \begin{itemize}
       \item[(a)] (Va{\v s}i{\v c}ek model) In the Va{\v s}i{\v c}ek model it is 
        $$dx^1_t= (b-ax_t^1)dt+\sigma_0 d\beta_t$$
        for $b,a,\sigma_0> 0$ and a one-dimensional Brownian motion $\beta$. The functions $g_0,g_1$ are then given by $g_1(x)=e^{-ax}$ and $g_0(x)=b\int_0^x g_1(y)dy-(\sigma_0^2/2)(\int_0^x g_1(y)dy)^2$ (c.f. Section 7.4.2 in \cite{Filipovic2000}).  %This yields the arbitrage-free forward curve dynamics
%\begin{align*}
% f_t(\cdot)= &f_0(t+\cdot)+\sigma_0^2\int_0^te^{-a (\cdot+t-s)}\left(\int_0^{\cdot+t-s}e^{-ay}dy\right)ds\\
% & \qquad\qquad+\sigma_0\int_0^t e^{-a(\cdot+t-s)}d\beta_s,
%\end{align*}
%with $f_0(\cdot)=b(\int_0^{\cdot} e^{-ay}dy)-(\sigma_0^2/2)(\int_0^{\cdot} e^{-ay}dy)^2+x^1_0 e^{-a \cdot} $,
%which defines arbitrage-free bond-market dynamics. 
In this case, the quadratic variation of the latent driving semimartingale $X$ is
$$ [X,X]_t=[f,f]_t= t \sigma_0^2(e^{-a \cdot})^{\otimes 2},\qquad t\geq 0.$$
With $Q=\sigma_0^2(e^{-a \cdot})^{\otimes 2}$, this is a special case of Example \ref{Ex: CPP in Hilbert space} without jumps.
\item[(b)] (CIR model) In the CIR model we have short rate dynamics of the form $$dx_t^1=(b-ar_t)dt+\sigma_0 \sqrt {x^1_t}d\beta_t $$
for $b, a, \sigma_0 \geq 0$ and a standard univariate Brownian motion $\beta$. The function $g_1$
is given as the derivative of the function $x\mapsto G_1(x):=2 (e^{cx}-1)/((c-a)(e^{cx}-1)+2c)$ with $c=\sqrt{a^2-2\sigma_0}$ and the function $g_0$ is given as $g_0=aG_1$ (c.f. Section 7.4.1) in \cite{Filipovic2000}).
%The implied risk-neutral forward curve dynamics are
%\begin{align*}
%f_t(\cdot)=&f_0(t+\cdot)+\sigma_0\int_0^t \left(g_1(\cdot+t-s)\int_0^{\cdot+t-s} g_1(y)dy\right)x_s^1  ds \\
%&\qquad\qquad+\sqrt{\sigma_0}\int_0^t g_1(\cdot+t-s) \sqrt{x^1_t} d\beta_s
%\end{align*}
%with $f_0(x)=b\int_0^x g(y)dy+g(x)x_0^1 $ (c.f. Section 7.4.1) in \cite{Filipovic2000}). 
In this case the quadratic variation of the latent driving semimartingale $X$ is
$$[X,X]_t= [f,f]_t=\sigma_0^2\left(\int_0^t x^1_t ds\right) g_1^{\otimes 2},\qquad t\geq 0.$$
\end{itemize}
    %We will not discuss the precise conditions imposed on the factors, here, but want to argue, that under risk neutral dynamics, the quadratic variation of the affine process $x$ is given by
    %$$([x,x]_t)_{i,j}=\rho_{i,j}\int_0^t \sqrt{s_{0,i}+s_{1,i}x^i_s}\sqrt{s_{0,j}+s_{1,j}x^j_s} ds\quad i,j=1,...,d,$$
%for some suitably chosen coefficients $\rho_{i,j}$. The conditions imposed on the dynamics of $x$ and the functions $g_0,...,d_d$ to guarantee the convenient affine form \eqref{affine forward curves} are quite restrictive (c.f. \cite{Filipovic2000}, Section 7.4 for a detailed description). However, these models are quite parsimonious, flexible in the sense that we can add arbitrarily many factors in principle and they are quite convenient to handle when it comes to pricing. 
\end{example}

Many simple ways to model  term structures lead to nonaffine dynamics, such as 
\begin{example}[Volterra spot rate models]\label{Ex: VMVP for forward curves}
For many term structure models the quadratic variation of  $f$ is not necessarily well-defined, as it must not be a semimartingale. For instance, take the forward rate dynamics of the form
\begin{align}\label{Volterra form of forward curves}
    f_t=f_0+\int_0^t \alpha_s(\cdot+t-s) ds +\sum_{i=1}^{d}\int_0^t k_i(\cdot+t-s) \sigma_s^i d\beta_s^i
\end{align}
for a multivariate standard Brownian motion $(\beta^1,...,\beta^d) $ for some $d\in \mathbb N$,  a $d$-dimensional volatility $(\sigma^1,...,\sigma^d)$  and deterministic kernels  $k^1,...,k^d\in  L^2(\mathbb R_+)$ as well as  a drift $\alpha$ which satisfies the HJM condition (c.f. \cite{Filipovic2000}). %and $N$ the random measure described in Example \ref{Ex: CPP in Hilbert space} corresponding to a compound Poisson process $$J_t= \sum_{i=1}^{N_t} \chi_i$$ in $L^2(\mathbb R_+)$.
In this scenario, the underlying driving semimartingale has quadratic variation equalling 
$$[X,X]_t=  \sum_{i=1}^{d}  \left(\int_0^t (\sigma_s^i)^2ds\right)k_i^{\otimes 2}.$$
Thus, if $\sigma_s^i= 1$ for all $s\geq 0$ and $i\in \mathbb N$, this is a special case of Example \ref{Ex: CPP in Hilbert space} with $Q=\sum_{i=1}^{d}k_i^{\otimes 2}$ and without jumps, but it does not always correspond to an affine term structure.
 In the energy market, for instance, fractional kernels such that $k_1(t)=\mathcal O(t^H)$ for $t\to 0$ and $k_i\equiv 0$ for $i\geq 2$ are used to model energy spot prices (c.f. \cite{BENNEDSEN2017} or \cite{BNBV2013}).  If we assume that $k_1(t)=t^H$, $\sigma^1\equiv 1$ and $\alpha_s\equiv 0$ for $t\in [0,1]$ one can prove that
 $f_t$ is not a semimartingale in $L^2(\mathbb R_+)$ and the quadratic variation does not converge as we show in  Appendix \ref{Sec: Proof of the nonsemimartingality of volterra spot models}.  
\end{example}

Example \ref{Volterra form of forward curves} (and also Example 3.16 in \cite{BSV2022}) shows that the process $(f_t)_{t\geq 0}$ is in general not an $L^2(\mathbb R_+)$-valued semimartingale.  However, all implied bond prices $(P_t(T-t))_{0\leq t\leq T}$ are semimartingales for all $T>0$, which is necessary to guarantee the absence of arbitrage in the bond market. %Indeed formula \eqref{formal dynamics of the forward rate process} guarantees that all bond prices $(P(t,T))_{t\in [0,T]}$  given by $P(t,T)=\exp(-\int_0^{T-t} f_t(x)dx)$ are real-valued semimartingales for all $T> 0$.
%A proof can be found in the appendix.
Moreover, while the quadratic covariations of $f$ must not be convergent,  we show in the next section, in which we present our main results,  that the realized covariation of difference returns measures quadratic covariations of the latent driver asymptotically and without further conditions.

%\cite{Lavagnini2021} 
%or for variance swaps in \cite{Buehler2006}.  

%Before we move to the asymptotic theory,  we need to

%Example \ref{Volterra form of forward curves} makes clear that even in simple settings with only a univariate driving random source (the Brownian motion $\beta^1$, if all $k_i\equiv 0 $ for $i\geq 0$),
%we have to develop an asymptotic theory that goes beyond the quadratic variation of finite-dimensional semimartingales. This is the subject of the subsequent section.

\section{Estimation of Quadratic Covariations}\label{Sec: term structure models theory}
In this section we present our asymptotic theory for estimation of quadratic variations. We start with the identifiability of $[X,X]$, $[X^C,X^C]$ and $[J,J]$.
\subsection{Identification of the quadratic covariation of the latent semimartingale}\label{Sec: Identification of the quadratic variaiton}
 %which is, hereafter, denoted 
%\begin{equation}\label{difference returns as vector entries}
%    d^n(i,j):=\tilde{\Delta}_{i\Delta_n}^n d(j\Delta_n)
%\end{equation}
%This is a realistic sampling scheme, if $\Delta_n$ corresponds to one day for instance and we rely on daily nonparametrically smoothed yield curves or discount curves, such as provided in \cite{FPY2022}. Measuring time in years (and having roughly 250 trading days each year, this correpsonds to $\Delta_n\approx 250^{-1}$ in practice. 
We rely on infill asymptotics
$\Delta_n\to 0$ as $n\to \infty$
%For theoretical reasons we consider to arbitrarily large maturities, so we  have for each $i\in \mathbb N$ a sequence $$d(i,\cdot)\in l^2$$ with entries $d(i,j)$ given by \eqref{difference returns as vector entries}.
and recall the definition of the realized covariation $(\hat q_t^n)_{t\geq 0}$ as a piecewise constant kernel. That is, 
%\begin{align*}
%\hat{\Sigma}_t^{n}:= & \sum_{i=1}^{\ul}\left(\Pi_n \tilde{\Delta}_{i\Delta_n}d\right)^{\otimes 2}\\
%= & \frac 1{\Delta_n^2}\sum_{i=1}^{\ul}\sum_{j_1,j_2}^{\lfloor M/\Delta_n\rfloor}d^n(i,j_1)d^n(i,j_2)\indicator_{[(j_1-1)\Delta_n,j_1\Delta_n]}\otimes \indicator_{[(j_2-1)\Delta_n,j_2\Delta_n]}
%\end{align*}
%This can be uniquely identified with the matrix $RV_t^{250}$ introduced in section ... and can be recovered by data easily. 
for $x\in [(j_1-1)\Delta_n,j_1\Delta_n] $ and $y\in [(j_2-1)\Delta_n,j_2\Delta_n]$, $j_1,j_2\in\mathbb N$ and $t\geq 0$ we define
\begin{equation}\label{discrete quadratic variation estimator}
    \hat q _t^n(x,y):= \sum_{i=1}^{\ul} d_i^n(j_1)d_i^n(j_2).
\end{equation}
For each $n\in \mathbb N$ and $t\geq 0$, we have that $\hat q_t^n\in L^2(\mathbb R_+^2)$, which follows from the Assumption that forward curves are elements in $L^2(\mathbb R_+)$ (c.f. Remark \ref{rem: semigroup adjustments are difference returns} below). 
We now state the general identifiability result for the quadratic covariation of $X$.
\begin{theorem}\label{T: General discrete LLN}
    It is as $n\to \infty$ and w.r.t. the Hilbert-Schmidt norm and $\mathcal T_{\hat q^n}$ as in  \eqref{Hilbert-Schmidt kernel equivalence}
    \begin{equation}\label{eq: abstract covariation convergence}
 \Delta_n^{-2}\mathcal   T_{\hat q^n}\overset{u.c.p.}{\longrightarrow}  [X,X].
    \end{equation}
   % and equivalently
  %  $$\hat{\Sigma}_t^n\overset{u.c.p.}{\longrightarrow}\int_0^t \Sigma_s ds+\left( Y_s-Y_{s-}\right)^{\otimes 2} .$$
\end{theorem}
%    In particular, we obtain w.r.t. the norm on $L^2([0,M]^2)$ the feasible result
 %   \begin{equation}
%      \hat q^n\big|_{[0,M]^2}\overset{u.c.p.}{\longrightarrow} \mathfrak q\big|_{[0,M]^2}.
 %   \end{equation}
As in the case of finite-dimensional semimartingales, we do not have to impose any further conditions to identify the quadratic covariations of the driving semimartingale $X$ although the observable process $(f_t)_{t\geq 0}$ is not necessarily an $L^2(\mathbb R_+)$-valued semimartingale.  This is due to the relation of difference returns to semigroup-adjusted forward rate returns, which where shown in \cite{Schroers2024}, \cite{Benth2022} and \cite{BSV2022} to be well-suited for volatility estimation for processes of the form \eqref{mild Ito process}. This relationship is made clear in the next remark.
\begin{remark}\label{rem: semigroup adjustments are difference returns}[Difference returns are discretized semigroup-adjusted increments]
 % At first glance, it is remarkable that economically derived difference returns lead to such a general identifiability result for volatility estimation in the general bond market framework. 
 The reason for economically motivated difference returns to lead to such a general identifiability result is that difference returns %and the estimator $\hat q_t^n$ %, $\hat q^{n,-}_t$ and $\hat q_t^n(+)$ (almost)
  coincide with orthonormal projections onto semigroup-adjusted increments of forward rate curves. That is, we have
  $$%d^n(i,j)=
 d_i^n(j)=-\langle \tilde \Delta_i^n f,\indicator_{[(j-1)\Delta_n,j\Delta_n]}\rangle.$$
  where  $\tilde \Delta_i^n f$ denotes the semigroup-adjusted forward rate increment
$$\tilde \Delta_i^n f:= f_{i\Delta_n}-\mathcal S(\Delta_n) f_{(i-1)\Delta_n}\quad i=1,...,\ulT.$$
and $(\mathcal S(t))_{t\geq 0}$ is the left shift semigroup on $L^2(\mathbb R_+)$.
%   $(\mathcal S(t))_{t\geq 0}$ denotes the
  %  nilpotent semigroup of left shifts in $L^2(0,H)$ given by 
%\begin{equation}\label{Nilpotent Shift}
%    \mathcal S(t)h(x):=\begin{cases} h(x+t), & x+t\leq H,\\
%0, & x+t>H,
%\end{cases}
%\end{equation}
%  To see that, recall that the price $P_t(x)$ of a bond at time $t$ maturing at time $t+x$ is linked to instantaneous forward rates by the formula $P_t(x)=exp(-\int_0^x f_t(y)dy)$ and, thus,  by substitution and formula \ref{Difference returns}
 %   \begin{align*}
 %       \tilde \Delta_{t}^n d(x) % = & (\log P_{t}(x)-\log d_{t}(x+\Delta_n))-( \log d_{t+\Delta_n}(x-\Delta_n)-\log d_{t+\Delta_n}(x))
        %= & \left(-\int_0^x f_t(y)dy+\int_0^{x+\Delta_n} f_t(y)dy\right)-\left(-\int_0^{x-\Delta_n} f_{t+\Delta_n}(y)dy+\int_0^x f_{t+\Delta_n}(y)dy\right)\\
       % = & \int_x^{x+\Delta_n} f_t(y)dy- \int_{x-\Delta_n}^x f_{t+\Delta_n}(y)dy\\
        %= & -\int_{x-\Delta_n}^x f_{t+\Delta_n}(y)-f_t(y+\Delta_n)dy\\
 %       = -\langle f_{t+\Delta_n}-f_t(\cdot+\Delta_n),\indicator_{[x-\Delta_n,x]}\rangle.
 %   \end{align*}
%In particular, we have %$\Pi_{n,M} \tilde{\Delta}_{i\Delta_n}d= \Pi_{n,M} \tilde \Delta_i^n f$ for w
%with
Define
\begin{equation}\label{Orthonomal Projection on Indicators}
    \Pi_{n,M} h:=  n \sum_{j=1}^{\lfloor M/\Delta_n\rfloor}\langle h, \indicator_{[(j-1)\Delta_n,j\Delta_n]}\rangle \indicator_{[(j-1)\Delta_n,j\Delta_n]}
\end{equation}
the projection onto $span(\indicator_{[(j-1)\Delta_n,j\Delta_n]}:j=1,...,\lfloor M/\Delta_n\rfloor)$ and observe that %for the integral operator %$T_{\hat q^{n,M}_t}$ corresponding to $\hat q_t^n\big|_{[0,M]^2}$ given by $
%T_{\hat q^{n,M}_t}:=  \sum_{i=1}^{\ul}\left(\Pi_{n,M}\tilde{\Delta}_{i\Delta_n}d\right)^{\otimes 2}$
%=  \frac 1{\Delta_n^2}\sum_{i=1}^{\ul}\sum_{j_1,j_2}^{\lfloor M/\Delta_n\rfloor}\tilde \Delta_{i\Delta} d(j_1\Delta)\tilde \Delta_{i\Delta} d(j_2\Delta)
%d^n(i,j_1)d^n(i,j_2)\indicator_{[(j_1-1)\Delta_n,j_1\Delta_n]}\otimes \indicator_{[(j_2-1)\Delta_n,j_2\Delta_n]}
\begin{align*}
 \Delta_n^{-2} \mathcal   T_{\hat q^{n,M}_t}= 
  %&\frac 1{\Delta_n^2}\sum_{i=1}^{\ul}\sum_{j_1,j_2}^{\lfloor M/\Delta_n\rfloor}d_i^n(j_1)d_i^n(j_2)\indicator_{[(j_1-1)\Delta_n,j_1\Delta_n]}\otimes \indicator_{[(j_2-1)\Delta_n,j_2\Delta_n]}
%&\sum_{}\Pi_{n,M}\left(\sum_{i=1}^{\ul} \tilde \Delta_i^n f^{\otimes 2}\right)\Pi_{n,M}\\
  \Pi_{n,M}(SARCV_t^n)\Pi_{n,M},
\end{align*}
where $SARCV_t^n=\sum_{i=1}^{\ul} \tilde \Delta_i^n f^{\otimes 2}$ is the semigroup-adjusted realized covariation, which was shown to be a consistent estimator of 
$[X,X]_t$ in \cite{Schroers2024} in the presence of jumps
and a consistent and asymptotically normal  estimator of $[X^C,X^C]$ in \cite{Benth2022} and \cite{BSV2022} when $J\equiv 0$.  This characterization of the realized covariation $\hat q^n$ also explains the appearance of the scalar $\Delta_n^{-2}$ in front of the covariation in \eqref{eq: abstract covariation convergence}.
\end{remark}

%\begin{remark}
% The sampling scheme in \eqref{difference returns as vector entries} indicates that we need the same 
%   daily resolution in the maturity dimension as in time, which cannot be provided by raw bond market data. On the one hand, this is actually not strictly necessary as long as we can reconstruct the difference returns at several points $x_1,...,x_m$ and assuming $\sup_{j=1,...,m} |x_{j+1}-x_j|\to 0$ as $m\to \infty$ at a potentially different rate than $n$.
%   On the other hand, the functional nature of forward and yield curves allows their approximation in various ways such as splines, kernel smoothing or basis expansions and formally justifies the use of nonparametrically smoothed yield or discount curves, which provide the curves in daily resolution.
%We discuss the aspect of smoothing, the use of nonparametric yield curve data as well as further transformations of the data as functional data objects from a practical point of view in Section \ref{Sec: Practical Considerations}. Our simulation study in Section \ref{Sec: Simulation Study} affirms that if the volatility is sufficiently regular, the sparse sampling in the maturity dimension and subsequent smoothing still leads to good results.
%\end{remark}

Next, we examine how to identify the continuous part of the quadratic covariation.

\subsection{Identification of $[X^C,X^C]$ and $[J,J]$ via truncated covariation estimators}\label{Subsec: Truncated Estimation in term structure models}
We now turn to the estimation  of the continuous part of the quadratic covariation. We will derive jump robust estimators by a truncated form of $\hat q_t^{n}$ defined 
by the piecewise constant kernel
$\hat q^{n,-}_t$ given for $x\in [(j_1-1)\Delta_n,j_1\Delta_n]$, $y\in [(j_2-1)\Delta_n, j_2\Delta_n]$,  $j_1,j_2\in \mathbb N$ and $t\geq 0$ by
\begin{align}\label{truncated estimator}
   \hat q^{n,-}_t(x,y):= 
   &\sum_{i=1}^{\lfloor t/\Delta_n\rfloor}d_i^n(j_1)d_i^n(j_2)\indicator_{ g_{n}\left(d_i^n/\Delta_n\right)\leq u_n}
\end{align}
for $u_n=\alpha \Delta_n^{w}$, with $w\in (0,1/2)$ and $\alpha>0$
and a particular sequence of truncation functions $g_{n}$ that takes into account only the discrete data $d^i_n(j)$ for $j\in \mathbb N$. 
 %That is, we now consider the truncated estimator $\hat q_t^n(-)$  given  by formula \eqref{truncated estimator} for $x\in [(j_1-1)\Delta_n,j_1\Delta_n]$, $y\in [(j_2-1)\Delta_n, j_2\Delta_n]$ and $j_1,j_2\in \mathbb N$ by formula \eqref{truncated estimator}.
%\begin{align*}
%   \hat{q}_t^n(-) (x,y):= 
%   &\frac 1{\Delta_n^2}\sum_{i=1}^{\ul}\tilde \Delta_{i\Delta} d(j_1\Delta_n)\tilde \Delta_{i\Delta_n} d(j_2\Delta_n) \indicator_{ g_n\left((\tilde \Delta_{i\Delta_n} d(j\Delta_n)/\Delta_n)_{j\geq 0}
   %d^n(i,\cdot)
 %  \right)\leq u_n}
%\end{align*}
%for a sequence 
%\begin{equation}\label{def of un}
%    u_n=\alpha \Delta_n^{\omega},\quad \omega\in (0,1/2), \alpha > 0
%\end{equation}
Precisely, the corresponding sequence of truncation functions $g_n: l^2\to \mathbb R_+$ must satisfy for constants $c,C>0$ and for all $f,h\in l^2$ and all $n\in \mathbb N$
\begin{align}\label{def: g}
  c\|f\|_{l^2}\leq g_n(f)\leq C \|f\|_{l^2},\quad \text{ and }\quad
  g_n(f+h)\leq g_n(f)+g_n(h).
\end{align}
%The choice of the truncation level and function require some further explanantion. 
While the particular choice of the functions $g_n$ will not play a role for the asymptotic behavior of $\hat q^{n,-}_t$, it is important to modify it in practice. For the moment, one can take in mind the legitimate choice $g_n=\|\cdot\|_{l^2}$ for all $n\in \mathbb N$ for which we have  with $\Pi_{n,\infty}$ defined as in \eqref{Orthonomal Projection on Indicators}  for $M=\infty$ that
$\|d_i^n/\Delta_n\|_{l^2}= \|
    \Pi_{n,\infty} \tilde \Delta_i^n f\|_{L^2(\mathbb R_+)}.$
   % The intuition behind this estimator is that an increment of a Brownian motion is of magnitude $\sqrt{\Delta_n}\log(1/\Delta_n)>u_n$ for large $n$, due to Levy's modulus of continuity theorem. Thus, the increments corresponding to jumps are either truncated in the estimator, or have no impact on our limit result. 
   We will discuss a data-driven specification of $g_n$ and the truncation level in Section \ref{Sec: Practical Considerations}. 
   
  The next result states that $\hat q^{n,-}$ consistently estimates the quadratic covariation of $X^C$.

%The next theorem shows that the integrated volatility $
%\int_0^t \Sigma_s ds$ is identifiable in the context of infill asymptotics in maturity and time.
%It should be underlined that as in the limit theory for finite-dimensional semimartingales no further assumption than Assumption \ref{As: H}(2) is imposed.
 
%To reduce notation, we say without loss of generality that
%\begin{equation}\label{reduced domain kernel convergence notation}
%    q^{n,-}\overset{u.c.p.}{\longrightarrow} \int_0^{\cdot}q_s^C ds,\quad \text{ in }L^2([0,M]^2),\quad \text{ when } \quad q^{n,M,-}\overset{u.c.p.}{\longrightarrow} \int_0^{\cdot}q_s^C\big|_{[0,M]^2}ds,\quad \text{ in }L^2([0,M]^2)
%\end{equation}
\begin{theorem}\label{T: General Limit discretized truncated LLN}
   Under Assumption \ref{As: H}(2) and  with the notation of \eqref{Hilbert-Schmidt kernel equivalence} it is  as $n\to\infty$ 
   $$\Delta_n^{-2}\mathcal T_{\hat q^{n,-}_{\cdot}} \overset{u.c.p.}{\longrightarrow}[X^C,X^C]. $$
  % where $\Pi_M f (x)= \indicator_{[0,M]}(x)f(x)$ for all $f\in L^2(\mathbb R_+)$.
%Moreover, it is
%   as $n\to\infty$ and for all $M>0$ and w.r.t. the norm on $L^2([0,M]^2)$
 %  $$\hat q^{n,M,-}\overset{u.c.p.}{\longrightarrow} \int_0^{\cdot}q_s^C\big|_{[0,M]^2}ds.$$
\end{theorem}

Let us make a  remark on the feasibility of the estimator.
\begin{remark}\label{Rem: Reduction of the time to maturity is fine}
In practice, we do not observe the $d_i^n(j)$ for all $j\in \mathbb N$ but rather up to a finite maturity $M$, that is ,for $j\in 1,...,\lfloor M/\Delta_n\rfloor-1$.  %In that regard, the truncated realized variation $\hat q^{n,-}$ is not a feasible estimator, so we need to consider a feasible version instead.
    Consistency of $\hat q^{n}$ from Theorem \ref{T: General discrete LLN} implies the consistency of $\hat q^{n}\big|_{[0,M]^2}$ for each $M>0$,  so there is no problem when we do not consider truncation.
    However,  $\hat q^{n,-}\big|_{[0,M]^2}$ is not a feasible estimator in this context, since it uses in the truncation function the whole infinitely long vector $(d^i_n(j))_{j\in\mathbb N}$.
A feasible estimator $\hat q^{n,M,-}$ is defined for $x\in [(j_1-1)\Delta_n,j_1\Delta_n]$, $y\in [(j_2-1)\Delta_n, j_2\Delta_n]$, $j_1,j_2\in \mathbb N$ and $t\geq 0$ by
\begin{align}\label{Spatially maximum truncated difference return estimator}
  \Delta_n^{-2}\hat q^{n,M,-}_t (x,y):= 
   & \Delta_n^{-2}\sum_{i=1}^{\ul}d_i^n(j_1)d_i^n(j_2)\indicator_{ g_n\left((
   %\delta_{j\Delta_n \leq \lfloor M/\Delta_n\rfloor}
   \indicator_{[0,M]}(j\Delta_n)
   d_i^n(j)/\Delta_n)_{j\geq 0}\right)\leq u_n}.
\end{align}
%where $\delta_{x\leq \lfloor M/\Delta_n\rfloor}=1$ if $x\leq \lfloor M/\Delta_n\rfloor$ and $\delta_{x\leq \lfloor M/\Delta_n\rfloor}=0$ otherwise. 

To relax the notation, we do not present the limit theorems for $\hat q^{n,M,-}$ in this Section. % or respectively $\mathcal T_{\hat q^{n,M,-}}$ which is a consistent estimator of $ [ \Pi_MX^C, \Pi_MX^C]$
  However, all results that we state for $\mathcal T_{\hat q^{n,-}}$ with limit $[X^C,X^C]$, that is, Theorems \ref{T: General Limit discretized truncated LLN}, \ref{T: Rate of convergence for discretized estimator}. \ref{T: CLT for truncated estimator}  hold for $\mathcal T_{\hat q^{n,-}}$ with limit $[\Pi_MX^C,\Pi_MX^C]$ and Theorem \ref{T: Long-time asymptotics for termstructure volatiltiy} holds for $\mathcal T_{\hat q^{n,-}_T}/T$ with limit $\Pi_M \mathcal C\Pi_M$  where $\Pi_M h(x)= \indicator_{[0,M]}(x)h(x)$ for all $h\in L^2(\mathbb R_+)$.  The formal proof for that can be found in the appendix. %supplementary material to this article within the proofs of the respective Theorems.
\end{remark}
Observe that we can also define an upward truncated estimator
    $ \Delta_n^{-2}\hat q^{n,+}$ given for $t\geq 0$, $x\in [(j_1-1)\Delta_n,j_1\Delta_n]$, $y\in [(j_2-1)\Delta_n, j_2\Delta_n]$ and $j_1,j_2\in \mathbb N$ by
\begin{align*}
   \hat q^{n,+}_t(x,y):= 
   & \sum_{i=1}^{\ul}d_i^n(j_1)d_i^n(j_2)\indicator_{ g_n(d_i^n/\Delta_n)> u_n}.
\end{align*}
%$$\hat{\Sigma}_t^n(+):=\sum_{i=1}^{\ul}\left(\Pi_n \tilde{\Delta}_{i\Delta_n}d\right)^{\otimes 2}\indicator_{g(\Pi_n \tilde{\Delta}_{i\Delta_n}d)> u_n}.$$
Obviously, 
$\hat q^n=\hat q^{n,-}+\hat q^{n,+}  $  and $\mathcal T_{\hat q^n}=\mathcal T_{\hat q^{n,-}}+\mathcal T_{\hat q^{n,+}  }$. %$\hat{\Sigma}_t^n=\hat{\Sigma}_t^n(-)+\hat{\Sigma}_t^n(+)$. 
Then, combining Theorem \ref{T: General Limit discretized truncated LLN} and Theorem \ref{T: General discrete LLN}, we also obtain 
\begin{corollary}
If  Assumption \ref{As: H}(2) holds, we have as $n\to\infty$ that
  $$ \Delta_n^{-2}\mathcal T_{\hat q^{n,+}_{\cdot}}\overset{u.c.p.}{\longrightarrow} [J,J].$$
\end{corollary}
This result shows that the quadratic covariation corresponding to the jump part is identifiable in the context of general bond market models.
However, the finer analysis of jumps is not part of this paper and relegated to future work.  Instead, 
we derive convergence rates for the estimation of the continuous part of the quadratic covariation in the next section.

\subsection{Convergence rates and central limit theorem for estimation for $\hat q^{n,-}$}\label{Sec: Error Bounds}

In order to derive rates of convergence and a central limit theorem for estimating the continuous part of the quadratic variation, we need to impose further regularity Assumptions, which depend on the smoothness of the kernel corresponding to the operators $[X^C,X^C]_t$.  For the error bounds,  this is Assumption \ref{As: spatial regularity}, which is discussed in detail in Section \ref{Sec: Assumptions for Error bounds}.  We discuss these Assumptions in the context of Example \ref{Ex: CPP in Hilbert space} right below the subsequent abstract result.
\begin{theorem}\label{T: Rate of convergence for discretized estimator}
          If Assumptions \ref{As: H}(r) and \ref{As: spatial regularity}($\gamma$) hold for some $r\in (0,2)$, $\gamma \in (0,1/2]$, it is for all $\rho<(2-r)w$ \begin{equation}\label{eq: Rate of HF estimator}
              \sup_{t\in [0,T]}\left\|\Delta_n^{-2}\mathcal T_{\hat q^{n,-}_t}-[X^C,X^C]_t\right\|_{\text{HS}}=\mathcal O_p\left(\Delta_n^{\min(\rho,\gamma)}\right).
        % \quad \text{ and }\quad    \sup_{t\in [0,T]}\left\|\Delta_n^{-2}\mathcal T_{\hat q^{n,M,-}_t}-[\Pi_mX^C,\Pi_mX^C]_t\right\|_{\text{HS}}=\mathcal O_p(\Delta_n^{\min(\rho,\gamma)}),
          \end{equation}
         % where, again, $\Pi_M f (x)= \indicator_{[0,M]}(x)f(x)$ for all $f\in L^2(\mathbb R_+)$. 
    In particular, if $r<2(1-\gamma)$,  $w\in (\gamma/(2-r),1/2]$ it is
   \begin{equation}\label{eq: Rate of HF estimator under stricter conditions}
              \sup_{t\in [0,T]}\left\|\Delta_n^{-2}\mathcal T_{\hat q^{n,-}_t}-[X^C,X^C]_t\right\|_{\text{HS}}=\mathcal O_p(\Delta_n^{\gamma}).
         % \quad \text{ and }\quad      \sup_{t\in [0,T]}\left\|\mathcal T_{\hat q^{n,M,-}_t}-[\Pi_MX^C,\Pi_MX^C]_t\right\|_{\text{HS}}=\mathcal O_p(\Delta_n^{\gamma}).
          \end{equation} 
\end{theorem}

The rate implied by Theorem \ref{T: Rate of convergence for discretized estimator} is at most $\mathcal O_p(\Delta_n^{ 1/2})$, which is achieved if Assumptions \ref{As: H}(r) and Assumption \ref{As: spatial regularity}($\gamma$) hold for $\gamma = 1/2$ and $r<1$. 
Let us now discuss Theorem \ref{T: Rate of convergence for discretized estimator} and  Assumptions \ref{As: spatial regularity}($\gamma$) and \ref{As: H}($r$) in the context of Example  \ref{Ex: CPP in Hilbert space}.
\begin{example}\label{Ex: continued about spatial regularity}  [Example \ref{Ex: CPP in Hilbert space} ctn.]
Let us again assume that $\sigma$ is constant,  write $Q:=\sigma\sigma^*$  and let $J$ be a compound Poisson process.  Since $Q$ is Hilbert-Schmidt, it can be written as an integral operator $\mathcal T_{q}$ corresponding to a kernel $q\in L^2(\mathbb R_+^2)$. The regularity Assumption \ref{As: spatial regularity}($\gamma$) for some $\gamma\in [0,1/2]$ is then guaranteed if for all $M>0$ it is
$$\sup_{r>0}\int_0^{M}\int_0^{M} \frac{(q(r+x,y)-q(x,y))^2}{r^{2\gamma}}dxdy<\infty.$$
This is the case, for instance, if $q$, as a function on $\mathbb R^2_+$  is locally $\gamma$-Hölder continuous. 

The regularity Assumption \ref{As: H}(r) is foremost an Assumption on the jump activity. Indeed,  in our case,  in which the jumps correspond to a compound Poisson process with jump-distribution $F$,  we always have that
$\int_{H\setminus \{0\}} (z\wedge 1)^r F(dz)\leq F(H\setminus \{0\})=1<\infty$ does hold for all $r> 0$,  and hence, the Assumption holds for all $r\in (0,2)$.  In particular, we can choose $w\in [\gamma ,\frac 12]$ and $r<2(1-\gamma)$ to derive the rate of convergence in \eqref{eq: Rate of HF estimator under stricter conditions}.
\end{example}

If we assume a slightly stronger Assumption than \ref{As: spatial regularity}($1/2$), which can also be found in Section \ref{Sec: Assumptions for Error bounds},  we can even obtain a stable CLT in the next result, where stable convergence in law is denoted by $\overset{st.}{\longrightarrow}$ \footnote{ Recall that a sequence of random variables $(X_n)_{n\in\mathbb N}$ defined on a probability space $(\Omega, \mathcal F,\mathbb P)$ and with values in a Hilbert space $H$ converges stably in law to a random variable $X$ defined on an extension $(\tilde{\Omega}, \tilde{\mathcal F},\tilde{\mathbb P})$ of $(\Omega, \mathcal F,\mathbb P)$ with values in $H$, if for all bounded continuous $f:H\to \mathbb R$ and all bounded random variables $Y$ on $(\Omega,\mathcal F)$ we have
$\mathbb E[Y f(X_n)]\to\tilde {\mathbb E}[Y f(X)]$ as $n\to\infty$,
where $\tilde {\mathbb E}$ denotes the expectation w.r.t. $\tilde{\mathbb P}$.}.
\begin{theorem}\label{T: CLT for truncated estimator}
  Let Assumption \ref{as: CLT} hold.
  Then Assumption \ref{As: spatial regularity}($1/2$) holds. Moreover, let Assumptions \ref{As: H}(r)  hold for  $r<1$. Then for $w\in [1/(2-r),1/2]$ we have  for every $t\geq 0$ that
  $$\sqrt n \left(\Delta_n^{-2}\mathcal T_{\hat q^{n,-}_t}-[X^C,X^C]_t\right)\overset{st.}{\longrightarrow} \mathcal N(0,\mathfrak Q_t),$$
    where $\mathcal  N(0,\mathfrak Q_t)$ is for each $t\geq 0$ a Gaussian random variable in $L_{\text{HS}}(L^2(\mathbb R_+))$ defined on a very good filtered extension\footnote{See Section 2.4.1 in \cite{JacodProtter2012} for the definition of very good filtered extensions.} $(\tilde{\Omega},\tilde{\mathcal F},\tilde{\mathcal F}_t,\tilde {\mathbb P})$ of $(\Omega,\mathcal F,\mathcal F_t, \mathbb P)$ with mean $0$
  and covariance process  $\mathfrak Q_t:L_{\text{HS}}(L^2(\mathbb R_+))\to  L_{\text{HS}}(L^2(\mathbb R_+))$ given as
  $$\mathfrak Q_t K =\int_0^t \Sigma_s\left(K+K^*\right) \Sigma_s   ds.$$
  Here $\Sigma_t:=\sigma_t\sigma_t^*=\partial_t [X,X]_t$ is the squared volatility operator.
%$$\int_0^t \int_0^{\infty}q_s^c(\cdot,y) \cdot(y)dy(\int_0^{\infty}q_s^c(\cdot,y) \cdot(y)dy\int_0^{\infty} a(\cdot, s)\cdot(s)ds+ \int_0^{\infty} a(s,\cdot)\cdot(s)ds)\int_0^{\infty}q_s^c(\cdot,y) \cdot(y)dy ds$$
%$$Q A f (t)=QAQ f(t)+QA^*Q f(t)$$
%$$QAQ f(\cdot)=QA \int_0^{\infty} q(\cdot,y) f(s)ds (\cdot)=\int_0^{\infty}\int_0^{\infty}\int_0^{\infty} q(\cdot,x) a(x,y) q(y,s)  dxdy f(s) ds $$
%so it has kernel
%$$qaq(t,s)=\int_0^{\infty}\int_0^{\infty} q(t,x) a(x,y) q(y,s)  dxdy+\int_0^{\infty}\int_0^{\infty} q(\cdot,x) a(x,y) q(y,s)  dxdy$$
\end{theorem} 
The partial  derivative  $\Sigma_t:=\partial_t [X,X]_t$  has to be interpreted as a Frechet-derivative and does always exists, due to the Assumption on $X$ being an It{\^o} semimartingale (c.f. Section \ref{Sec:  semimartingales in Hilbert spaces}).
Let us derive the form of $\mathfrak Q_t$ in the context of Example   \ref{Ex: CPP in Hilbert space}:
\begin{example}[Example  \ref{Ex: CPP in Hilbert space} ctn.]
in the setting of example \ref{Ex: CPP in Hilbert space} the asymptotic covariance operator $\mathfrak Q_t$ has the form
$$\mathfrak Q_t=t (Q(\cdot+\cdot^*)Q).$$
Equivalently,  $\mathfrak Q_t$ can be interpreted as a kernel operator on $L^2(\mathbb R_+^2)$ with kernel
$$\mathfrak q_t(x,z,w,y):= t\left(q(x,z)q(w,y)+q(x, w)q(z,y)\right).$$
In particular, $\mathfrak q_t(x,z,w,y)$ can be consistently estimated by the plug-in estimator
$\hat{\mathfrak q}_T^n(x,z,q,y):= \Delta_n^{-4}T^{-1}(\hat q_T^{n,-}(x,z)\hat q_T^{n,-}(w,y)+\hat q_T^{n,-}(x,w)\hat q_T^{n,-}(z,y))$.
%$$\mathfrak Q_tB f(x)=tQ(B+B^*)Q f(x)=\int_{\mathbb R_+}\int_{\mathbb R_+} \int_{\mathbb R_+}  q(x,z) q(w,y) (b(w,z)+b(z,w))   f(y)dy dw dz$$
%$$\int_{\mathbb R^2 }\mathfrak q_t(\cdot,   \int_{\mathbb R_+}\left(\int_{\mathbb R_+^2}  q(\cdot,z)b(z,w) q(w,y) dw dz+\int_{\mathbb R_+^2}  q(\cdot,w) b(z,w) q(z,y) dw dz \right) f(y)dy (x)$$
Assumption \ref{as: CLT} holds, for instance, if $q$ is locally $\gamma$-Hölder  continuous for some $\gamma >1/2$ except on finitely many discontinuity points. In particular,  the CLT holds if $q$ is smooth.
\end{example}

So far we have discussed limit theorems for infill asymptotics leaving $T$ fixed. In the next section, we outline how it is possible under additional assumptions and letting $T\to\infty$ to make use of all available data to estimate the stationary instantaneous covariance for difference returns.

\subsubsection{Long-time volatility estimation}\label{Sec: Long time Volatility estimation}
%Our analysis in section ... shows that interpretation of the sum of sqared forward returns as a proxy of a global covariance might be flawed, due to the heteroskedastic behavior of the volatility operator and the different behavior when using varying sampling frequencies.

The truncated estimation procedure described previously enables estimations of a time series of the integrated volatilities $\int_{i}^{i+1}\Sigma_s ds=[X^C,X^C]_i-[X^C,X^C]_{i-1} $ for $i\in \mathbb N$. %, or respectively the quadratic variations $[X,X]_i-[X,X]_{i-1}$ for $i=1,...,T-1$, which might vary with $i$. 
If the aim is to derive a time-invariant mean for the volatility, we have to impose further conditions, which are described in detail in Section \ref{Sec: Assumptions for Long-time estimators}. 
These Assumptions are much stricter than the ones we considered in the previous section for the infill asymptotics on finite intervals and in particular imply that the mean 
$$\mathcal C:=\frac 1 T\mathbb E[[X,X]_T]$$
is independent of $T$.
However, they
 allow us to derive a stationary mean of $\Sigma_t$ via large $T$ asymptotics. 
\begin{theorem}\label{T: Long-time asymptotics for termstructure volatiltiy}   %gamma=1/2, r>0, 1/(4-2r)=w
%p>2(2-r)/(1-r)>4
%r=1/2
%p>(4-3r)/(2-r) =1+2(1-r)/(2-r)>2
Let Assumptions \ref{As: Mean stationarity and ergodicity}, \ref{In proof: Very very Weak localised integrability Assumption on the moments}(p,r) and \ref{In proof II}($\gamma$) hold for some $r\in (0,2)$, $\gamma \in (0,1/2]$
and $p>\max(2/(1-2w),(1-wr)/(2w-rw))$. 
%Moreover, let $c:\mathbb R_+^2\to\mathbb R$ be the integral kernel corresponding to the covariance operator $\mathcal C$.
Then we have as $n,T\to\infty$ that 
$$\Delta_n^{-2}T^{-1}\mathcal T_{\hat q_T^{n,-}}\overset{p}{\longrightarrow}\mathcal C. %\quad \text{ and } \quad\frac{\mathcal T_{\hat q_T^{n,M,-}}}T\overset{p}{\longrightarrow}\Pi_M\mathcal C\Pi_M,
$$ If even $r<2(1-\gamma)$ and $w\in(\gamma/(1-2w),1/2)$ and additionally $p\geq 4$ we have with $a_T= \| [X^C,X^C]/T- \mathcal C\|_{\text{HS}}$ that
$$\left\| \Delta_n^{-2}T^{-1}\mathcal T_{\hat q_T^{n,-}}-\mathcal C\right\|_{L^2(\mathbb R_+^2)}=\mathcal O_p(\Delta_n^{\gamma}+a_T).
%\quad\text{ and } \quad \left\| \frac{\mathcal T_{\hat q_T^{n,M,-}}}T-\Pi_M \mathcal C\Pi_M\right\|_{L^2([0,M]^2)}=\mathcal O_p(\Delta_n^{\gamma}+a_T).
$$
\end{theorem}
Assumption \ref{As: Mean stationarity and ergodicity} does not impose very strong Assumptions on the dynamics of the volatility and is satisfied by most stochastic volatility models.
To verify this condition for particular models for the infinite-dimensional volatility process one might investigate the vast literature for ergodic properties of Hilbert space-valued processes and, in particular, SPDEs (c.f. \cite[Sec.10]{DPZ2014} or \cite[Sec.16]{PZ2007}). For the existence of invariant measures for term structure models, we further mention \cite{vargiolu1999}, \cite{Tehranchi2005}, \cite{marinelli2010}, \cite{rusinek2010},% \cite{goldys2000}
and \cite{FFRS2020}. Recently, \cite{friesen2022} examined the long-time behavior of infinite-dimensional affine volatility processes.
Here, we only review the validity of the Assumptions employed in Theorem \ref{T: Long-time asymptotics for termstructure volatiltiy} in the context of our running Example \ref{Ex: CPP in Hilbert space}:
\begin{example}[Example  \ref{Ex: CPP in Hilbert space} ctn.]
Once more,  consdier the setting of Example \ref{Ex: CPP in Hilbert space}.  Assumption \ref{As: Mean stationarity and ergodicity} requires stationarity and mean ergodicity  on the continuous part of the quadratic variation, which is trivially fulfilled, since $\mathcal C=[X,X]_T/T=Q$ for all $T>0$. Assumption \ref{In proof: Very very Weak localised integrability Assumption on the moments}(p,r) is valid for all $p>0$ and $r>0$ since all coefficients of the semimartingale $X$ are deterministic and constant
%, since the drift is deterministic and constant over time, while $\mathbb E[\|\sigma_s\|^p_{\text{HS}}]=\|\sigma\|^p_{\text{HS}}$ and $\nu= F$ such that $\mathbb E[\int_{H\setminus\{0\}}\|\gamma_s(z)\|^r F(dz)]=\int_{H\setminus\{0\}}\|z\|^r F(dz)= \mathbb E[\chi_1^r]$, which is, by Assumption, finite for all $r\in [0,2]$.  
Assumption \ref{In proof II}($\gamma$)
holds for $\gamma\in (0,\frac 12]$ under analogous conditions in  as Example \ref{Ex: continued about spatial regularity}.

While $Q$ can be estimated without the long-time regime,  it is simple to find situations when long time asymptotics provide additional information such as for the estimation of HEIDIH models from \cite{Petersson2022}, which is described in \cite{Schroers2024}.  Another example, which is implemented in the simulation study in Section \ref{Sec: Simulation Study} is that $\Sigma_s = x_s Q$ for a positive scalar mean reverting process $x$, which models the changing magnitude of  volatility over time. Then the long time estimator can be used to determine the mean-reversion level of $x$. 
\end{example}

\subsection{ Practical considerations}\label{Sec: Practical Considerations}

In this section,  we discuss some practical complications on the implementations of the estimator. Namely,  we present a data-driven truncation rule and comment on the use of nonparametrically smoothed yield or bond price curve data.

We start with a data driven choice of the truncation function $g_n$ and the tuning parameters $\alpha$ and $w$.

\subsubsection{Truncation in practice}\label{Sec: truncation in practice}

While the asymptotic theory of Section \ref{Sec: term structure models theory} justifies the use of truncated estimators, the choice of the truncation level and functions remains a practical issue. Even in finite dimensions, this can be challenging and we refrain from finding optimal choices. However, we outline how the truncation rule can be reasonably implemented.

%We assume that the jumps are given by a compound Poisson process defined in Example \ref{Ex: CPP in Hilbert space} with jump intensity $\lambda>0$. We assume that we have a rough idea on an upper bound of how many jumps occur on average in the course of a year in the daily bond market data. For instance, we might assume that not more than 10 percent of the incremental data are contaminated by a jump, such that we have $\lambda<25$. This is a decision on how many samples we want to sort out, but it does not tell us, which increments should be discarded. For that, we need a reasonable way to measure the outlingness of increments.

Truncation rules in finite dimensions often necessitate preliminary estimators for the average realized variance in the corresponding interval of interest (c.f. \cite{Mancini2012}, page 418, for an overview of some truncation rules). One sorts
out a large amount of data first, to obtain a preliminary estimator of $[X^C,X^C]_T$. As this can be interpreted as the average covariation of the increments in the interval $[0,T]$, one then chooses truncation levels in terms of multiples of standard deviations as measured by the preliminary estimate.
%In that way, one can gain a rough idea of what are typical increments corresponding to the continuous part of the semimartingale and sort out untypical ones. 
In our infinite-dimensional setting, we mimic this procedure,
but it is harder to distinguish typical increments and outliers as we cannot argue componentwise. 
While the choice $g_n\equiv \|\cdot\|$ leads to consistent estimators in terms of the limit theory developed in Section \ref{Subsec: Truncated Estimation in term structure models}, it is not necessarily a good choice in the context of finite data since the continuous martingale might vary considerably more in one direction %$d_1\in L^2(\mathbb R_+)$ 
than another% $d_2\in L^2(\mathbb R_+)$
. 

We Therefore present a method that is based on a measure of functional outlyingness in the spirit of \cite{ren2017}: Assuming that $\Sigma$ is independent of the driving Wiener process and does not vary too wildly on the interval $[0,T]$, and that no jumps exist, we have approximately that $\Sigma_t \approx \frac 1T\int_0^T \Sigma_s ds$ for $t\in [0,T]$ and $\tilde \Delta_i^n f/\sqrt{\Delta_n}|\Sigma\sim N(0, \frac 1T\int_0^T \Sigma_s ds)$.  
If the largest $d$ eigenvalues $e_1,...,e_d$ of $\int_0^T \Sigma_s ds/T$ account for a large amount of the variation as measured by the summed eigenvalues of $\int_0^T \Sigma_s ds/T$ (e.g. 90 percent), we know that $P_d \tilde \Delta_t^n f$ with $P_d=\sum_{i=1}^d e_i^{\otimes 2}$ is a linearly optimal approximation of $\tilde \Delta_t^n f$ in the $L^2(\mathbb R_+)$-norm.  We can also define
$P_d(\int_0^T \Sigma_s ds/T)^{-1}P_d =\sum_{i=1}^d(1/\lambda_i)e_i^{\otimes 2}$ and define $g^d(h):=\| P_d(\int_0^T \Sigma_s ds/T)^{-1}P_d h\|^2$.
This distance resembles the measure proposed in \cite{ren2017}, however, it is not a valid truncation function, since \eqref{def: g} cannot hold. Further, if a truncation at level $d$ is made, outliers impacting the higher eigenfactors might be overlooked.
  We Therefore propose an adjusted method defining $g(x)=g^d(x)+ \frac{\|(I-P_d)x\|^2}{\sum_{i=d+1}^{\infty} \lambda_i} $. Then, for $p_n:l^2\to L^2(\mathbb R_+)$ given by 
        $p_n x:= \sum_{j=1}^{\infty} x_j \indicator_{[(j-1)\Delta_n,j\Delta_n]}$
        we define the sequence of truncation functions $(g_n)_{n\in \mathbb N}$ via
\begin{equation}\label{Mahalanobis truncation}
            g_n(x)^2:=g(p_n(x))^2:= \sum_{i=1}^d \frac {\langle p_n(x), e_i\rangle}{\lambda_i}^2+ \frac{\sum_{i=d+1}^{\infty} \langle p_n(x), e_i\rangle^2}{\sum_{i=d+1}^{\infty} \lambda_i},
        \end{equation}
      where the index $d$ can be chosen freely as long as $\lambda_{d+1}>0$. E.g. we can choose $d$ such that the first $d$ eigenfactors for $\mathcal C$ explain $90$ \% of the variation. %If there was a jump in $[(i-1)\Delta_n,i\Delta_n]$ that differs from typical increments in size or shape,
       %then $$   \mathbb E\left[g_n\left(\delta_{\cdot\leq \lfloor M/\Delta_n\rfloor}\frac {d^n(i,\cdot)}{\Delta_n}\right)\right]
       %    \approx  \mathbb E\left[g\left(\Pi\tilde \Delta_i^n f\right)\right]$$
      %    should be large. Hence, $g_n$ is a reasonable measure of outlyingness. A similar approach to detecting outliers for functional data was conducted by \cite{ren2017}.
  It is then with $\lambda_i^d:=\lambda_i$ for $i\leq d$ and $\lambda_i^d:=\sum_{j=d+1}^{\infty} \lambda_j$ for $i\geq d$ and with $\Delta_n$ small 
     \begin{align*}
           \mathbb E\left[g_n\left(
          \frac{d^i_n}{\Delta_n}\right)^2|\Sigma\right]
           \approx  \mathbb E\left[g\left(\tilde \Delta_i^n f\right)^2|\Sigma\right]
           %= & \mathbb E\left[ \sum_{i=1}^{\infty}\frac{\langle\tilde \Delta_i^n f,e_i\rangle^2}{\lambda_i^d}|\Sigma\right] 
           \approx %\sum_{i=1}^{\infty} \frac{\langle \Delta_n \frac 1T(\int_0^T \Sigma_s ds) e_i,e_i\rangle}{\lambda_i^d} =
           \Delta_n \left( d+ 1\right).
       \end{align*}
    In practice, we do not know the eigenvalues and eigenfunctions of $\int_0^T \Sigma_s ds/T$ and derive them from a preliminary estimate. 
We suggest a simple truncation procedure in two steps: 
\begin{itemize}
%    \item[(i)] For large $T$ choose a truncation level $u$, such that a large amount, say $0.25$, of the increments is sorted out by $\hat q_t(-)$ (resp. $SARCV_t^n(u,-)$) with the choice $g_n(x)=\|\cdot\|_{l^2}$. That is, $25$ percent of the increments satisfy $g_n\left(%\delta_{\cdot\leq \lfloor M/\Delta_n\rfloor}
%    \frac {d^n(i,\cdot)}{\Delta_n}\right)\approx\|\tilde{\Delta}_i^n f\|_{L^2(\mathbb R_+)}>u$. Then we define the preliminary estimate
 %   \begin{align*}
%\frac {1}{0.75*T}q_T(-)        \approx\mathcal C.   \end{align*}
   \item[(i)] First we have to specify a preliminary estimator which can be found as follows: For fixed $T$ choose a truncation level $u$, such that a large amount, say $0.25$, of the increments is sorted out by $\hat q_T^{n,-}$ with the choice $g_n(x)=\|\cdot\|_{l^2}$. That is, $25$ percent of the increments satisfy $g_n\left(%\delta_{j\leq \lfloor M/\Delta_n\rfloor}
    d_n^i/\Delta_n\right)
    %\approx\|\tilde{\Delta}_i^n f\|_{L^2(\mathbb R_+)}
    >u$. Then we define the preliminary estimate
    \begin{align*}
\rho^*\Delta_n^{-2}T^{-1} \hat q^{n,-}_t       \approx\frac 1T \int_0^Tq_t^cdt,   \end{align*}
where $\rho^*>0$ properly rescales the preliminary estimator  (one rescaling procedure is outlined in the appendix. %supplementary material to this article). %No rescaling leads to lower truncation levels, which can be acceptable.
%The scaling factor can be chosen in many reasonable ways and one is described in Section \ref{Sec: bias adjustment preliminary estimator}. 
\item[(ii)] Now set $g_n$ as in \eqref{Mahalanobis truncation}, $d$ such that the first $d$ eigenvalues of the operator corresponding to the kernel $\rho^*\Delta_n^{-2}T^{-1} \hat q^{n,-}_t$ explain $90\%$ of the variation measured by the sum of eigenvalues of this operator and choose $u_n= l \sqrt{d+1}\Delta_n^{0.49}$ for an $l\in \mathbb N$. E.g. we might take $l=3,4$ or $5$. Observe that for $d$ large enough $g_n(d_n^i/\Delta_n)^2/\Delta_n\approx g^d(p_n(d_n^i/\Delta_n)^2/\Delta_n$ is under the above local normality assumptions approximately $\chi^2$ distributed with $d$ degrees of freedom. Hence, the probability that $g_n(d_n^i/\Delta_n)<u_n$ can be approximated by the cumulative distribution function of a $\chi^2$-distribution with $d$ degrees of freedom. For instance, if $d=4$ we have that $g_n(d_n^i/\Delta_n)<u_n$ with $l=4$ would hold for approximately 98.26\% of the increments.
Then we can implement the estimators of Section \ref{Sec: term structure models theory} with these choices for truncation function and level.
%as the largest real number such that for at least 10 percent of the incremental data we have \textcolor{red}{Better in terms of 3,4,5 times the standard deviations!}
%$$g_n\left(\delta_{\cdot\leq \lfloor M/\Delta_n\rfloor}\frac {d^n(i,\cdot)}{\Delta_n}\right)>u_n.$$
\end{itemize}
Arguably, there can be many other methods for deriving truncation rules, which however have to take into account the infinite dimensionality of the data and deal with the subtlety of functional outliers. The simulation study in Section \ref{Sec: Simulation Study} shows the good performance of our method.  

\subsubsection{Presmoothing bond market data}
It is rarely the case that term structure data are observed in the same resolution in time as in the maturity dimension. For bond market data, points on the discount curve $x\mapsto P_{i\Delta_n}(x)$ are observed irregularly with a lower resolution than daily along the maturity dimension. Additionally, information on the discount curve is sometimes latent as bonds are often coupon-bearing and assumed to be corrupted by market microstructure noise.
To account for these difficulties and in accordance with the classical ``smoothing first, then estimation" procedure for functional data analysis advocated in \cite{Ramsay2005}
we pursue the simple yet effective approach of presmoothing the data. We derive smoothed yield or discount curves, as described, for instance, in \cite{FPY2022}, \cite{liu2021} or \cite{Linton2001}, which allows us to derive approximate zero coupon bond prices for any desired maturity and for which the impact of market microstructure noise is mitigated.
%Such preprocessed  data are also widely considered as primary data inputs  for the calibration of parametric term structure models. 
%From the point of view of functional data analysis this corresponds to the classical approach advocated in \cite{Ramsay2005} to smooth first and then use the derived curves as data inputs for estimation. 
While some theoretical guarantees in terms of asymptotic equivalence of discrete and noisy to perfect curve observation schemes could be derived for certain smoothing techniques and the task of estimating means and covariances of i.i.d. functional data (c.f. \cite{Zhang2007}), in our case they would depend on the respective smoothing technique, the volatility and the semigroup as well as the magnitude of distorting market microstructure noise. A detailed theoretical analysis in that regard is beyond the scope of this article and instead, we showcase the robustness of our approach in the context of sparse, irregular and noisy bond price data within a simulation study in Section \ref{Sec: Simulation Study}.

\section{Simulation study}\label{Sec: Simulation Study}
%\textcolor{red}{Our aim is to focus on effects of sparse and random sampling as well as microstructure noise on estimation errors and feauture selection of the volatilities. For that reason, we choose a framework which can mimic easily potential low-dimensional and infinite-dimensional volatility situations and from which we can sample more or less exactly. For that purpose, and since there is not much work on exact sampling of infinite-dimensional stochastic volatility model, we necessarily trade off some realistic features of forward curves, most prominently their positivity, for the simplicity of the frameworks, which allows us to draw (discretized) trajectories exactly.  }
% As an important application of our theory is the identification of the number of statistically relevant drivers, we examine how reliable the estimator can be used to determine the effective dimensions of X.

In our simulation study we examine the performance of the truncated estimator $\mathcal T_{\hat q^{n,-}_{\cdot}}$ defined in \eqref{truncated estimator} as a measure of the continuous part $[X^C,X^C]$ of the quadratic covariation of the latent driver. As an important application of our theory is the identification of the number of statistically relevant drivers, we also examine how reliable the estimator can be used to determine the effective dimensions of $X^C$.
In this context, we also want to assess the robustness of our estimator concerning three important aspects: First, we need to confirm the robustness of the truncated estimator to jumps. Second,  we study the effect of the  common practice of presmoothing sparse, noisy, and irregular bond price data on the estimator's performance. Moreover, we examine how the routine of projecting these data onto a small finite set of linear factors (c.f. for instance \cite{Litterman1991} or the survey \cite{piazzesi2010}), influences conclusions on the quadratic covariation. % and in particular, the measured number of statistically relevant random drivers.

%We make two comparisons: 
%First, we  compare its performance of $\hat q^{n,-}_t$ when derived from perfectly observed dense data to the case when it is derived via presmoothing of the discount curves from sparse and noisy data as it is often the case in applications. Second we evaluate its performance when the yield data are projected onto the first three eigenfunctions of the covariance of yields, which is a common procedure for term structure data. 

%investigate if $\hat q^{n,-}_t$ indicates the correct number of random drivers needed to describe the variation of the latent driving semimartingale adequately. As a benchmark, we compare these numbers to the ones indicated by the commonly considered  principal component analysis for yield curves. Third, we want to check whether our method correctly detect jumps.

%accuracy, driver number, Jumpscompared to the quadratic variation in the case without jumps and perfectly observed discount curves.

%In our simulation study we want to compare the performance of the truncated estimator $\hat q^{n,-}_t$ in different scenarios considering different sizes and numbers of jumps and the effect of using nonparametrically smoothed discount curvesfrom sparse and noisy data compared to the quadratic variation in the case without jumps and perfectly observed discount curves. 

For that,  
%From the local average-samples we can derive synthetic log-Bond prices 
%$$\log P_{i\Delta_n}(j\Delta_n)= -\sum_{l=1}^{j} F_{i,l}, \qquad i=1,...,100,\quad j=1,...,1000.$$
%To mimic the situation for bond market data, we  sample sparsely and irregularly %where the number of samples locally decreases with the size of the maturity. 
%and consider the prices to be contaminated by i.i.d. noise to account for the existence of market-microstructure noise. Precisely, 
 we simulate log bond prices for some sampling size $m\leq 1000$, and $n=100$ time points, that is,
$$ P_{i,l}:=\log P_{i\Delta_n}(j_{i,l}\Delta_n)+\epsilon_{i,l}, \qquad i=1,...,100,\quad l=1,...,m,$$
where $(\epsilon_{i,j})_{i,j=1,...,100}$ are i.i.d centered Gaussian errors with variance $\sigma_{\epsilon }>0$ and the $j_{i,r}$ are drawn randomly from $\{0,...,1000\}$ without replacement for $r=1,...,m$.
We distinguish two cases: First, as a benchmark, we observe the data densely and without noise such that $\sigma_{\epsilon}=0$ and $m=1000$ and, second, we observe the data with noise $\sigma_{\epsilon}=0.01$ and sparsely with $m=100$.
In this case, the prices  for all maturities $j\Delta_n, j=1,...,1000$ are recovered  by quintic spline smoothing.
%and with that proxies for the discretized forward curve data $F_{i,j},j=1,...,n$, with which we can implement the estimator $\hat q_t^{n,10}(-)$ from section \ref{Sec: Estimation of bond market Volatility}.
The roughness penalty for the smoothing splines is for each date chosen by a Bayesian information criterion and implemented via the \texttt{ss}-function from the \texttt{npreg} package in \texttt{R}.  
Using quintic splines and a Bayesian information criterion induces smooth implied forward curves.% and, hence, yields high smoothness on the estimated volatility term structures.

To analyze the impact of the customary procedure
of projecting the bond price data onto a low-dimensional linear subspace, we conduct our experiments in two scenarios. Scenario 1 in which  we do not project the log bond prices and Scenario 2 in which
we project the log bond prices onto the first three eigenvalues of their covariance $c_{logbond}\equiv\frac 1{100}\sum_{i=2}^{100} (P_{i,\cdot}-P_{i-1,\cdot})(P_{i,\cdot}-P_{i-1,\cdot})'$ before we calculate $\hat q^{n,-}_t$. Indeed, as usual for bond market data, the first three eigenvectors of $c_{logbond}$ explain over 99\% the variation in the log bond prices.

 The log bond prices are derived from simulated instantaneous forward rates
%discrete samples
$F_{i,j}:=\langle \indicator_{[(j-1)\Delta_n, j\Delta_n]}, f_{i\Delta_n}\rangle_{L^2(\mathbb R_+)}$
for $n=100$, $i=1,...,100$ and $j=1,...,1000$
from a forward rate process driven by a semimartingale $X$.
%The dynamics and parameter specifications of $X$ reflect different qualitative properties such as the dimension of the quadratic variation as well as the number and magnitude of jumps. 
Precisely, we define
$X_t= \int_0^t \sqrt{\Sigma_s}dW_s+ J_t$ where
 \begin{align*}
   % \alpha(\cdot)= e^{-a_0 \cdot},\qquad 
   \Sigma_s = x(s) Q_{\alpha},%+Q_0, 
    \qquad J_t=J_t^1+ J_t^2 \text{ where }J_t^i=  \sum_{l=1}^{N^i_t} \chi_l^i.
\end{align*}
Here $x$ is a univariate mean-reverting square root process 
$$dx(t)=1.5\left(0.058- x(t)\right)dt+ 0.05 \sqrt{x(t)}d\beta(t), \quad t\geq 0,\,\,x(0)=0.058$$
%$$dx(t)=\kappa\left(\mu- x(t)\right)dt+ b \sqrt{x(t)}d\beta(t), \quad t\geq 0,\,\,x(0)=\mu$$
and $Q_{a}$ is a covariance operator on $L^2(\mathbb R_+)$ such that the corresponding covariance kernel $q_a$ (s.t. $Q_a=\mathcal T_{q_a}$) restricted to $[0,M]$ is a Gaussian covariance kernel $q_{a}(x,y) \propto \exp(-a(x-y)^2)$ for some $a>0$ and $\|q_{a}\big|_{[0,M]^2}\|_{L^2([0,M]^2)}=1$.   The jumps are specified by two Poisson processes $N^1,N^2$ with intensities $\lambda_1, \lambda_2>0$ and jump distributions $\chi_i^2  \sim N(0,\rho_2 Q_{0.01})$ and $\chi_i^1  \sim N(0,\rho_1K)$ for $\rho_1,\rho_2 \geq 0$ and where $K$ is another covariance operator with kernel $k(x,y)\propto e^{-(x+y)}$ and $\|k\big|_{[0,M]^2}\|=1$. 
% The operator $Q_a$ has infinitely many nonzero eigenvalues and induces an infinite-dimensional term structure model. However, a high parameter $a$ implies more complex eigenstructures of this operator.
%We will compare both a high-dimensional ($a=50$) and a low-dimensional ($a=5$) setting.

We specify the corresponding parameters of this infinite-dimensional model as follows: We choose $a=50$ reflecting a high dimensional setting since the decay rate of the eigenvalues of $Q_a$ is slow (10 eigencomponents are needed to explain 99\% of the variation of $X^C$). %A corresponding low-dimensional setting is discussed in the appendix.
 The mean reversion level $0.058$ of the square root process corresponds to the Hilbert-Schmidt norm of the long-time estimator of volatility derived from bond market data as discussed in the next section.
 %For instance, choosing $a\approx 30 $ we need ... eigenelements to explain 99\% of the variation induced by the covariance $Q_a$ whereas for $a=.1$ we only need one.
  %, $Q_1$ leads to the nonsemimartingality of the continuous component and implies a rough short rate model. (int_0^10 e^{-10 x}dx)^2=(1/10)((1-e^{-10*10})/10)^2
 Jumps corresponding to the first component are considered large and rare outliers reflected by a high $\rho_1=0.0116$ and low $\lambda_1=1$. The second component describes outliers which are more frequent and smaller in norm reflected by a lower $\rho_2=0.0029$ and higher $\lambda_2=4$ but correspond to changes of the shape of the forward curves. Both jump processes are chosen such that their Hilbert-Schmidt norm accounts for approximately 10 \% of the quadratic variation.
 %, that is, $\rho_1/(0.058+\rho_1+4\rho_2)=0.1$ and $4\rho_2/(0.058+\rho_1+4\rho_2)=0.1$. 
 We also consider cases in which no jumps are present (corresponding to the parameter choices $\lambda_1=\lambda_2=0$).

In each considered scenario, we compute the estimator $\Delta_n^{-2}\hat q_1^{n,10,-}$ (as defined in Remark \ref{Rem: Reduction of the time to maturity is fine}) for $q_a\cdot \int_0^1 x(s)ds$.
In the cases in which jumps are present,
we consider the truncated estimator via the data-driven truncation rule of Section \ref{Sec: truncation in practice} with different values of $l=3,4,5$ for the truncation level $u_n=l\sqrt{d+1}\Delta_n^{0.49}$. Here $d$ is chosen  as the smallest value such that the first $d$ eigencomponents account for 90\% of the variation as measured by the preliminary covariance estimator for which we truncate at the 0.75-quantile of the sequence of difference return %$(\delta_{j\leq 1/\Delta_n}\tilde \Delta_{i\Delta_n}d(j\Delta_n))_{j\geq 0}$
curves as measured in their $l^2$ norm.  Models M1 and M2 for which no jumps are present serve as benchmarks for the truncated estimators and no truncation is conducted ($l=\infty$). 

%Moreover, we consider three different levels of sparsity. A very sparse scenario in which we only use $m=10$ data per day, a sparse scenario in which we observe $m=100$ points and a fully observed scenario with $m=1000$ without noise, which serves as our benchmark for the other scenarios. We consider the i.i.d. noise to be centered Gaussian with variance $\sigma_{\epsilon}=0.0058$. 

We assess the performance of the respective estimator in the context of two criteria of which each reflects an important application of our estimator. First, we measure the relative approximation error $rE(\Delta_n^{-2} \hat q_1^{n,10,-}, IV)$ where
for $x\in [(j_1-1)\Delta_n,j_1\Delta_n] $ and $y\in [(j_2-1)\Delta_n,j_2\Delta_n]$ the $\Delta_n$-resolution of the integrated volatility is  $IV(x,y):=\int_0^1x(s)ds \cdot n^2\int_{(j_1-1)\Delta_n}^{j_1\Delta_n}\int_{(j_2-1)\Delta_n}^{j_2\Delta_n}q_{a}(z_1,z_2) dz_1,dz_2$    and
 $$rE(k_1,k_2)
 :
 %=\frac{\|(k_1(j_1\Delta_n,j_2\Delta_n)-k_1(j_1\Delta_n,j_2\Delta_n))_{j_1,j_2=1,...,\lfloor M/1000\rfloor}\|_{Frob}}{\|(k_1(j_1\Delta_n,j_2\Delta_n))_{j_1,j_2=1,...,\lfloor M/1000\rfloor}\|_{Frob}}\approx 
 \frac{\| k_1- k_2 \|_{L^2([0,10]^2)}}{\| k_2 \|_{L^2([0,10]^2)}}\qquad k_1,k_2\in L^2([0,M]^2).$$
%$$rE(k_1,k_2):=\frac{\|(\hat q_1^{n,10}(-)(j_1\Delta_n,j_2\Delta_n)-IV(j_1\Delta_n,j_2\Delta_n))_{j_1,j_2=1,...,\lfloor M/1000\rfloor}\|_{Frob}}{\|(IV(j_1\Delta_n,j_2\Delta_n))_{j_1,j_2=1,...,\lfloor M/1000\rfloor}\|_{Frob}}\approx \frac{\|\hat q_1^{n,10}(-)- IV\big|_{[0,10]} \|_{L^2([0,10]^2)}}{\| IV\big|_{[0,10]} \|_{L^2([0,10]^2)}}.$$
%where $\|\cdot\|_{Frob}$ denotes the Frobenius norm. 
%Second, we report ratio of the number of detected jumps of the number $N_1^1+N_1^2$ of actual jumps
%$$JDR:= \frac{\text{Number of correctly detected jumps}}{\text{True number of jumps}}.$$
Second, we will investigate how reliably the estimator can be used to determine the number of factors needed to explain certain amounts of variation of the latent driving semimartingale. For that, we define 
\begin{equation}\label{Dimensionality measure}
    D_{ C}^{e}(p):=min\left\{d\in \mathbb N:\frac{\sum_{i=1}^d \langle e_i,Ce_i\rangle}{\sum_{i=1}^{\infty} \langle e_i,Ce_i\rangle}>p\right\}\qquad p\in [0,1]
\end{equation} 
for a symmetric positive nuclear operator $ C$ and an orthonormal basis $e=(e_i)_{i\geq 0}\subset L^2([0,M])$. Let  $\hat e=(\hat e_i)_{i\in \mathbb N}$ denote the eigenfunctions of $\mathcal T_{\hat q_1^{n,10,-}}$ ordered by the magnitude of the respective eigenvalues. We  report the numbers $D_C^e$ for
$C=\mathcal T_{\hat q_1^{n,10,-}}$, $e=\hat e$ and $p=0.85$, $0.90$, $0.95$ and $0.99$, which are the numbers of factors needed to  explain respectively $85\%$, $90\%$, $95\%$ and $99\%$ of the variation.
%It reports medians and quartiles for the relative errors defined for
% $IV(x,y):=\int_0^1x(s)^2ds q_{a}(x,y)$
%$$rE=\frac{\|(\hat q_1^{n,10}(-)(j_1\Delta_n,j_2\Delta_n)-IV(j_1\Delta_n,j_2\Delta_n))_{j_1,j_2=1,...,\lfloor M/1000\rfloor}\|_{Frob}}{\|(IV(j_1\Delta_n,j_2\Delta_n))_{j_1,j_2=1,...,\lfloor M/1000\rfloor}\|_{Frob}}\approx \frac{\|\hat q_1^{n,10}(-)- IV\big|_{[0,10]} \|_{L^2([0,10]^2)}}{\| IV\big|_{[0,10]} \|_{L^2([0,10]^2)}}$$
%and for any symmetric positive nuclear operator $\mathcal C$ and a an orthonormal basis $e=(e_i)_{i\geq 0}\subset L^2([0,M])$ we define
%$$d_{ C}^{e}(p):=min\left\{d\in \mathbb N:\frac{\sum_{i=1}^d \langle e_i,Ce_i\rangle)}{\|C\|_{nuc}}>p\right\}$$ for $p=0.85,0.9, 0.95, 0.99$ where $\hat \lambda_i^C$ denote the $i$th leading eigenvalue of the respective estimated operators $\hat \Sigma_1^{n,10}(-)$.$

Table \ref{Tab: Monte Carlo Results1} shows the results of the simulation study based on $500$ Monte-Carlo iterations. 
For Scenario S1, reflecting our proposed fully infinite-dimensional estimation procedure,  the log bond prices are not projected onto a finite-dimensional subspace a priori.
In this scenario, at least if jumps are truncated at a low level ($l=3$), the medians of relative errors are of a comparable magnitude when using either nonparametrically smoothed data (M2,M4) or perfect observations (M1,M3). While some jumps are overlooked by the truncation rule,
the medians of the relative errors in the cases  with jumps (M3,M4)  just moderately increased compared to the respective cases in which no jumps appeared (M1,M2), at least for a low truncation level.  %The jump detection rates indicate that the task of classifying an increment as an outlier performs a bit too conservatively, although, at least for $l=3$, this does not seem to affect the estimation accuracy severly. 
The reported dimensions needed to explain the various levels of explained variation are estimated quite reliably (observe that the true thresholds %as measured by the dimensionality of $Q_{50}$ 
are respectively $5,6,7$ and $10$). In the noisy and irregular settings (M2, M4) the estimators tend to add a dimension compared to the perfectly observed settings in the median but have low interquartile ranges, which contain the correct dimension. %For the 99\% threshold, the noise even corrects for the dimension which was too low in the case of perfect observations. While the latter is arguably a coincidence, 
We conclude that the measurement of dimensions can be conducted accurately under realistic conditions. 

Comparing scenarios S1 and S2 it becomes evident that the customary finite-dimensional projections of log bond prices (S2) affected the estimator's performance significantly. 
All of the medians of relative errors are significantly higher compared to the case in which no projection was conducted, while for the practically important case in which data were smoothed from irregular sparse and noisy observations ($M4$) and jumps were truncated at level $l=3$, the error more than doubled.
For all considered thresholds of explained variation ($85\%$, $90\%$, $95\%$ and $99\%$) the reported dimension is constantly $3$, where we just reported the results for the $99\%$ threshold in the table. This is not surprising, since we started from a three-factor model for the log bond prices, but it demonstrates, that the common practice of projecting price or yield curves onto a few linear factors can disguise statistically important information, despite their high explanatory power for the variation of log bond prices, which is in line with \cite{Crump2022}. %relative errors suggests that the projection of discount curves onto a few linear factors can indeed disguise statistically important information.

%  Moreover, the table reports the medians of the correctly detected jumps in each setting as well as the measured dimensions of the the estimators in terms of the number of eigenfunctions needed to explain respectively 85\%, 90\%, 95\% and 99\% of the variation. 
\begin{table}\caption{The table reports the medians and quartiles (in braces) rounded to the second decimal place for the relative errors and the  number of factors that is necessary to explain respectively 85\%, 90\%, 95\% and 99\% of the variation for which the true thresholds are resp. $5$, $6,$ $7$ and $10$.
M1 is a benchmark case in which no jump took place and log bond prices were observed densely without error ($m=1000$, no jumps). M2 shows results for the case without jumps but for noisy and sparse samples ($m=100$, no jumps). M3 ($m=1000$, with jumps) and M4 ($m=100$, with jumps) report the results for the cases with jumps, for which in M3 the data were observed densely and perfectly and for M4 sparsely and noisy. While for Scenario $S1$ no preliminary projection of the data was conducted, Scenario $S2$ describes the customary case in which log bond prices were projected onto the first three leading eigenvalues of their empirical covariance. }\label{Tab: Monte Carlo Results1}
\resizebox{\textwidth}{!}{
\begin{tabular}{ cc  c  c c c c c c c }
\toprule
& Model& $M1$ & $M2$ & & $M3$&  &  & $M4$&  \\
\cmidrule(lr{1em}){3-3}\cmidrule(lr{1em}){4-4}\cmidrule(lr{1em}){5-7}\cmidrule(lr{1em}){8-10}
%\hline & $a$ & $50$ & $50$ &$50$ & $50$&  $50$& $50$ & $50$&  $50$\\
% &$m$, Jumps & $1000$, no jumps& $100$, no jumps& & $1000$, with jumps&  &  & $100$, with jumps&  \\
 &  trunc. level &  $l=\infty$ & $l=\infty$ & $l=3$ & $l=4$&   $l=5$&  $l=3$ & $l=4$&   $l=5$\\
\midrule
  %Truncation level& $l=\infty$ & $l=\infty$ &$l=3$ & $l=4$&  $l=5$& $l=3$ & $l=4$&  $l=5$\\
  \multirow{6}{*}{S1}  &   $rE(\Delta_n^{-2}\hat q_1^{n,10,-},IV)$    & $.26 (.23,.29)$& $.30 (.26,.35)$ & $.30 (.26,.36)$ &$.37 (.28,.56)$&$.56 (.32,.90)$&  $.35 (.29,.44)$& $.43 (.33,.62)$ &$.63 (.40,.98)$ \\
   % &   JDR   & $0$& $0$& $.67 (.50,.80)$ &$.50 (.30,.67)$&$.40 (.25,.60)$&   $.62 (.44,.80)$&  $.50 (.33,.67)$&$.40 (.25,.60)$  \\
    &   $D_{\hat q_1^{n,-}}^{\hat e}(0.85)$  & $5 (5,5)$& $6 (5,6)$& $5 (5,5)$ &$5 (5,5)$&$5 (5,5)$&$6 (5,6)$& $6 (5,6)$ &$6 (5,6)$ \\
    &  $D_{\hat q_1^{n,-}}^{\hat e}(0.90)$   & $6 (6,6)$& $6 (6,6)$& $6 (6,6)$ &$6 (6,6)$&$6 (6,6)$&  $6 (6,7)$& $6 (6,7)$ &$6 (6,7)$ \\
    &  $D_{\hat q_1^{n,-}}^{\hat e}(0.95)$   & $7 (7,7)$& $8 (7,8)$& $7 (7,7)$ &$7 (7,7)$&$7 (7,7)$&  $8 (7,8)$& $8 (7,8)$ & $8 (7,8)$ \\
    &  $D_{\hat q_1^{n,-}}^{\hat e}(0.99)$   & $9 (9,9)$& $10 (10,11)$& $9 (9,9)$ &$9 (9,9)$&$9 (9,9)$&  $11 (10,11)$& $11 (10,11)$ &$11 (10,11)$\\
      \midrule
   \multirow{3}{*}{S2} &  $rE(\Delta_n^{-2}\hat q_1^{n,10,-},IV)$    & $.83 (.73,.96)$& $.83 (.73,.96)$ & $.82 (.74,.93)$ &$.88 (.78,1.04)$&$.99 (.81,1.21)$&  $.82 (.73,.93)$& $.87 (.78,1.05)$ &$.99 (.81,1.19)$ \\
  % &    JDR   & $0$& $0$& $.67 (.50,.80)$ &$.50 (.30,.67)$&$.40 (.25,.60)$&   $.62 (.44,.80)$&  $.50 (.33,.67)$&$.40 (.25,.60)$  \\
   &    $D_{\hat q_1^{n,-}}^{\hat e}(0.99)$  & $3 (3,3)$& $3 (3,3)$& $3 (3,3)$ &$3 (3,3)$&$3 (3,3)$&$3 (3,3)$& $3 (3,3)$&$3 (3,3)$ \\
  %&    $d_{\hat q_1^n(-)}^{\hat e}(0.90)$   & $6 (6,6)$& $6 (6,6)$& $6 (6,6)$ &$6 (6,6)$&$6 (6,6)$&  $6 (6,7)$& $6 (6,7)$ &$6 (6,7)$ \\
   %&   $d_{\hat q_1^n(-)}^{\hat e}(0.95)$   & $7 (7,7)$& $8 (7,8)$& $7 (7,7)$ &$7 (7,7)$&$7 (7,7)$&  $8 (7,8)$& $8 (7,8)$ & $8 (7,8)$ \\
 %  &   $d_{\hat q_1^n(-)}^{\hat e}(0.99)$   & $9 (9,9)$& $10 (10,11)$& $9 (9,9)$ &$9 (9,9)$&$9 (9,9)$&  $11 (10,11)$& $11 (10,11)$ &$11 (10,11)$\\
      \bottomrule
\end{tabular}}
\end{table}

\section{Empirical analysis of bond market data}\label{Sec: Empirical Study}
 In this section, we apply our theory to bond market data. In particular, we investigate the influence of jumps on the estimators and the dimensions of the integrated volatility, that is, the continuous part of the quadratic covariation, to determine how many random drivers are statistically relevant. 
 
 We consider nonparametrically smoothed yield curve data from \cite{FPY2022}. For constructing smooth curves on each day, the authors of  \cite{FPY2022} use a kernel ridge-regression approach based on the theory of reproducing kernel Hilbert spaces. %Their smoothing approach yields that discount curves are twice weakly differentiable and, hence, implied forward curves are weakly differentiable. 
 We measure time in years and the data are available for approximately $250$ trading days in each year, yielding $n\approx 250$, using a day count convention in trading days, with a daily resolution in the maturity direction, where we consider a maximal time of $M=10$ years to maturity. The data are given as yields, which we first transform to zero coupon bond prices and then derive the difference returns $d_i^n(j)$ for $i=1,...,\ulT$ and $j=1,...,\lfloor M/\Delta_n\rfloor$ by formula \eqref{Difference returns}. 
 We consider data from the first trading day of the year $1990$ ($i=1$) to the last trading day of the year $2022$ ($i=33$). 
 We then derive the estimators $\hat q_{i}^n\big|_{[0,10]},$ and $\hat q_{i}^{n,10,-}$ (as defined in Remark \ref{Rem: Reduction of the time to maturity is fine}) for $i=1,...,33$ and derive the yearwise covariation kernels %the corresponding estimates for $\int_{i-1}^i q_t^c\big|_{[0,10]}dt$ and $ \int_{i-1}^i q_t^C\big|_{[0,10]}dt+\sum_{s\in [i-1,i]} q_s^J\big|_{[0,10]}$ by 
 $$\hat q_{i}^{*,-}:=\Delta_n^{-2}(\hat q_{i}^{n,10,-}-\hat q_{i-1}^{n,10,-})\quad \text{ and }\quad \hat q_{i}^*:=\Delta_n^{-2}(\hat q_{i}^n\big|_{[0,10]}-\hat q_{i-1}^n\big|_{[0,10]})$$  for $i=1,2,...,33$ and $q_{0}^n=0$. 
 The data-driven truncation rule described in Section \ref{Sec: truncation in practice} is applied for a preliminary truncation at the $0.75$-quantile of the sequence of difference return %$(\delta_{j\leq 10/\Delta_n}\tilde \Delta_{i\Delta_n}d(j\Delta_n))_{j\geq 0}$
 curves as measured in their $l^2$ norm and $d$ is chosen as the smallest value such that $d$ eigencomponents explain 90\% of the variation of the preliminary estimator.  Importantly,  the truncation rule is conducted for different $l=3,4,5$ and for each year separately and only takes into account data within the respcetive year.
 We also consider an estimator for a potential long-term volatility given by
 $$\hat q_{long}^*:= \frac 1{33}\sum_{i=1}^{33} q_i^{*,-}.$$
Under the Assumption of Section \ref{Sec: Long time Volatility estimation}, this is an estimator for a stationary volatility kernel. 
The results suggest that quadratic covariations in each year are rather complex in the sense that they exhibit a slow relative eigenvalue decay, unveil a varying shape and magnitude over time, and often differ quite substantially from measured quadratic covariations due to the existence of jumps. 
Subsequently, we provide a thorough discussion of these observations.
A table containing all results of the analysis is contained in the appendix %supplementary material.

 \subsection{Impact of jumps}
  %Table \ref{Tab: Dimensionalities} reports in the columns 10-12 the number of truncated jumps and their relative influence on the quadratic variation via the ratio $\|\hat q_i^{*,-}\|_{L^2}/\|\hat q_i^{*}\|_{L^2}$.

\begin{table}\caption{Columns $2$ to $5$ report the numbers $D_{C}^{\hat e^{*,i}}(p)$ for $C=\mathcal T_{q_i^{*,-}},$ 
defined in \eqref{Dimensionality measure} of linear factors needed in each year to explain $p=85\%, 90\%, 95\%, 99\%$ of the variation of difference returns as measured by the truncated covariation estimators $\hat q_i^{*,-}$ where the truncation rule was conducted with $l=3$ and $\hat e^{*,i}=(\hat e_1^{*,i},\hat e_2^{*,i},...)$ is the basis of eigenfunctions corresponding to the kernel $\hat q_i^*$. Columns $6$ to $9$ report $D_{C}^{\hat e^{long}}(p)$ for $C=\mathcal T_{q_i^{*,-}}$, which explain how many leading eigenvectors of the static estimator $\hat q_{long}^*$ are needed as approximating factors to explain the covariation in all years separately. }
\label{Tab: Dimensionalities} 
\begin{tabular}{ c cccccccc  }
 \toprule
 Year & \multicolumn{4}{c}{  $D_{\mathcal T_{\hat q_i^{*,-}}}^{\hat e^{*,i}}(p)$}& \multicolumn{4}{c}{$D_{\mathcal T_{\hat q_i^{*,-}}}^{\hat e^{long}}(p)$} \\
 \cmidrule(lr{1em}){2-5} \cmidrule(lr{1em}){6-9}
 &$0.85$&$0.90$&$0.95$& $0.99$ &$0.85$&$0.90$&$0.95$& $0.99$ \\
        \midrule  
% 1990      &4& 6  & 9 & 15&$5$&$7$&$10$& $15$ \\
%       1991      &5& 6  & 9 & 15&$5$&$7$&$10$& $15$ \\
%       1992      &5& 6  & 8 & 14&$6$&$7$&$9$& $15$  \\
%      1993      &4& 5  & 7 & 14&$4$&$5$&$9$& $14$ \\
%      1994      &3& 5  & 8 & 14&$3$&$5$&$9$& $15$\\
%      1995      &4& 5  & 8 & 14&$4$&$6$&$9$& $15$ \\
%      1996      &4& 5  & 8 & 13&$4$&$6$&$8$& $15$ \\
 %     1997      &3& 4  & 7 & 13&$3$&$4$&$9$& $14$ \\
%      1998      &5& 6  & 8 & 14&$5$&$6$&$10$& $15$ \\
%      1999      &3& 5  & 8 & 14&$4$&$6$&$10$& $16$ \\
 %     2000      &4& 5  & 8 & 13&$4$&$6$&$10$& $14$ \\
  %    2001      &4& 5  & 8 & 13&$5$&$6$&$9$& $13$ \\
   %   2002      &3& 5  & 7 & 12&$4$&$6$&$8$& $14$\\
    %  2003      &2& 4  & 6 & 10&$2$&$4$&$6$& $12$ \\
  %    2004      &2& 3  & 6 & 11&$2$&$4$&$7$& $13$\\
      2005      &2& 3  & 5 & 11&$2$&$3$&$8$& $12$ \\
      2006      &2& 2  & 4 & 10&$2$&$2$&$7$& $11$\\
      2007      &2& 3  & 6 & 10&$3$&$6$&$10$& $13$\\
 %     2008      & 3& 4  & 7 & 12&$3$&$5$&$9$& $13$ \\
%     2009      &3& 4  & 6 & 11&$4$&$5$&$7$& $13$ \\
 %     2010      &2& 3  & 6 & 11&$3$&$4$&$7$& $13$ \\
   %   2011      &2& 3  & 5 & 10&$2$&$4$&$6$& $12$\\
  %    2012      &1& 2  & 3 & 8 &$2$&$3$&$4$& $10$\\
 %     2013      &2& 2  & 3 & 8 &$2$&$3$&$5$& $10$\\
%      2014      &2& 2  & 4 & 10 &$2$&$3$&$6$& $11$\\
%      2015      & 2& 2  & 3 & 10 &$2$&$2$&$5$& $11$\\
%      2016      & 2& 2  & 4 & 11 &$2$&$2$&$5$& $12$\\
%      2017      &2& 2  & 5 & 11 &$2$&$3$&$6$& $12$\\
%      2018      & 2& 2  & 5 & 12 &$2$&$3$&$7$& $13$\\
 %     2019      &2& 2  & 5 & 11 &$2$&$2$&$7$& $12$\\
 %     2020      &2& 3  & 6 & 12 &$2$&$4$&$8$& $13$\\
   %     2021      &2& 2  & 4 & 10 &$2$&$3$&$6$& $12$\\
  %      2022      &2& 2  & 4 & 8&$2$&$2$&$5$& $11$\\
      \bottomrule
\end{tabular}
\end{table}
%Year       &2000& 2001  & 2002 & 2003&$2004$&$2005$&$2006$& $2007$ & 2008 & 2009 &2010 & 2011\\

 On one hand, jumps that have a moderate impact on the magnitude of the overall quadratic covariation can visually distort the shape of the volatility. Figure \ref{Fig: Volatility surfaces 2005 to 2007} 
 depicts plots of the graphs of the estimated truncated kernels $q_i^{*,-}$ (with the truncation level $l=3$) and nontruncated kernels $q^*_i$ for the years 2005, 2006 and 2007. In 2006,  which is also a year in which the yield curve inverted before the financial crisis in 2007, two jumps %(corresponding to the dates in August when the FED raised short term interest) 
 had a visible impact on the shape of the measured quadratic covariation kernels although they together accounted for less than 3 \% of the magnitude of the quadratic covariation. This is due to a higher emphasis on the variation in difference returns with short maturities where one should note the different scalings in the plots.
 Removing these two jumps leads to a more time-homogeneous shape of the integrated volatility surfaces in the sense of the relation of the variation in the shorter maturities to the variation in higher maturities.
 On the other hand,  jumps influence the magnitude of the quadratic variation.  For instance, 
  nine increments in the year 2020 (Covid-19 outbreak) sorted out by the truncation rule for $l=3$ accounted for more than 50\% of the overall quadratic variation in the data as measured by its norm. The  statistics for jumps in all years can be found in the supplment to this article.  %At the same time, the truncation rule does not necessarily pick outlyiers with the largest impact on the magnitude of the measured variations. This can be seen by the fact that the 2 increments sorted out first in 1994 account for 4\% of the magnitude of the quadratic variation, whereas the increment that was just picked with a stricter truncation ($l=3$) accounts for 11\% of the magnitude of the quadratic variation. 
 Interestingly, our measurements suggest that jumps tend to cluster.
 
% From an economic perspective, the jump occurrences suggest that the years 1994 (bond market crisis), 1997 and 1998 (currency crises), 2001 (dot-com burst), 2006 (right before the global financial crisis), and 2021 (Covid-19 outbreak) are noteworthy, which appear together with (or briefly predate) important economic events. 

\begin{figure}
\begin{subfigure}{.3\linewidth}
  \centering
  \includegraphics[width=\linewidth]{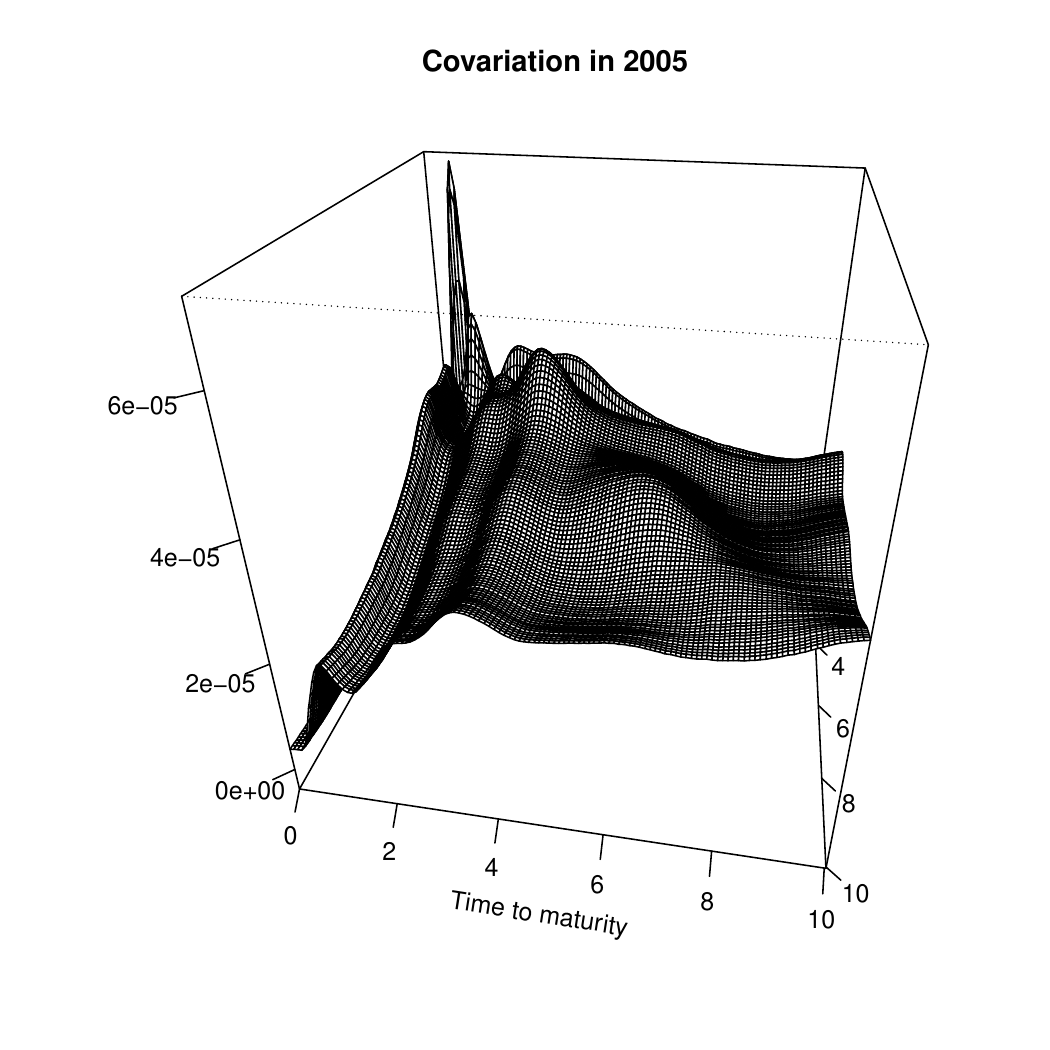}
\end{subfigure}\hfil
\begin{subfigure}{.33\linewidth}
  \centering
  \includegraphics[width=\linewidth]{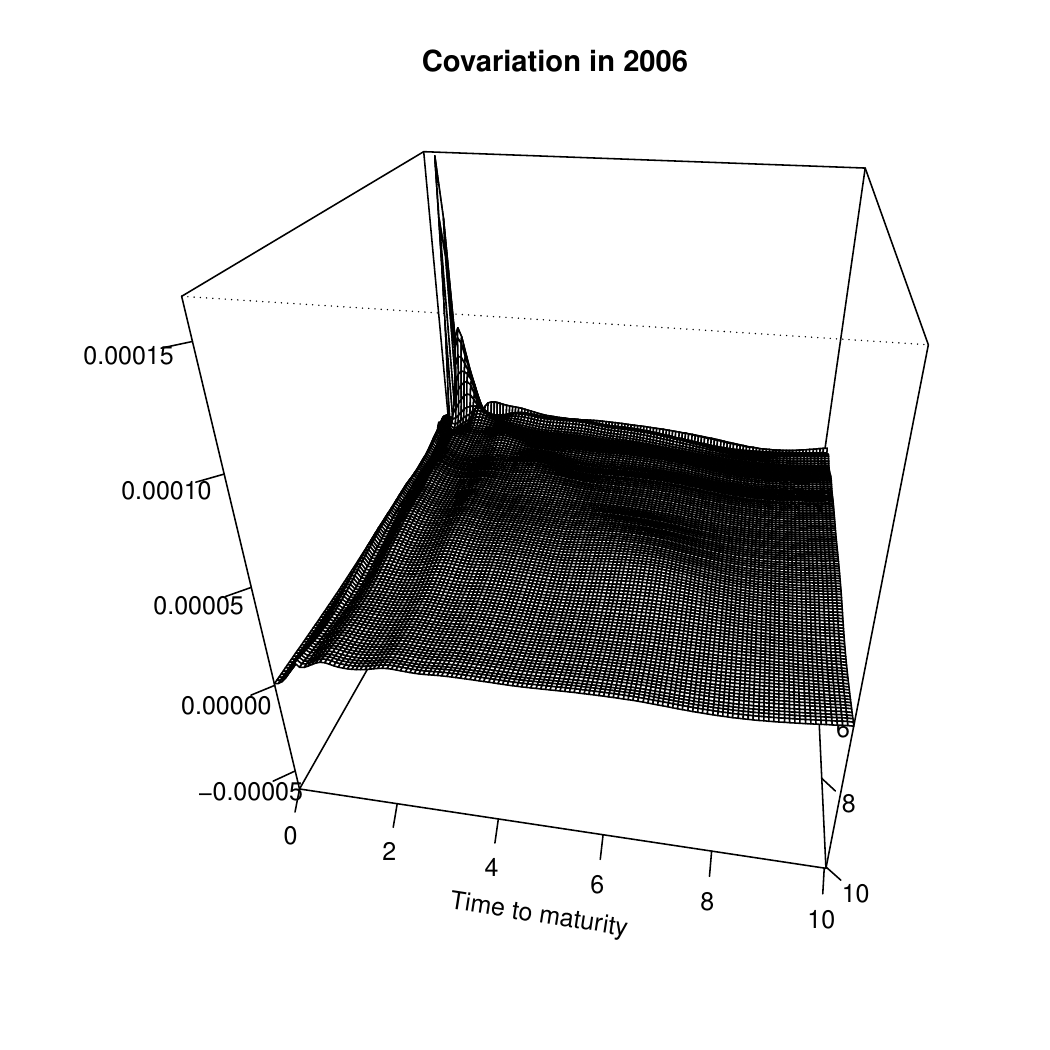}
\end{subfigure}\hfil
\begin{subfigure}{.33\linewidth}
  \centering
  \includegraphics[width=\linewidth]{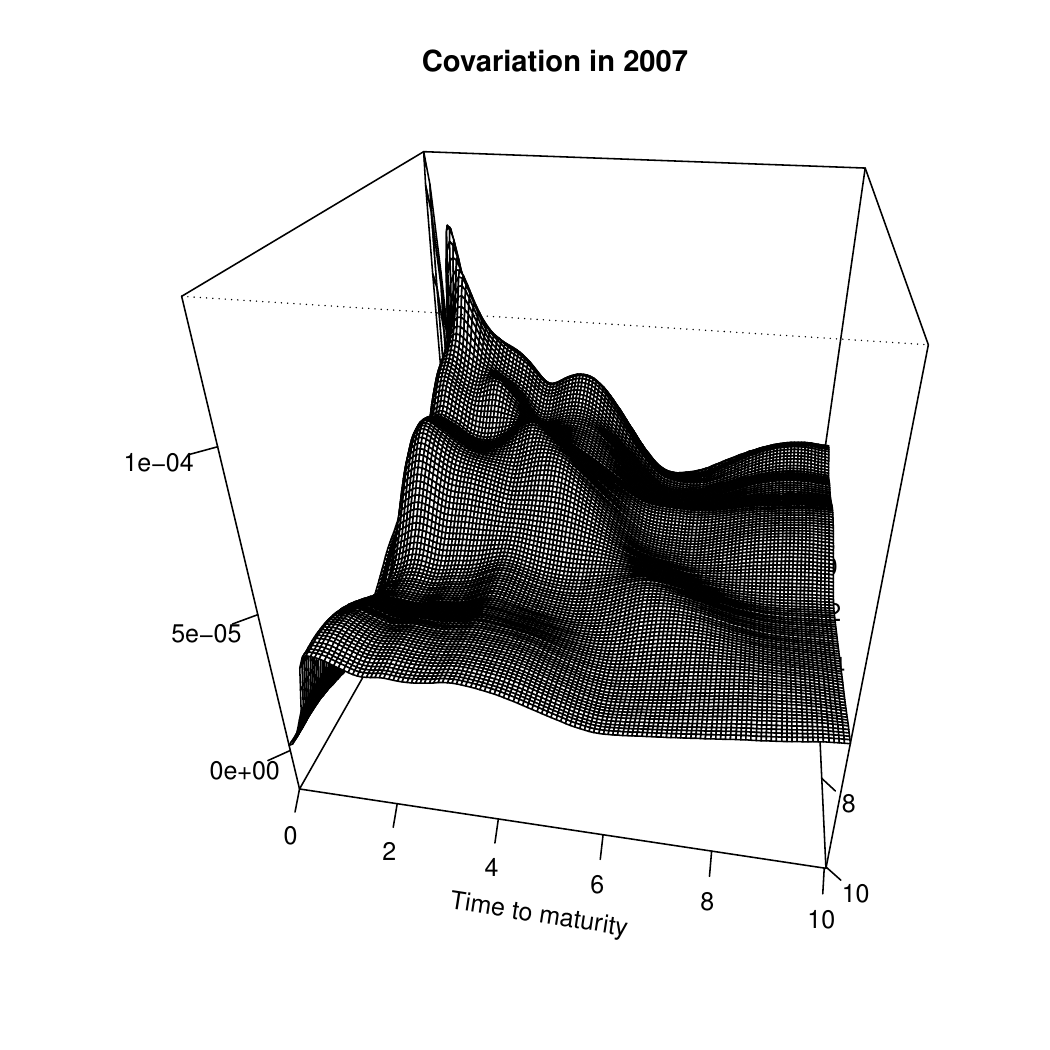}
\end{subfigure}\hfil
\begin{subfigure}{.33\linewidth}
  \centering
  \includegraphics[width=\linewidth]{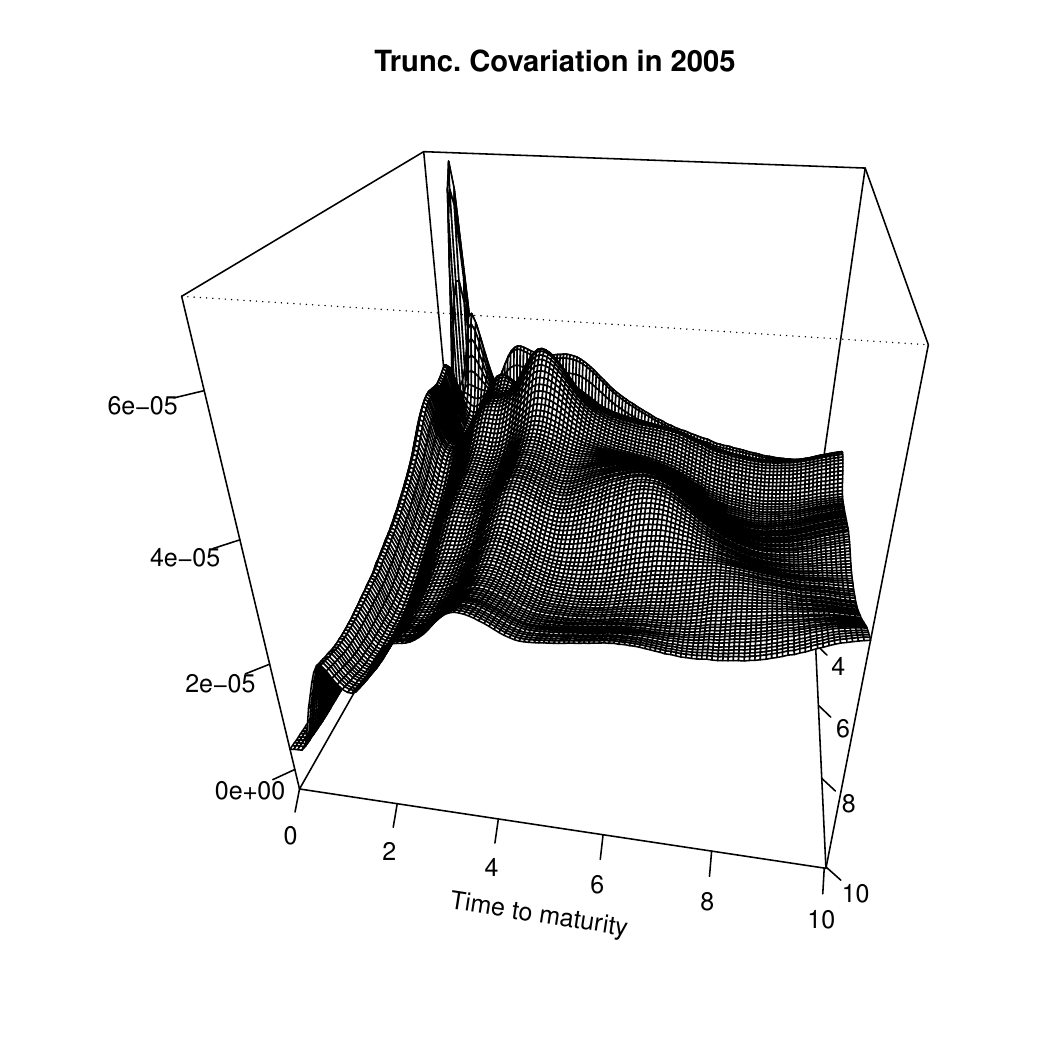}
\end{subfigure}\hfil
\begin{subfigure}{.33\linewidth}
 \centering
  \includegraphics[width=\linewidth]{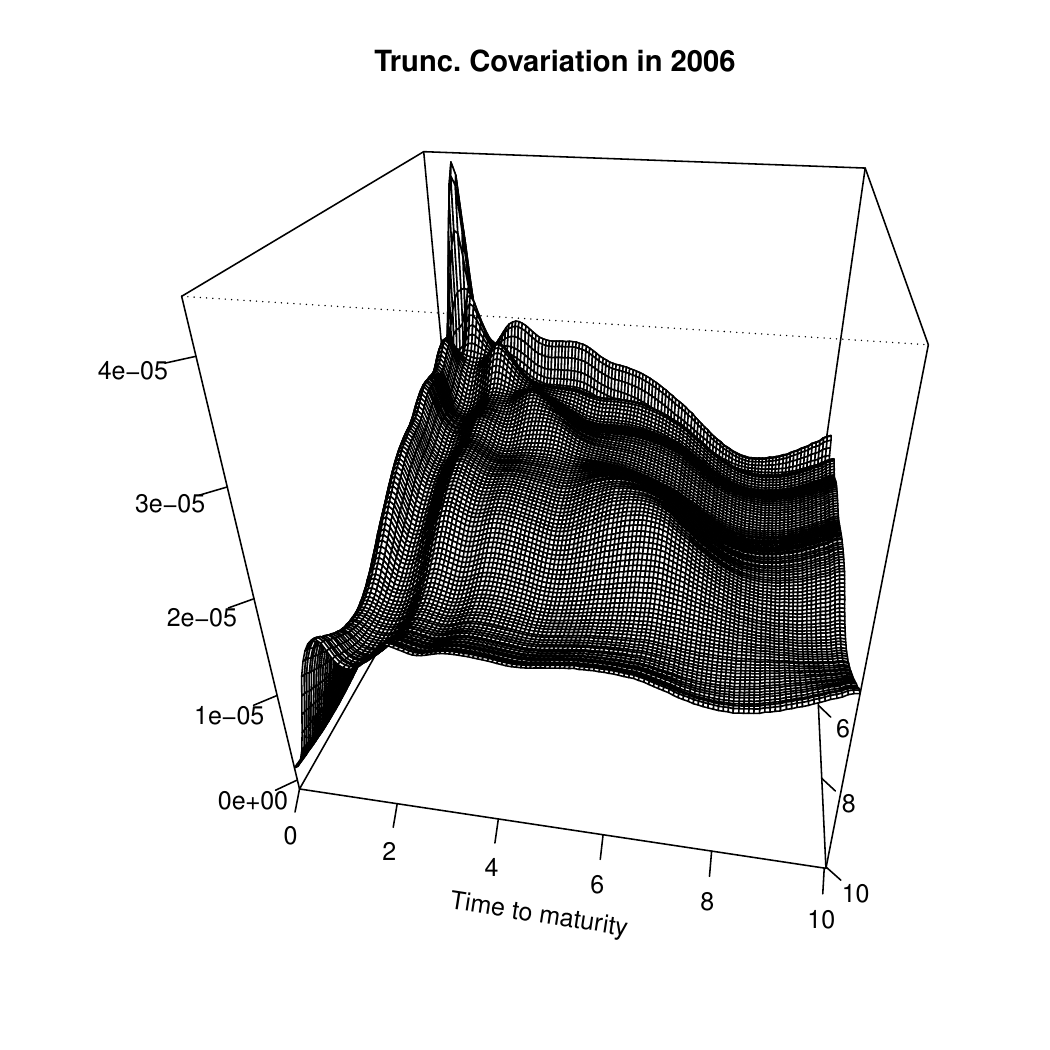}
\end{subfigure}\hfil
\begin{subfigure}{.33\linewidth}
  \centering
  \includegraphics[width=\linewidth]{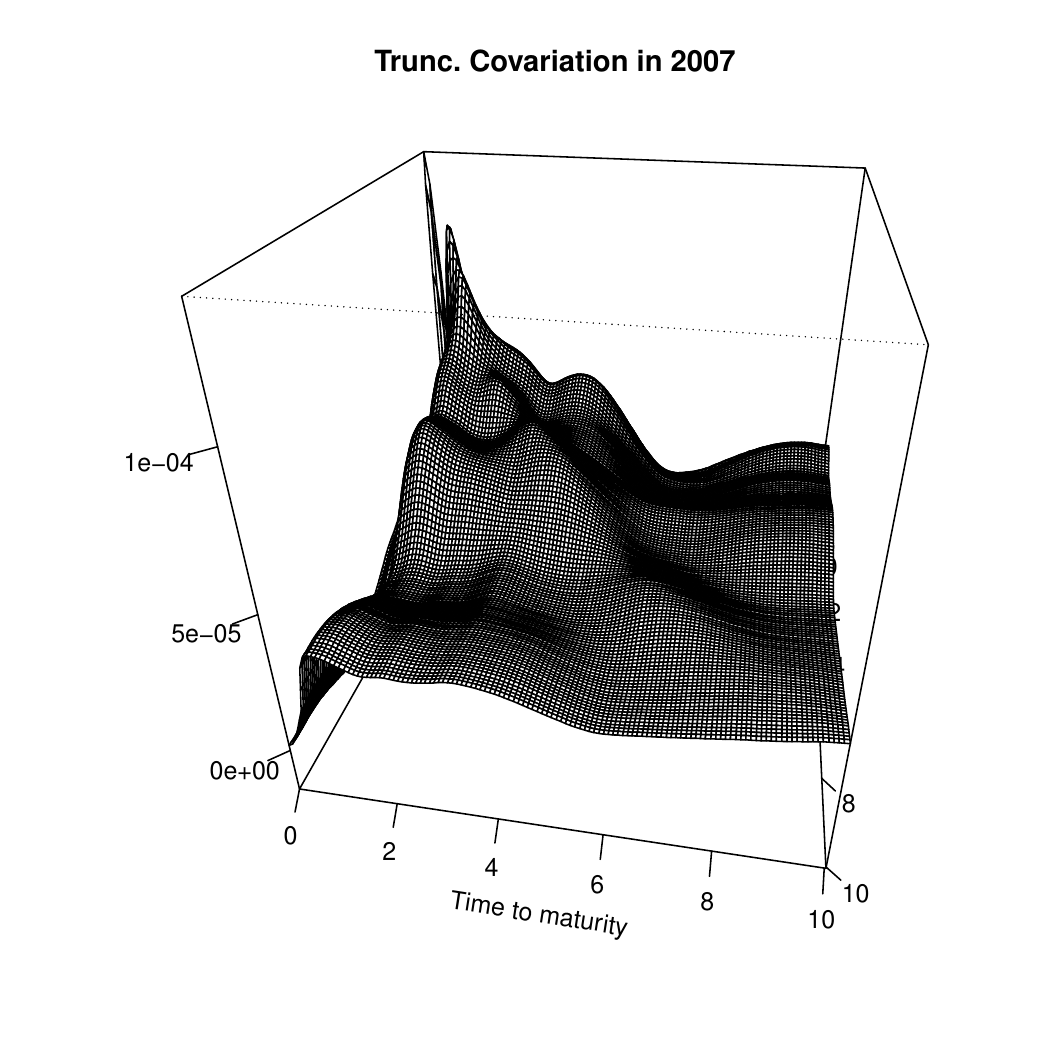}
\end{subfigure}\caption{The upper row shows the graphs of the estimators $q_i^*$ for the years 2006, 2007, and 2008 ($i=17,18,19$) and the lower row their truncated counterparts $q_i^{*,-}$ at a truncation level of $l=3$.}
\end{figure}\label{Fig: Volatility surfaces 2005 to 2007}

\subsection{Dimensionality of the continuous part of the quadratic variations}
%Term structure modeling is often preceded by a preprocessing step which determines the number of random drivers needed to adequately describe the covariation in bond market data. 
%We can reconsider this question here from a slightly different and specific angle. Namely, w
We now examine  the statistically relevant number of random processes that are driving the continuous part of the forward curve dynamics by investigating the dimensionality of the continuous quadratic covariation of the latent driving semimartingale via the estimator $\mathcal T_{\hat q_i^{*,-}}$. %An advantage is that we do not need to make any assumptions neither on the cross-sectional nor the temporal dependence structure of forward curve evolutions.

Table \ref{Tab: Dimensionalities} reports the number of eigenfunctions of $\mathcal T_{\hat q_i^{*,-}}$ that are needed to explain resp. $85\%$, $90\%$, $95\%$, and $99\%$ of the continuous covariation in the years 2005, 2006 and 2007 showing that to explain 99\% of the variation at least 10 factors are needed in each year. The situation looks similar for all other years from 1990-2022,  while the detailed results were relegated to the appendix. %supplementary material.
 %That is, if $\hat \lambda_j^i$ is the $j$'th largest eigenvalue of $\hat q_i(-)$ the columns report the smallest number $d\in \mathbb N$ such that
%$(\sum_{i=1}^d \hat \lambda_j^i)/(\sum_{i=1}^{\infty} \hat \lambda_j^i) \geq q$
%for $q=0.85, 0.9, 0.9, 0.99$.
We find that the complexity of the covariation seems to have decreased over the years, indicating a time-varying pattern of the volatility term structure that goes beyond its overall level. 
It is  noteworthy that in almost every year (30 out of 33), the number of linear factors needed to explain at least 99\% of the truncated variation of the data is at least $10$. 
These dimensions even increase if we employ the static factors $(e_j^{long})_{j\in \mathbb N}$ of eigenvectors of $\mathcal T_{\hat q_{long}^*}$  and do not update them in each year. In that case, in 27 out of 33 years at least 12 factors are needed to explain at least $99\%$ of the variation in each year.

%These values give a lower bound on how  many factors are needed to explain a certain amount of variation in each year separately. However, as the factors could change from year to year, this does still not reveal how many factors are needed to describe the variation in the data over the course of the whole 30 years considered. The linear factors that capture most of the variation in all data are the leading eigenfunctions of the long-term estimators $\hat q_{long}^*$. %Columns six to nine of Table \ref{Tab: Dimensionalities} report the number of these long term factors needed to explain resp. 85\%, 90\%, 95\% and 99\% of the variation in the data in each year separately. 
%If $\hat \Sigma_i (-)$ defines the integral operator corresponding to the kernel $\hat q_i(-)$ and defining the factors $\hat \lambda_j^{i,long}:= \langle \hat e_j^{long}, \hat \Sigma_i(-) \hat e_j^{long}\rangle_{L^2}$, columns six to nine report the smallest natural number $d$ such that
%$(\sum_{i=1}^d \hat \lambda_j^{i, long})/(\sum_{i=1}^{\infty} \hat \lambda_j^i) \geq q$
%for $q=0.85, 0.9, 0.9, 0.99$. Of course, the reported dimensions must be even higher than for eigenfunctions of the operators $\hat q_i(-)$ themselves, which are optimal in this regard. On average around 1.58 additional factors are needed to explain at least 99\% of the variation in each year when we use static factors. 

\subsection{Importance of higher-order factors for short term trading strategies}\label{Sec: Importance of higher order factors}
%Low dimensional factor approximations for term structure models are used to price bond related derivatives or portfolios or measure incicated risk. These factors are usually derived by a principle component analysis of yield curve or log-bond price differences, which typically indicates that a low number of factors is needed to explain most of the variation in the data. That this rational does not carry over to daily difference returns was suggested by the analysis of realized volatilities in the previous section. 
%An important question is if the high-dimensional cross-covariance structures indicated in the previous section are of any economic meaning.
%To find that out,
A natural question is if the higher-order factors 
indicated by the analysis of real bond market data in Section \ref{Sec: Empirical Study}
are of economic significance beyond capturing variation in difference returns. Therefore, we investigate whether the high dimensionality of the continuous quadratic varitations indicated by the estimators $\hat q_i^*$ and $q^*_{long}$ are important for other short term trading strategies than difference returns.
Precisely, Define the daily return $(d_L)_{i}^n(j)$ of the trading strategy of buying an $(j+L)\Delta_n$ bond and shorting an $j\Delta_n$ bond
$$(d_L)_{i}^n(j):= \sum_{l=0}^{L}d_i^n(j+l)=\tilde\Delta^n_{i\Delta_n} \log P((j+L)\Delta_n)-\tilde\Delta_{i\Delta_n}^n \log P(j\Delta_n)$$
where $\tilde\Delta_t^n \log P(x)=\log P_{t+\Delta_n}(x)-\log P_{t+\Delta_n}(x-\Delta_n)$. Evidently, we can derive them as linear functionals of either log bond price returns or difference returns. 

We want to 
determine the adequacy of approximation of  these higher-order difference returns when they are derived either from approximated log price curves, which are projected onto its leading principal components  or when they are derived from difference returns, which are projected onto the leading eigencomponents of the long-term volatility estimator.
For that, we calculate the 
relative mean absolute error ($RMAE$) for a set $\mathcal V=\{i_1 \Delta_n^{val},...,i_{825}^{val}\Delta_n\}$ of dates (where in each year from 1990 to 2022 we randomly draw 25 dates making a total of 825 validation dates). That is,  defining the piecewise constant kernels $(\tilde d_L)_{i}^n=\sum_{j=1}^{\lfloor M/\Delta_n\rfloor}(d_L)_i(j)\indicator_{[j\Delta_n,j\Delta_n)}$
we calculate 
$$RMAE_L(f_1,...,f_d):=\frac 1{825}\sum_{l=1}^{825}\frac{\left\|(\tilde d_L)_{i_l}^n-\mathcal P_{f_1,...,f_d}(\tilde d_L)_{i_l}^n \right\|_{L^2(0,10)}}{\left\|(\tilde d_L)_{i_l}^n\right\|_{L^2(0,10)}}.%=\frac 1{825}\sum_{l=1}^{825}\frac{\sqrt{\int_0^{10}\left((I-P_{f_1,...,f_d})\tilde \Delta_{i_l\Delta} d_L(x) \right)^2 dx}}{\sqrt{\int_0^{10}\left(\tilde \Delta_{i_l\Delta} d_L(x)\right)^2 dx}}
$$
where $\mathcal P_{f_1,...,f_d}:= \sum_{i=1}^d f_i^{\otimes 2}$.
The factors are derived in two different ways. In the first scenario (S1), the factors $f_1,...,f_d$ correspond to principal components of the empirical covariance of log-price differences $\tilde \Delta \log P_{i\Delta_n}$ for $i\notin \mathcal V$ and in the second scenario (S2) the factors $f_1,...,f_d$ correspond to the leading eigenfunctions of the estimated stationary volatility kernel $\hat q_{long}^*$ %where also all increments in $\mathcal V$ were truncated previously 
where as before the truncation of jumps is conductcted yearwise with truncation level $l=3$ according to the truncation procedure described in Section \ref{Sec: truncation in practice}.

%While for $L = 0$ the lagged difference returns correspond to difference returns, for higher $L$ more and more overlapping difference returns are needed to derive a lagged difference return. The more overlapping terms there are, the higher the mechanically induced high local correlation, which was suspected to be responsible for low-dimensional structure of yield or discount curve returns in \cite{Crump2022}.
We compare the results for lags of 7, 30, 90 and 180 days,  since they approximately correspond to the returns of buying a bond and shorting another bond with time to maturity that is resp. a week, a month, a quarter or half a year higher. The $RMAE$s can be found in Table \ref{Tab: MRRISEs}. 
\begin{table}\caption{The table shows the $RMAE$s for different numbers of factors lags and for the two different ways (S1 and S2) in which the factors are derived. In bold are errors for which 99\% of the variation in difference returns resp. the log-bond prices were explained.
}\label{Tab: MRRISEs} 
\resizebox{\textwidth}{!}{
\scriptsize{
\begin{tabular}{ cc cccccccccccccccc}
 \toprule
 \multirow{2}{*}{Lag}& \multirow{2}{*}{Scenario} &  \multicolumn{16}{c}{ $d$} \\
  \cmidrule(lr{1em}){3-18} 
 & &$1$&$2$& $3$   &$ 4$   & $5$ & $6$ & $7$ & $8$ & $9$ & $10$ & 11 & 12 & 13 & 14 & 15 & 16\\
       \midrule  
      \multirow{2}{*}{$Lag = 7$} & S1    & $0.63$ & ${\bf 0.50}$ &$0.44$ & $0.40$  & $0.36$ & $0.32$ & $0.30$&  $0.27$& $0.25$& $0.23$&$0.21$&$0.19$&$0.17$ &$0.16$  &$0.14$&$0.13$\\
         & S2    & $0.62$ & $0.50$ &$0.45$ & $0.40$  & $0.35$ & $0.32$ & $0.28$&  $0.25$& $0.23$& $0.20$&$0.17$&$0.15$&$0.12$ &${\bf 0.10}$ & $0.09$& $0.07$\\
          
      \multirow{2}{*}{$Lag = 30$} & S1    & $0.63$ & ${\bf 0.50}$ &$0.44$ & $0.40$  & $0.35$ & $0.32$ & $0.29$&  $0.27$& $0.24$& $0.22$&$0.20$&$0.18$&$0.17$ &$0.15$  &$0.14$&$0.12$\\
        & S2    & $0.62$ & $0.50$ &$0.45$ & $0.40$  & $0.35$ & $0.31$ & $0.28$&  $0.25$& $0.22$& $0.20$&$0.17$&$0.15$&$0.12$ &${\bf 0.10}$ & $0.08$& $0.07$\\
      
     \multirow{2}{*}{$Lag = 90$}& S1    & $0.62$ & ${\bf 0.49}$ &$0.43$ & $0.38$  & $0.34$ &$ 0.30$ & $0.27$& $ 0.25$& $0.22$& $0.20$&$0.17$&$0.15$&$0.17$ &$0.14$  &$0.12$&$0.12$\\
       & S2    & $0.61$ & $0.49 $&$0.43$ & $0.38$  & $0.33$ & $0.29$ & $0.26$&  $0.23$& $0.19$& $0.17$&$0.16$&$0.13$&$0.11$ &${\bf 0.09}$ & $0.07$& $0.06$\\
      
         \multirow{2}{*}{$Lag = 180$} & S1    & $0.61$ &$ {\bf 0.47}$ &$0.41 $& $0.36$  & $0.31 $& $0.27$ & $0.24$&  $0.21$& $0.18$& $0.15$&$0.13$&$0.11$&$0.09$ &$0.07$  &$0.06$&$0.05$\\
       & S2    & $0.60$ & $0.47$ &$0.41$ & $0.35 $ & $0.30$ & $0.26$ & $0.22$&  $0.19$& $0.15$& $0.13$&$0.13$&$0.10$&$0.08$ &${\bf 0.07 }$& $0.06$& $0.05$\\
      \bottomrule
\end{tabular}}
}
\end{table}
It can be observed that a high number of factors is needed to approximate the lagged difference returns precisely and that approximations based on low factor structures as indicated by the covariance of difference returns imply high approximation errors. While it is not surprising that the approximation gets better if we use more factors, the high discrepancy of the approximation errors is noteworthy. The errors for a typically chosen three factor model based on log price differences (the factors correspond to level, slope and curvature), which explain more than 99,7\% of the variation in log-price returns is for all lags higher than $0.4$, whereas for  the approximation error for 14 factors, which we would need to explain 99\% of the variation in difference returns as measured by $\hat q_{long}^*$ is never higher than $0.11$. 
Interestingly, for all lags, choosing the factors equal to the leading eigenfunctions of the long term volatility $\hat q_{long}^*$ instead of the ones indicated by log-price differences can reduce the error for the higher-order approximations quite significantly and for $d=16$ and for lags not higher than 90 days by almost 50 \%. %This effect to vanishes for for the highest lag of $180$ days, which resembles the aggregating effect of overlapping difference returns.
%This highlights the benefits of the interpretation of volatility as an instantaneous covariance.
Higher-order factors of volatility can, thus, not easily  be ignored and might carry important economic information.

\subsection{Concluding remarks on the empirical study}

We conclude that  the reported dimensions are overall quite high compared to the few factors needed to explain a large amount of the variation in yield and discount curves. This suggests that low-dimensional factor models are not able to capture all statistically relevant codependencies of bond prices. 
Still,  exact magnitudes of explained variations of the higher order components have to be interpreted cautiously and  conditional on the smoothing technique that was employed to derive yield or discount curves. 
%This conditional interpretation is inherent to most arbitrage free term structure models, whether it is the form of nonparametric smoothing or the precise form of the employed parametric model that influences the 
 %Further research on the effect of different smoothing techniques on the measured dimensionalities of difference return covariances are . 
However,  higher-order factors seem to be economically relevant for capturing variations in short term trading strategies as indicated by
 the out-of-sample study of Section \ref{Sec: Importance of higher order factors}.

Underestimation of the number of statistically relevant random drivers can have undesirable effects. For instance, \cite{Crump2022} showcase the potential economic impact on mean-variance optimal portfolio choices and  hedging errors. % In Section \ref{Sec: Importance of higher order factors} %supplementary material to this article,  more evidence for the importance of the higher-order factors in capturing returns of short-term trading strategies was provided.  
At the same time, not every model that is parsimonious in its parameters needs to
entail a low-dimensional factor structure  such as the simple volatility model of Section \ref{Sec: Simulation Study}.
It seems desirable to derive parsimonious models that match the empirical observation of high or infinite-dimensional covariations and reflect the characteristics of their dynamic evolution.

\section*{Acknowledgements}
I would like to thank Dominik Liebl, Fred Espen Benth, Alois Kneip and Andreas Petersson
for helpful comments and discussions. Funding by the 
Argelander program of the University of Bonn is gratefully acknowledged.
%Almut Veraart
\bibliographystyle{plain}
%\bibliography{bibliography}

\begin{appendix}

\section{It{\^o} semimartingales in Hilbert spaces}\label{Sec:  semimartingales in Hilbert spaces}

In this  appendix,  we provide an introduction and technical details for the class of  It{\^o} semimartingales that we consider throughout the paper.

First, we specify the components of the driver $X$  which is an $L^2(\mathbb R_+)$-valued  right-continuous process with left-limits (c{\`a}dl{\`a}g) that can be decomposed as 
$$X_t:= X_t^C+J_t:=(A_t+M_t^C)+J_t\quad t\geq 0.
    $$
Here, $A$ is a continuous process of finite variation, $M^C$ is a continuous martingale and $J$ is another martingale modeling the jumps of $X$. We assume that $X$ is an It{\^o} semimartingale for which the components have integral representations
\begin{equation}
    A_t:=\int_0^t \alpha_s ds,\quad M_t^C:=\int_0^t\sigma_s dW_s, \quad J_t:=\int_0^t\int_{H\setminus \{0\}}\gamma_s(z) (N-\nu)(dz,ds). 
\end{equation}
For the first part $(\alpha_t)_{t\geq 0}$ is an $H$-valued and and almost surely integrable (w.r.t. $\|\cdot\|_{L^2(\mathbb R_+)}$) process that is adapted to the filtration $(\mathcal F_t)_{t\geq 0}$.  
%We remark, that under arbitrage-free dynamics, that is, under an equivalent local martingale measure  the drift is necessarily a deterministic function of $\sigma$  and $\gamma$ (c.f. \cite{bjork1997}), %that is
%\begin{equation}\label{HJM drift condition}
%\alpha_t(x)=\alpha_{HJM}(\sigma_t,\gamma_t)(x)= \sum_{j=1}^{\infty} \sigma_t e_j (x)\int_0^x\sigma_t e_j (y)dy-\int_{H\setminus\{0\}} \gamma_s(z)\left(e^{-\int_0^x \gamma_s(z)(y)dy}-1\right)F(dz)
%\end{equation}
%which was in the continuous case the original inside leading to the popular Heath-Jarrow-Morton framework of \cite{HJMoriginal} for pricing bonds and interest rate sensitive contingent claims. However, this is not of particular importance to us, since the drift vanishes asymptotically in our limit theory.

The volatility process $(\sigma_t)_{t\geq 0}$ is predictable and takes  values in the space of Hilbert-Schmidt operators $L_{\text{HS}}(U,L^2(\mathbb R_+))$ from a separable Hilbert space $U$ into $L^2(\mathbb R_+)$. Moreover, we have  $\mathbb P[\int_0^T\|\sigma_s\|_{\text{HS}}^2ds<\infty]=1$. The space $U$ is left unspecified, as it is just formally the space on which the Wiener process $W$ is defined and does not affect the distribution of $X$.  The  cylindrical Wiener process $W$  is a weakly defined Gaussian process with independent stationary increments and covariance $I_U$, the identity on $U$.  One might consult the standard textbooks \cite{DPZ2014}, \cite{mandrekar2015} or \cite{PZ2007} for the integration theory w.r.t.  $W$.
     %Importantly, as long as the variation of $M^C$ is absolutely continuous w.r.t. the Lebesgue measure, such a representation for $M^c$ can always be found (c.f. Section 2.2.5 in \cite{mandrekar2015}). 
  
    For the jump process $J$,  we define a homogeneous Poisson random measure $N$ on $\mathbb R_+\times H\setminus \{0\}$ and its compensator measure $\nu$ which is of the form $\nu(dz,dt)=F(dz)\otimes dt$ for a $\sigma$-finite measure $F$ on $\mathcal B(H\setminus \{0\})$. The process $\gamma_s(z))_{s\geq 0, z\in H\setminus \{ 0\}}$  is the  $l^2(\mathbb R_+)$-valued jump volatility process and is predictable and stochastically integrable w.r.t. the compensated Poisson random measure $\tilde N:=(N-\nu)$. For a detailed account on stochastic integration w.r.t. compensated Poisson random measures in Hilbert spaces, we refer to \cite{mandrekar2015} or \cite{PZ2007}.
%Summing up,  we have 
%$$X_t=\int_0^t \alpha_s ds+\int_0^t \sigma_sdW_s
%+\int_0^t\int_{H\setminus \{0\}}\gamma_s(z) (N-\nu)(dz,ds)\quad t\geq 0.\notag%+\left(\mathcal S\ast \gamma \indicator_{\|\gamma\|< 1}\right)\star (N-\nu)_t%+\left(\mathcal S\ast \gamma \indicator_{\|\gamma\|\geq  1}\right)\star N_t
 %   $$
   
    Let us now rewrite the quadratic covariation \eqref{Quadratic variation abstract probability limit} of $X$  in terms of the volatility $\sigma$ and the jumps of the process  as 
\begin{align}\label{semimartingale quadratic variation}
    [X,X]_t= [X^C,X^C]_t+[J,J]_t= \int_0^t \Sigma_s ds+ \sum_{s\leq t} (X_s-X_{s-})^{\otimes 2},
\end{align}
where $\Sigma_s = \sigma_s\sigma_s^*$ (where $\sigma_s^*$ is the Hilbert space adjoint) and  
$X_{t-}:= \lim_{s\uparrow t}X_s$ is the left limit of $(X_t)_{t\geq 0}$ at $t$, which is well-defined, since $X_t$ has c{\`a}dl{\`a}g paths. This characterization follows as a special case of Theorem 3.1 in \cite{Schroers2024}

  Let us now reconsider Example \ref{Ex: CPP in Hilbert space}.
  
     \begin{example}[Rewriting an $L^2(\mathbb R_+)$-valued Poisson random measure in compensated form]\label{Ex: Compensate CPP}
In Example \ref{Ex: CPP in Hilbert space} it was remarked  that a compound Poisson process $J_t=\sum_{i=1}^{N_t} \chi_i$ is strictly speaking not a valid choice for the jump process, since it is not a martingale. Here we show that the semimartingale in the example can be easily rewritten to have the desired form:
    For that, define the Poisson random measure
    $N(B,[0,t]):= \#\{i\leq N_t: \chi_i\in B\} $ for $B \in \mathcal B(H\setminus \{0\}), t\geq 0.$
    This has compensator measure $\nu= \lambda dt\otimes F(dz)$, so   we can redefine $J$ in a formally correct manner by
    $J_t=\sum_{i=1}^{N_t} \chi_i-\lambda t \mathbb E[\chi_1]=\int_0^t \int_{L^2(\mathbb R_+)    \setminus \{0\}} z (N-\nu)(dz,ds)$ and set $A_t=(a+\lambda \mathbb E[X_1])t$.  %This means that exchanging  $J_t$ and $\tilde J_t$ in the formula for $f_t$ is valid for our limit theory.
  
%     \item[(i)][Hawkes process] This is  similar to the previous example, but in comparison, we have a time-dependent intensity. That is, 
 %    $$J_t:=\sum_{i=1}^{N_t^H} X_i,$$
%      and $N_t^H$ is a point process such that its intensity satisfies 
%$$\lambda(s)=\mu+\int_0^t \phi(t-s) dN_t^{Hawkes}$$
%for a function $\phi:\mathbb R_+\to\mathbb R_+$ such that $\int_{\mathbb R_+}\phi(s)ds<1$. If we define
%$$M_t^D:= J_t-\int_0^t \lambda(s)ds\mathbb E[X_1],$$
%this is a martingale and, hence, has a Jacod-Grigelionis representtion of the form ...
 %Hawkes processes in infinite-dimensions are hardly found in the literature and as already in the multivariate case they appear in various forms as e.g. with vector valued marks (as in our case) or even with multivariate intensities (see EMBRECHTS). Here we stick to the very simple case of a univariate intensity. 
%\end{itemize}
\end{example}

\section{Technical Assumptions}
This section contains the technical Assumptions that are needed for the validity of Theorems \ref{T: General discrete LLN}, \ref{T: General Limit discretized truncated LLN},  \ref{T: Rate of convergence for discretized estimator},  \ref{T: CLT for truncated estimator} and \ref{T: Long-time asymptotics for termstructure volatiltiy}.  

\subsection{Assumption for derivation of  idenifiability of $[X^C,X^C]$ and $[J,J]$ }\label{Sec: Identifiability assumptions for cont and disc}

To derive asymptotic results for $\hat q^{n,-}_t$ in Theorem \ref{T: General Limit discretized truncated LLN}, we introduce
\begin{assumption}[r]\label{As: H}
    $\alpha$ is locally bounded, $\sigma$ is c{\`a}dl{\`a}g and there is a localizing sequence of stopping times $(\tau_n)_{n\in \mathbb N}$ and for each $n\in \mathbb N$ a real valued function $\Gamma_n:H\setminus \{0\}\to \mathbb R$ such that $\|\gamma_t(z)(\omega)\|\wedge 1\leq \Gamma_n(z)$ whenever $t\leq \tau_n(\omega)$ and $\int_{L^2(\mathbb R_+)\setminus \{0\}}\Gamma_n(z)^rF(dz)<\infty$.
\end{assumption} 
Assumption \ref{As: H} used in Theorem \ref{T: General Limit discretized truncated LLN}  is a direct generalization of Assumption (H-r) in \cite{JacodProtter2012}.  It  implies that for $r<2$,  the jumps of the process are $r$-summable, that is, we have
$$\sum_{s\leq t}\|X_s-X_{s-}\|^{l}<\infty\quad \forall l>r.$$

\subsection{Assumption for derivation of convergence rates}\label{Sec: Assumptions for Error bounds}

For the derivation of convergence rates in Theorem \ref{T: Rate of convergence for discretized estimator}, observe that, since $\Sigma_t=\sigma_t\sigma_t^*$ is for each $t\geq 0$ a Hilbert-Schmidt operator, we can find a process of kernels 
\begin{equation}\label{eq: notation for volatility kernel}
q_t^C,\text{ such that }\Sigma_t= \mathcal T_{q_t^C}\quad \forall t\geq 0.
\end{equation}
It is seems natural to impose H{\"o}lder-regularity assumptions on the volatility kernel $q_t^C$  for $t\geq 0$ to derive the error bounds.
%To find upper bounds for the convergence rate of %the truncated estimator, it seems natural to impose further assumptions in the sense of H{\"o}lder regularity on the volatility kernel $q_t^C$ for every $t\geq 0$.
For instance, one might consider a H{\"o}lder continuous volatility kernel, such that $q_t^C\in C^{\gamma}(\mathbb R_+^2)$ for
$$C^{\gamma}(\mathbb R_+^2):=\left\{q:\mathbb R_+^2\to \mathbb R:  %\|q\|_{C^{\gamma}([0,M]^2)}:=
\sup_{x,y,x',y'\leq M}\frac {|q(x,y)-q(x',y')|}{\|(x,y)-(x',y')\|_{\mathbb R^2}^{\gamma}}<\infty\quad \forall M\geq 0\right\}.$$
%Indeed, we essentially have that if $q_t^C$ has finite $\gamma$-H{\"o}lder norms, which are square-integrable in time  $s$, we obtain a rate of convergence $\mathcal O_p(\Delta_n^{\gamma})$ (c.f. Corollary ... below).
However, we can consider weaker regularity conditions, which do not necessarily assume the kernels to be continuous.
Namely, we require $q_t^C\in \mathfrak F_{\gamma}$ where
\begin{align*}
   \mathfrak F_{\gamma}
:= & \left\{q \in L^2(\mathbb R_+^2): \|q\|^2_{\mathfrak F_{\gamma}(\mathbb R_+^2)}:=\sup_{r>0}\int_{\mathbb R_+^2} \frac{(q(r+x,y)-q(x,y))^2}{r^{2\gamma}}dxdy<\infty\right\}.
\end{align*}
%\begin{remark}\label{rem: Favard relation to the regularity condition}
%    The space $\mathfrak F_{\gamma}$ is related to the $\gamma$-Favard space $\mathfrak F_{\gamma}^{\mathcal S}$ for the semigroup of left shifts on $L^2(\mathbb R_+)$ (c.f. \cite{Engel1999} for a detailed discussion of Favard spaces).
%Recall that if $H$ is an arbitrary separable Hilbert space and $\mathcal S:=(\mathcal S(t))_{t\geq 0}$ a strongly continuous semigroup of linear operators on $H$ which is bounded (i.e. $\sup_{t\geq 0}\|\mathcal S(t)\|_{\text{op}}<\infty$), then the $\gamma$-Favard space corresponding to $\mathcal S$ is defined as 
%$$\mathfrak F_{\gamma}^{\mathcal S}:=\{h\in H: \|h\|_{\mathfrak F_{\gamma}^{\mathcal S}(H)}:=\sup_{r>0}\frac {\|(\mathcal S(r)-I)h\|_H}{r^{\gamma}}<\infty\}.$$
%The operation $\mathfrak S_t A:=\mathcal S(t) A$ for $A\in L_{\text{HS}}(L^2(\mathbb R_+))$ also defines a strongly continuous semigroup on $L_{\text{HS}}(L^2(\mathbb R_+))$ (where $\mathcal S(t) A f=\mathcal S(t)(Af)$ is defined as the operator product). Then, with $\mathcal T_q$ as in \eqref{Hilbert-Schmidt kernel equivalence}, we have $q\in \mathfrak F_{\gamma}$ iff $\mathcal T_q\in \mathfrak F_{\gamma}^{\mathfrak S}$, since
%$
%\|\mathcal T_q\|_{\mathfrak F_{\gamma}^{\mathfrak S}}^2=  
%      \|q\|^2_{\mathfrak F_{\gamma}}.
%   $
%\end{remark}
The classes $\mathfrak F_{\gamma}$ might appear abstract but, in particular, it contains H{\"o}lder spaces, that is,
\begin{equation}\label{Holder implies Favard}
   C^{\gamma}(\mathbb R_+^2)\subset \mathfrak F_{\gamma}.
   \end{equation}
   %and
%for all $M\geq 0$
%$$\|q\|_{\mathcal F_{\gamma}([0,M]^2)}\leq M\|q\|_{C^{\gamma}([0,M]^2)}$$
  %  \end{lemma}0
Vice versa, $ \mathfrak F_{\gamma}$ is not a subset of $ C^{\gamma}$  but it is strictly larger, allowing for discontinuities in volatility kernels:
Let $g(x,y):= \indicator_{[a,b]}(x)\indicator_{[a,b]}(y)$ for an interval $[a,b]\subset \mathbb R_+$. Then clearly, $g$ is not an element of $C^{\frac 12}(\mathbb R_+^2)$ as it is discontinuous.  However, it is %$   \|q\|_{\mathcal F_{1/2}([0,M]^2)}<\infty$ if and only if $\sup_{0<r<a}\int_0^M\int_0^M \frac{(\indicator_{[a,b]}(r+x)-\indicator_{[a,b]}(r+x))\indicator_{[a,b]}(y)}rdxdy<\infty$ and
   $ \|g\|_{\mathfrak F_{1/2}}
   %=  \sup_{r>0}\int_0^{\infty}\int_0^{\infty} \frac{(\indicator_{[a,b]}(r+x)-\indicator_{[a,b]}(r+x))\indicator_{[a,b]}(y)}{r}dxdy
      % =  (b-a)\sup_{r>0}\int_0^{\infty} \frac{\indicator_{[a-r,a]}(x)+\indicator_{[b-r,b]}(x)}{r}dx\\
       =  2(b-a) <\infty.$
   Hence $g \in \mathfrak F_{1/2}$, while $g\notin  \mathfrak F_{\rho}$ for any $\rho>1/2$.
   
   We now state our formal regularity assumption. 
\begin{assumption}\label{As: spatial regularity}[$\gamma$]
    Let $\gamma \in (0,1/2]$. We have $q^C_t\in \mathfrak F_{\gamma}$ $\mathbb P\otimes dt$-almost everywhere and 
    \begin{equation}
        \mathbb P\left[\int_0^T \|q_s^C\|_{\mathfrak F_{\gamma}}ds<\infty\right]=1,\quad T>0.
    \end{equation}
%    Equivalently, using the notation from Remark \ref{rem: Favard relation to the regularity condition} we can formulate this in terms of the volatility operator, for which $\Sigma_t\in \mathfrak F_{\gamma}^{\mathcal S}$ $\mathbb P\otimes dt$-almost everywhere and 
%      \begin{equation}\label{AS spatial regularity in terms of SIGMA}
 %       \mathbb P\left[\int_0^T \|\Sigma_s\|_{\mathfrak F_{\gamma}^{\mathcal S}}ds<\infty\right]=1.
%    \end{equation}
\end{assumption}
   
   \begin{remark}
   Regularity Assumption \ref{As: spatial regularity}  is sharp in Theorem \ref{T: Rate of convergence for discretized estimator} in the sense that for every $\gamma'<\gamma$ we can always specify a squared volatility process $(\Sigma_t)_{t\geq 0}$ such that in probability $\Delta_n^{-\gamma}\sup_{t\in [0,T]}\left\|\mathcal T_{\hat q^{n,-}_t}-[X^C,X^C]\right\|_{\text{HS}}$
diverges  but the process of kernels $(q_t^C)_{t\geq 0}$ fulfills Assumption \ref{As: spatial regularity} for $\gamma'$ and (and not for $\gamma$) 
(c.f. Example 3.6 in \cite{BSV2022}).
\end{remark}
As a result of Theorem \ref{T: Rate of convergence for discretized estimator} and \eqref{Holder implies Favard}, we can derive rates of convergence also under H{\"o}lder regularity assumptions.
\begin{corollary}
%          Let $q_s$ be the integral operator which corresponds to the nuclear operator $\Sigma_s$ for each $s\geq 0$. 
If Assumption \ref{As: H}(r) 
holds for some $r\in (0,2)$ and for all $t\geq 0$ it is $q_t\in C^{\gamma}(\mathbb R_+^2)$ $\mathbb P\otimes dt$-almost everywhere for some $\gamma\in (0,1/2]$, then \eqref{eq: Rate of HF estimator} holds for all $\rho<(2-r)w$ and \eqref{eq: Rate of HF estimator under stricter conditions} holds if  $r<2(1-\gamma)$ and  $w\in [\gamma/(2-r),1/2]$.
          \end{corollary}

For the central limit theorem, we further need 
\begin{assumption}\label{as: CLT}
It is almost surely
  \begin{equation}\label{slightly stronger than favard 1/2 condition}\int_0^T\sup_{r\geq 0}\frac{\|(I-\mathcal S(r))\sigma_s\|_{\text{op}}^2}{r} ds <\infty,\quad T>0.
  \end{equation}
\end{assumption}

\subsection{Assumptions for Long-time estimators}\label{Sec: Assumptions for Long-time estimators}
We introduce 
\begin{assumption}\label{As: Mean stationarity and ergodicity}
    The process $(\Sigma_t)_{t\geq 0}$ is mean stationary and mean ergodic, in the sense that $\mathbb E\left[\|\sigma_s\|^2_{L_{\text{HS}}(U,L^2(\mathbb R_+))}\right]<\infty$ and there is an operator $\mathcal C$ such that for all $t$ it is $\mathcal C=\mathbb E[ \Sigma_t]$   and as $T\to\infty$ we have in probability an w.r.t. the Hilbert-Schmidt norm that 
    \begin{equation}
        \frac 1T \int_0^T\Sigma_s ds=\frac{[X^C,X^C]_T}T\to \mathcal C.
    \end{equation}
\end{assumption}
Under Assumption \ref{As: Mean stationarity and ergodicity} we have that 
$\mathbb E[ (M_t^C)^{\otimes 2}]=\mathbb E[ (\int_0^t \sigma_s dW_s)^{\otimes 2}] = t \mathcal C\quad \forall t\geq 0.  $
Hence, $\mathcal C$ is the covariance of the driving continuous martingale $M^C$ (scaled by time).
Hence, as for regular functional principal component analyzes, we can find approximately a linearly optimal finite-dimensional approximation of the driving martingale, by projecting onto the eigencomponents of $\mathcal C$. Even more, $\mathcal C$ is the instantaneous covariance of the process $f$ in the sense that
$\mathcal C=\lim_{n\to\infty}\mathbb E[(f_{t+\Delta_n}-\mathcal S(\Delta_n)f_t)^{\otimes 2}]/\Delta_n.$
%Observe that thesemigroup-adjusted increment in the formulation above natural to use from an economic point of view, compared to regular increments without the adustment, since it correspond to   price increments of the same contracts (the contract with maturity $x+\Delta_n$ at time $t$ has maturity $x$ at time $t+\Delta_n$).
To estimate $\mathcal C$, we make use of a moment assumption for the coefficients.
\begin{assumption}\label{In proof: Very very Weak localised integrability Assumption on the moments}[p,r]
For $p,r>0$ such that $\mathbb E\left[\|\gamma_s(z)\|^r\right]=\Gamma(z)$ independent of $s$ for all $s\geq 0$ and
there is a constant $A>0$ such that for all $s\geq 0$ it is
$$ \mathbb E\left[\|\alpha_s\|_{L^2(\mathbb R_+)}^p+\|\sigma_s \|_{\text{HS}}^p+\int_{L^2(\mathbb R_+)\setminus \{0\}}\|\gamma_s(z)\|^r\nu(dz)\right] \leq A.$$
\end{assumption}
Moreover, we also make an assumption on the regularity of the volatility.
\begin{assumption}\label{In proof II}[$\gamma$]
With the notation \eqref{eq: notation for volatility kernel} we have for
 $\gamma\in (0,\frac 12]$ that
there is a constant $A>0$ such that for all $s\geq 0$ it is
$$ \mathbb E\left[\|q_s^C\|_{\mathfrak F_{\gamma}}\right] \leq A.$$
\end{assumption}

\section{Proofs of section \ref{Sec: bond market primer}}\label{Sec: Proof of the nonsemimartingality of volterra spot models}

\begin{proof}[Proof of the general nonsemimartingality of models in Example \ref{Ex: VMVP for forward curves}]
We need to prove that 
\begin{align}\label{Volterra form of forward curves}
    f_t=f_0 +\int_0^t k(\cdot+t-s)  d\beta_s
\end{align}
is not a continuous semimartingale where $k\in L^2(\mathbb R_+)$, $\beta$ is a univariate standard Brownian motion.
  Therefore, assume that  $f_t$ defines a semimartingale in $L^2(\mathbb R_+)$ of the form
  $f_t= A_t+M_t $
  for an $H$-valued continuous martingale $M$ and a finite variation process $A$.
  Observe that we also have that $f$ is a weak solution to the stochastic partial differential equation
    $$\frac d{dx} f_t dt+ (e\otimes k) dW_t, \quad t\geq 0,$$
    for a cylindrical Wiener process $W$ such that $\beta= \langle e, W\rangle$. Hence, for an orthonormal basis $(e_j)_{j\in \mathbb N}\subset D(d/dx)$ we find that
    $$\langle f_t, e_j \rangle = \langle f_0,e_j\rangle +\int_0^t \langle f_s, \left(\frac d{dx}\right)^* e_j\rangle ds + \langle k,e_j\rangle \beta_t.$$
    These are a one-dimensional semimartingales for which the first integral is of finite variation and the second part is of quadratic variation. As the decomposition of a continuous semimartingale into a continuous part with finite variation and a continuous martingale (which vanishes at $0$) with quadratic variation is unique up to $\mathbb P\otimes dt$ nullsets, we obtain that $\mathbb P\otimes dt$-almost everywhere
    $$\langle A_t,e_j\rangle= \int_0^t \langle f_s ,\left(\frac d{dx}\right)^*e_j\rangle\qquad \langle M_t,e_j\rangle= \langle k, e_j\rangle \beta_t\quad \forall t\geq 0, j\in \mathbb N.$$
    Therefore, we must have $M_t=\beta_t k$ and we must have $\sum_{i=1}^n \Delta_i^n f^{\otimes 2}=\sum_{i=1}^n \Delta_i^n M_t^{\otimes 2}\to k^{\otimes 2}$ in probability as $n\to \infty$. 
Defining
$$S_t^n := \sqrt n\langle f_t, \indicator_{[0,\Delta_n]}\rangle,$$
we also obtain that in probability
$$\left|\sum_{i=1}^n (\Delta_i^n S^n)^2-n\langle k, \indicator_{[0,\Delta_n]}\rangle^2\right|\leq% \|\sum_{i=1}^n \Delta_i^n f^{\otimes 2}-k^{\otimes 2}\|_{L_{\text{HS}}(L^2(\mathbb R_+)} \|\indicator_{[0,\Delta_n]}\|^2 n=
\|\sum_{i=1}^n \Delta_i^n f^{\otimes 2}-k^{\otimes 2}\|_{L_{\text{HS}}(L^2(\mathbb R_+)}\to 0 $$
and since as $n\to \infty$ it is
$\sqrt n \langle k, \indicator_{[0,\Delta_n]}\rangle%= \sqrt n \int_0^{\Delta_n} y^{H} dy=\frac{\sqrt n \Delta_n^{1+H}}{H+1}
=\Delta_n^{1/2+H}/(H+1)\to 0$
we also find that as $n\to\infty$ and in probability that
$\sum_{i=1}^n (\Delta_i^n S^n)^2\to 0$
must hold.
Moreover, we find, since the kernel $k$ is square integrable and $k(t)=t^H$ on $t\in [0,1]$ that by the Burkholder-Davis-Gundy inequality for $\epsilon>0$
\begin{align*}
   & \mathbb E[(\Delta_i^n S^n)^{2+\epsilon}]^{\frac 2{2+\epsilon}} \\
 %   \leq &  2 n\mathbb E\left[\langle\int_{(i-1)\Delta_n}^{i\Delta_n} k(i\Delta_n+\cdot-s) d\beta_s,\indicator_{[0,\Delta_n]}\rangle)^{2+\epsilon}\right]^{\frac 2{2+\epsilon}}\\
  %  &\qquad+2n\mathbb E\left[\langle\int_0^{(i-1)\Delta_n}k(i\Delta_n+\cdot-s) -k((i-1)\Delta_n+\cdot-s)d\beta_s,\indicator_{[0,\Delta_n]}\rangle)^{2+\epsilon}\right]^{\frac 2{2+\epsilon}}\\
    \leq & 2\int_{(i-1)\Delta_n}^{i\Delta_n} \|k(i\Delta_n+\cdot-s)\|^2 ds\\
    &\qquad+2n\int_0^{(i-1)\Delta_n}\langle k(i\Delta_n+\cdot-s) -k((i-1)\Delta_n+\cdot-s) ,\indicator_{[0,\Delta_n]}\rangle^2 ds\\
%    \leq &2 \int_0^{\Delta_n} \|k(\Delta_n+\cdot-s)\|^2 ds\\
%    &\qquad+2n\int_0^{(i-1)\Delta_n}(\int_0^{\Delta_n}(i\Delta_n+y-s)^H -((i-1)\Delta_n+y-s)^H dy)^2 ds\\
    \leq &2 \|k\|^2\Delta_n
    +2n\int_0^{(i-1)\Delta_n}\left(\int_0^{\Delta_n}(i\Delta_n+y-s)^H -((i-1)\Delta_n+y-s)^H dy\right)^2 ds\\
 %   = & 2\|k\|^2\Delta_n+2n\int_0^{(i-1)\Delta_n}\frac 1{(H+1)^2}\left(((1+i)\Delta_n-s)^{H+1} -2(i\Delta_n-s)^{H+1}+((i-1)\Delta_n-s)^{H+1}\right)^2 ds\\
     \leq  &2 \|k\|^2\Delta_n
    +2n\int_0^{(i-1)\Delta_n}\frac 4{(H+1)^2}\left((i\Delta_n-s)^{H+1}-((i-1)\Delta_n-s)^{H+1}\right)^2 ds.
\end{align*}
Now, using the mean value theorem and since $t^H$ is decreasing in $t$ we find 
\begin{align*}
    \mathbb E[(\Delta_i^n S^n)^{2+\epsilon}]^{\frac 2{2+\epsilon}} 
     \leq  &% 2\|k\|^2\Delta_n+\frac 8{(H+1)^2} n\int_0^{(i-1)\Delta_n}\left(\Delta_n ((i-1)\Delta_n-s)^{H} (H+1)\right)^2 ds\\= &  
    2\|k\|^2\Delta_n
    +8 \Delta_n\int_0^{(i-1)\Delta_n} ((i-1)\Delta_n-s)^{2H} ds\\
 %   = & 2\|k\|^2\Delta_n  +8 \Delta_n ((i-1)\Delta_n)^{2H+1} \frac 1{2H+1}\\
    \leq & 2\|k\|^2\Delta_n
    +8\Delta_n\frac 1{2H+1}.
\end{align*}
This shows in particular, that by Jensen's inequality we have
\begin{align*}
    \mathbb E\left[\left(\sum_{i=1}^n(\Delta_i^n S^n)^2\right)^{\frac {2+\epsilon}2}\right]\leq & n^{\frac {2+\epsilon}2-1}\sum_{i=1}^n\mathbb E[(\Delta_i^n S^n)^{2+\epsilon}]
%    \leq  & n^{\frac {2+\epsilon}2-1} n \left(2\|k\|^2\Delta_n+ 8 \Delta_n\frac 1{2H+1}\right)^{\frac {2+\epsilon}2}\\
     \leq \left(\|k\|^2
    +\frac 8{2H+1} \right)^{\frac {2+\epsilon}2},
\end{align*}
which shows that $\sum_{i=1}^n(\Delta_i^n S^n)^2$ is uniformly integrable. Thus,  convergence in probability must imply convergence of the mean and we must have
$$ \mathbb E[\sum_{i=1}^n(\Delta_i^n S^n)^2]\to 0\quad \text{ as } n\to\infty.$$
However, we can show similarly to the calculations before that using the mean value theorem it is
\begin{align*}
    \mathbb E\left[\sum_{i=1}^n(\Delta_i^n S^n)^2\right]%= & n\mathbb E\left[\langle\int_{(i-1)\Delta_n}^{i\Delta_n} k(i\Delta_n+\cdot-s) d\beta_s,\indicator_{[0,\Delta_n]}\rangle)^{2}\right]\\
    %&\qquad+n\mathbb E\left[\langle\int_0^{(i-1)\Delta_n}k(i\Delta_n+\cdot-s) -k((i-1)\Delta_n+\cdot-s)d\beta_s,\indicator_{[0,\Delta_n]}\rangle)^{2}\right]\\
    \geq  &  n\int_0^{(i-1)\Delta_n}\langle k(i\Delta_n+\cdot-s) -k((i-1)\Delta_n+\cdot-s),\indicator_{[0,\Delta_n]}\rangle^2 ds\\
   %  = &  n\int_0^{(i-1)\Delta_n}\frac 1{(H+1)^2}\left(((1+i)\Delta_n-s)^{H+1} -2(i\Delta_n-s)^{H+1}+((i-1)\Delta_n-s)^{H+1}\right)^2 ds\\
     \geq & 4 \Delta_n\int_0^{(i-1)\Delta_n}((1+i)\Delta_n-s)^{2H}) ds\\
   %  = &  4 \Delta_n((1+i)\Delta_n)^{2H+1}-(2\Delta_n)^{2H+1})  \frac 1{2H+1}\\
     = & \frac 4{2H+1} \Delta_n^{2H+2}((1+i)^{2H+1}-2^{2H+1}). \\
\end{align*}
Thus, writing $K= 4/(2H+1)H^2(H+1)^2$ we find
$$  \mathbb E[\sum_{i=1}^n(\Delta_i^n S^n)^2]\geq K \Delta_n^{2H+2} \sum_{i=1}^n(1+i)^{2H+1}-K \Delta_n^{2H+2} \sum_{i=1}^n 2^{2H+1}.$$
    While the second term is $o(1)$,  for the first term it is 
    $$K \Delta_n^{2H+2} \sum_{i=1}^n(1+i)^{2H+1}%\geq K \Delta_n^{2H+2} \sum_{i=1}^n i^{2H+1}
    \geq K \Delta_n^{2H+2} \int_0^{n+1} x^{2H+1} dx=\frac{K}{2H+2} \left(\frac {n+1}n\right)^{2H+2}\geq \frac{K}{2H+2}.$$
     This cannot hold, since by the uniform integrability of the sequence $\sum_{i=1}^n(\Delta_i^n S^n)^2$ we have that the mean $ \mathbb E[\sum_{i=1}^n(\Delta_i^n S^n)^2]$ must converge to $0$.
\end{proof}

\section{Proofs of Section \ref{Sec: term structure models theory}}
In this Section we prove the results of Section \ref{Sec: term structure models theory}. For that, we first prove an abstract limit theory for general evolution equations in Section \ref{Sec: AbstractLimit theorems}. We then derive the results of Section \ref{Sec: term structure models theory} using this abstract result in Section \ref{Formal proofs for term structure models}
\subsection{An Abstract limit theorem}\label{Sec: AbstractLimit theorems}
The asymptotic theory elaborated in the article follows by an abstract result for abstract evolution equations in Hilbert spaces, which we present and prove in this section.  Roughly speaking, we prove that the results in \cite{Schroers2024} are valid, also when we discretized the functional data also in the cross-section in a particular manner, that we will make precise next. For now let $f$ be a mild solution to a stochastic evolution equation of the for described in \eqref{mild Ito process}.

We also introduce the notation
$$H:=L^2(\mathbb R_+).$$
We do this, because the subsequent Theorem \ref{T: Theorems for abstract semigroups} holds under much more general conditions than for the term structure setting and with this notation it becomes simple to appreciate this generality.  That is,  Theorem \ref{T: Theorems for abstract semigroups} holds for general separable Hilbert spaces $H$,   semigroups $\mathcal S$ and general $H$-valued It{\^o} semimartingale as described in \cite{Schroers2024}. To be consistent with the notation and since we do not want to restate all Assumptions for the abstract case, (they can be found in \cite{Schroers2024} we formally chose to state the theorem and its proofs for the term structure setting only. 

For the cross-sectional discretization we  introduce a sequence of projections $(\Pi_m)_{m\in \mathbb N}$ that coverges strongly to a projection operator $\Pi:H\to H$, which is not necessarily the identity.  In the case of term structure models,  $\Pi_m$ is defined as in \eqref{Orthonomal Projection on Indicators} for which $\Pi f (x)=f(x)\indicator_{[0,M]}(x)$.
%If jumps occur this is still consistently estimating the quadrati c variation of the noise process, which however, is a sum of the integrated volatility and the jumps. To disentangle these two components,
We define the discretized truncated semigroup-adjusted realized covariation as
\begin{equation}
    SARCV_t^n(u_n,-,m):= \sum_{i=1}^{\ul} \Pi_m\tilde{\Delta}_i^n f^{\otimes 2}\indicator_{g_{n}(\Pi_m\tilde{\Delta}_i^n f)\leq  u_n}
\end{equation}
for $m,n\in \mathbb N \cup\{\infty\}$ and a sequence $(u_n)_{n\in \mathbb N}\subset \mathbb R\cup\{\infty\}$ and a sequence of truncation functions $g_{n}: L^2(\mathbb R_+) \to\mathbb R_+$, such that there are constants $c,C>0$ such that for all $f\in H$ we have
\begin{align}\label{abstract g conditions}
    c\|f\|_{H}\leq g_n(f)\leq  C\|f\|_{H}, \quad g_n(f+h)\leq g_n(h)+g_n(f)
\end{align}

%Analogously, we set
%\begin{equation}
%    RV_t^n(u_n,-,m):= \sum_{i=1}^{\ul} \Pi_m \Delta_i^n Y^{\otimes 2}\indicator_{g(\Pi_m\Delta_i^n Y)\leq  u_n}
%\end{equation}
Observe that if $\Pi=I$ is the identity on $H$,  it is  $SARCV(u_n,-,\infty)=SARCV(u_n,-)$ as in the previous section.
 As a consequence of the possibile noncommutativity of the semigroup and the projections $\Pi_m$, the rates of convergence also depends on 
\begin{equation}\label{Abstract Spatial Convergence Rate}
 b_m^T:= \int_0^T \|\Pi\Sigma\Pi-\Pi_m\Sigma_s\Pi_m\|_{\text{HS}}ds.   
\end{equation}
 Here we again use the notation $\Sigma_t=\sigma_t\sigma_t^*$ for $t\geq 0$.
That $b_m^T$ indeed converges to $0$ almost surely as $m\to\infty$ is a Corollary of Proposition 4 and Lemma 5 in \cite{Panaretos2019}.
\begin{theorem}\label{T: Theorems for abstract semigroups}
 \begin{itemize}
     \item[(i)]
     % It is as $n\to \infty$
    %  $$ SARCV_t^n\overset{u.c.p.}{\longrightarrow}\int_0^t \Sigma_s ds+\left( Y_s-Y_{s-}\right)^{\otimes 2} $$
      As  $n,m\to \infty$ and w.r.t. the Hilbert-Schmidt norm it is
    $$ SARCV_t^n(\infty,-,m)\overset{u.c.p.}{\longrightarrow}\Pi[X,X]_t\Pi=\int_0^t \Pi\Sigma_s\Pi ds+\sum_{s\leq t}\left(\Pi X_s-\Pi X_{s-}\right)^{\otimes 2}. $$
    %This holds, in particular, for $m=\infty$, and $\Pi=I$.
    \item[(ii)]  Under Assumption \ref{As: H}(2) and w.r.t. the Hilbert-Schmidt norm and as $n,m\to\infty$ it is
$$SARCV_t^n(u_n,-,m)\overset{u.c.p.}{\longrightarrow}\Pi[X^C,X^C]_t\Pi=\int_0^t \Pi\Sigma_s \Pi ds.$$
\item[(iii)]    %Let the Assumptions of Theorem \ref{T: General Limit discretized truncated LLN} hold, but with Assumption \ref{As: spatial regularity} being fulfilled both by the abstract semigroup $\mathcal S$ as well as the semigroup of left shifts.
Let Assumptions \ref{As: H}(r) hold for some $r\in (0,2)$ and Assumption \ref{As: spatial regularity}($\gamma$) hold for some $\gamma \in (0,1/2]$.
Then it is for each $\rho<(2-r)w$, $T\geq 0$ as $n,m\to \infty$
    $$\sup_{t\in [0,T]}\left\|SARCV_t^n(u_n,-,m)-\Pi[X^C,X^C]_t\Pi\right\|_{text{HS}}=\mathcal O_p(\Delta_n^{\min(\gamma,\rho)}+b_m^T)$$
    In particular, if $r<2(1-\gamma)$ and $w\in [\gamma/(2-r),1/2]$ we have  
       $$\sup_{t\in [0,T]}\left\|SARCV_t^n(u_n,-,m)-\Pi[X^C,X^C]_t\Pi\right\|_{\text{HS}}=\mathcal O_p(\Delta_n^{\gamma}+b_m^T)$$
 %   \item[(iv)]        If Assumptions \ref{As: H}(r) and \ref{As: spatial regularity}($\gamma$) hold for some $r\in (0,1)$, $\gamma \in (1/2,1]$, $w\in [1/(4-2r),1/2]$, it is
 %   Then we have
%    $$\left\|RV_t^n(u_n,-,m_n)-\int_0^t \Sigma_s ds\right\|_{L_{\text{HS}}(H)}=\mathcal O_p(\Delta_n^{\gamma-1/2}+b_n^T)$$
\item[(iv)] Assume that 
\begin{align}\label{Abstract CLT assumption}
 \mathbb P\left[   \int_0^T\sup_{r\geq 0}\frac{\|(I-\mathcal S(r))\sigma_s\|_{\text{op}}^2}{r} ds<\infty\right]=1.
\end{align}
Then Assumption \ref{As: spatial regularity}($1/2$) holds. Let, moreover, Assumption \ref{As: H}(r)  hold for $r<1$,  let $w\in [1/(2-r),1/2]$ and assume that $b_m^T=o_p(\Delta_n^{\frac 12})$. 
Then we have w.r.t. the $\|\cdot \|_{L_{\text{HS}}(H)}$ norm and as $n,m\to\infty$ that
  $$\sqrt n \left( SARCV_t^n(u_n,-,m)_t^n-\Pi[X^C,X^C]_t\Pi\right)\overset{st.}{\longrightarrow} \Pi\mathcal  N(0,\mathfrak Q_t)\Pi,$$
   where $\mathcal  N(0,\mathfrak G_t)$ is for each $t\geq 0$ a Gaussian random variable in $L_{\text{HS}}(H)$ defined on a very good filtered extension $(\tilde{\Omega},\tilde{\mathcal F},\tilde{\mathcal F}_t,\tilde {\mathbb P})$ of $(\Omega,\mathcal F,\mathcal F_t, \mathbb P)$ with mean $0$
  and covariance given for each $t\geq 0$ by a linear operator $\mathfrak Q_t:L_{\text{HS}}(H)\to  L_{\text{HS}}(H)$ such that
  $$\mathfrak Q_t  =\int_0^t \Sigma_s (\cdot+\cdot^*) \Sigma_s ds.$$
    \item[(v)] 
  
   Let Assumption \ref{As: Mean stationarity and ergodicity} hold and  $\mathcal C=\mathbb E[ \Sigma_t]$ denote the global covariance of the continuous driving martingale.
Let furthermore Assumption \ref{In proof: Very very Weak localised integrability Assumption on the moments}(p,r) and \ref{In proof II}($\gamma$) hold (for the abstract semigroup $\mathcal S$) for some $r\in (0,2)$, $\gamma \in (0,1/2]$
and $p>\max(2/(1-2w),(1-wr)/(2w-rw))$.
 Then we have w.r.t. the Hilbert-Schmidt norm that as $n,m,T\to\infty$
$$\frac 1T SARCV_T^{n}(u_n,-,m)\overset{p}{\longrightarrow}\Pi\mathcal C\Pi.$$
If $r<2(1-\gamma)$ and $w\in(\gamma/(1-2w),1/2)$, $p\geq 4$ and observing that $\varphi_m= \text{tr}((I-\Pi_m)\mathbb E[\Sigma_1](I-\Pi_m))$ converges to $0$ as $m\to\infty$ (where $\text{tr}$ denotes the trace operation) we have  with $a_T= \| [X^C,X^C]/T- \mathcal C\|_{\text{HS}}$ that
$$  \left\|\frac 1T SARCV_T^{n,m}(-)- \Pi \mathcal C\Pi\right\|_{\mathcal H}=\mathcal O_p(\Delta_n^{\gamma}+\varphi_m+a_T).$$
 \end{itemize}
\end{theorem}
To prove this abstract result,  we make use of the limit theory  established in \cite{Schroers2024}.  However, Theorem \ref{T: Theorems for abstract semigroups} is not a direct corollary of these results, since we have to take into account that jump-truncation rules now also depend on possible discrete approximations. 
The key result to bridge this gap is 
\begin{lemma}\label{L: Bridging Lemma}
Assume that Assumptions \ref{In proof: Very very Weak localised integrability Assumption on the moments}(p,r) holds for $r\in (0,2]$ and $p> (1-\rho/((2-r)w))^{-1}$ for some $\rho<(2-r)w$ when $r<2$ or $\rho = 0$ if $r=2$.  Then we have
\begin{align}\label{discretization of the truncation rule does no harm-moment based}
\mathbb E\left[\left\|\sum_{i=1}^{\ul} (\Pi_m \tilde \Delta_i^n f)^{\otimes 2} \indicator_{g_n(\Pi_m \tilde \Delta_i^n f)\leq  u_n}-\sum_{i=1}^{\ul} (\Pi_m \tilde \Delta_i^n f)^{\otimes 2} \indicator_{g_n( \tilde \Delta_i^n f)\leq  u_n}\right\|\right]=Kt\Delta_n^{\rho}\phi_n
\end{align}
for a real sequence $(\phi_n)_{n\in \mathbb N}$ converging to $0$ and a constant $K>0$. 

If Assumption \ref{As: H} holds, it is 
\begin{align}\label{discretization of the truncation rule does no harm}
\left\|\sum_{i=1}^{\ul} (\Pi_m \tilde \Delta_i^n f)^{\otimes 2} \indicator_{g_n(\Pi_m \tilde \Delta_i^n f)\leq  u_n}-\sum_{i=1}^{\ul} (\Pi_m \tilde \Delta_i^n f)^{\otimes 2} \indicator_{g_n( \tilde \Delta_i^n f)\leq  u_n}\right\|=o_p(\Delta_n^{\rho})
\end{align}
\end{lemma}
Before we prove this Lemma, let us introduce some notation. 
In the case that Assumption \ref{As: H}(r) is valid for $0<r\leq 1$ we write 
\begin{align*}
    f_t= \mathcal S(t)f_0+\int_0^t \mathcal S(t-s)\alpha_s' ds+\int_0^t \mathcal S(t-s)\sigma_s dW_s+\int_0^t \int_{H\setminus \{0\}} \mathcal S(t-s)\gamma_s(z) N(dz,ds),
\end{align*}
where 
$$\alpha_s'= \alpha_s-\int_{H\setminus \{0\}} \gamma_s(z) F(dz)$$
and the integral w.r.t. the (not compensated) Poisson random measure $N$ is well defined (for the second term recall the definition of the integral e.g. from \cite[Section 8.7]{PZ2007}). %, due to Assumption \ref{In proof: Very very Weak localised integrability Assumption on the moments}(p,r) with $0<r<1$.
We then define
\begin{align}\label{Eq: Continuous discontinuous condition for low r}
   & f_t':=\mathcal S(t) f_0+\int_0^t \mathcal S(t-s)\alpha_s'+\int_0^t \mathcal S(t-s)\sigma_sdW_s,\\
    &f_t'':=\int_0^t \int_{H\setminus \{ 0\}}\mathcal S(t-s)\gamma_s(z) N(dz,ds).\notag
\end{align}

If Assumption \ref{As: H}(r) holds for $r\in (1,2)$,  we define
\begin{align}\label{Eq: Continuous discontinuous condition for large r}
   & f_t':=\mathcal S(t) f_0+\int_0^t \mathcal S(t-s)\alpha_s+\int_0^t \mathcal S(t-s)\sigma_sdW_s,\\
    &f_t'':=\int_0^t \int_{H\setminus \{ 0\}}\mathcal S(t-s)\gamma_s(z) (N-\nu)(dz,ds).\notag
\end{align}
%Also in this case, $Y=Y'+Y''$ and $Y'$ is continuous. 

\begin{proof} 
We start with the case that Assumption \ref{In proof: Very very Weak localised integrability Assumption on the moments}(p,r) holds for $r\in (0,2]$ and $p>(1-\rho/((2-r)w))^{-1}$
Observe that 
\begin{align}
& \left\|\sum_{i=1}^{\ul} (\Pi_m \tilde \Delta_i^n f)^{\otimes 2} \indicator_{g_n(\Pi_m \tilde \Delta_i^n f)\leq  u_n}-\sum_{i=1}^{\ul} (\Pi_m \tilde \Delta_i^n f)^{\otimes 2} \indicator_{g_n( \tilde \Delta_i^n f)\leq  u_n}\right\|\notag\\
\leq  & \sum_{i=1}^{\ul} \left\|\Pi_m \tilde \Delta_i^n f\right\|^2 \left(\indicator_{g_n(\Pi_m \tilde \Delta_i^n f)\leq  u_n<g_n( \tilde \Delta_i^n f)}+\indicator_{g_n( \tilde \Delta_i^n f)\leq  u_n<g_n(\Pi_m \tilde \Delta_i^n f)}\right)\notag\\
\leq  & \sum_{i=1}^{\ul} \left\|\Pi_m \tilde \Delta_i^n f\right\|^2 \left(\indicator_{c\|\Pi_m \tilde \Delta_i^n f\|\leq  u_n<C\|\tilde \Delta_i^n f\|}+\indicator_{c\|\Pi_m \tilde \Delta_i^n f\|\leq  u_n< C\|\Pi_m \tilde \Delta_i^n f\|}\right)\notag\\
\leq & 2 c^2u_n^2 \sum_{i=1}^{\ul} \left(1\wedge \frac{\left\|\Pi_m \tilde \Delta_i^n f\right\|^2}{c^2u_n^2}\right) \indicator_{c\|\Pi_m\tilde \Delta_i^n f\|\leq u_n<C\|\tilde \Delta_i^n f\|}\notag\\
\leq & 2 \sum_{i=1}^{\ul} \left\|\Pi_m \tilde \Delta_i^n f'\right\|^2 \indicator_{c\|\Pi_m\tilde \Delta_i^n f\|\leq u_n<C\|\tilde \Delta_i^n f\|, c\|\tilde \Delta_i^n f'\|\leq u_n}\label{Auxterm1}\\
&\qquad+2  \sum_{i=1}^{\ul} \left\|\Pi_m \tilde \Delta_i^n f'\right\|^2 \indicator_{c\|\Pi_m\tilde \Delta_i^n f\|\leq u_n<C\|\tilde \Delta_i^n f\|, c\|\tilde \Delta_i^n f'\|> u_n}\label{Auxterm2}\\
&+ 2 c^2u_n^2 \sum_{i=1}^{\ul} \left(1\wedge \frac{\left\|\Pi_m \tilde \Delta_i^n f''\right\|^2}{c^2u_n^2}\right) \indicator_{c\|\Pi_m\tilde \Delta_i^n f\|\leq u_n<C\|\tilde \Delta_i^n f\|, \|\tilde \Delta_i^n f'\|\leq u_n}\label{Auxterm3}\\
&\qquad+2 c^2u_n^2 \sum_{i=1}^{\ul} \left(1\wedge \frac{\left\|\Pi_m \tilde \Delta_i^n f''\right\|^2}{c^2u_n^2}\right) \indicator_{c\|\Pi_m\tilde \Delta_i^n f\|\leq u_n<C\|\tilde \Delta_i^n f\|, \|\tilde \Delta_i^n f'\|> u_n}\label{Auxterm4}
%\leq & 2\sum_{i=1}^{\ul} \left\|\Pi_m \tilde \Delta_i^n f'\right\|^2 \indicator_{u_n<C\|\tilde \Delta_i^n f\|}+ 2 c^2u_n^2 \sum_{i=1}^{\ul} \left(1\wedge \frac{\left\|\Pi_m \tilde \Delta_i^n f''\right\|^2}{c^2u_n^2}\right) \indicator_{u_n<C\|\tilde \Delta_i^n f\|}
\end{align}
We show for all summands \eqref{Auxterm1}, \eqref{Auxterm2}, \eqref{Auxterm3} and \eqref{Auxterm4} that they are are bounded by $Kt\Delta_n^{\rho}\phi_n$ for a real sequence $(\phi_n)_{n\in \mathbb N}$ converging to $0$ and a constant $K>0$.

We start with \eqref{Auxterm1}.  Since Assumption \ref{In proof: Very very Weak localised integrability Assumption on the moments} holds, we can use Lemma A.1 from \cite{Schroers2024}. Since $\|\Pi_m\tilde\Delta_i^n f\| \leq  u_n$ and $\|\Pi_m\tilde\Delta_i^n f'\|\leq \|\tilde\Delta_i^n f'\|\leq  u_n$ implies that $\|\Pi_m\tilde\Delta_i^n f''\| \leq  2 u_n$ we find a constant $K>0$ such that
\begin{align*}
&\mathbb E\left[\sum_{i=1}^{\ul} \left\|\Pi_m \tilde \Delta_i^n f'\right\|^2 \indicator_{c\|\Pi_m\tilde \Delta_i^n f\|\leq u_n<C\|\tilde \Delta_i^n f\|, c\|\tilde \Delta_i^n f'\|\leq u_n}\right]\\
%\leq &\sum_{i=1}^{\ul} \mathbb E\left[\left\|\Pi_m \tilde \Delta_i^n f'\right\|^2 \left(1\wedge \frac{\|\tilde \Delta_i^n f''\|}{2u_n}\right)\right]\\
\leq &\sum_{i=1}^{\ul} \mathbb E\left[\left\|\Pi_m \tilde \Delta_i^n f'\right\|^p\right]^{\frac 2p} \mathbb E\left[ \left(1\wedge \frac{\|\tilde \Delta_i^n f''\|}{2u_n}\right)\right]^{1-\frac 2p}\\
%\leq & K \sum_{i=1}^{\ul} \Delta_n (\Delta_n^{1-rw})^{1-\frac 2p}\\
%\leq  & Kt \Delta_n^{\frac{p- 2}p(1-rw)}\\
\leq & Kt \Delta_n^{\rho}\phi_n
\end{align*}
For the second summand \eqref{Auxterm2}, we apply Markov's inequality, choose $l=(2-2rw)/(2-4w)>1$ and again Lemma A.1 from \cite{Schroers2024} to obtain a constant $K>0$ such that
\begin{align*}
& \mathbb E\left[ \sum_{i=1}^{\ul} \left\|\Pi_m \tilde \Delta_i^n f'\right\|^2 \indicator_{c\|\Pi_m\tilde \Delta_i^n f\|\leq u_n<C\|\tilde \Delta_i^n f\|, c\|\tilde \Delta_i^n f'\|> u_n}\right]\\
\leq & \sum_{i=1}^{\ul}  \mathbb E\left[\left\|\Pi_m \tilde \Delta_i^n f'\right\|^p\right]^{\frac 2p}\mathbb P\left[c\|\tilde \Delta_i^n f'\|> u_n\right]^{\frac{p-2}{p}}\\
%\leq & K t \Delta_n^{\frac{p-2}p(l/2-lw)}\\
%\leq & Kt \Delta_n^{\frac{p-2}p(1-rw)}||
\leq & Kt \Delta_n^{\rho}\phi_n
\end{align*}

Turning to the third summand, we again make use of Lemma A.1 from \cite{Schroers2024} to obtain a constanr $K>0$ and a real sequence $(\phi_n)_{n\in \mathbb N}$ convrging to $0$ such that
\begin{align*}
& \mathbb E\left[c^2u_n^2 \sum_{i=1}^{\ul} \left(1\wedge \frac{\left\|\Pi_m \tilde \Delta_i^n f''\right\|^2}{c^2u_n^2}\right) \indicator_{c\|\Pi_m\tilde \Delta_i^n f\|\leq u_n<C\|\tilde \Delta_i^n f\|, \|\tilde \Delta_i^n f'\|\leq u_n}\right]\\
 \leq &\mathbb E\left[c^2u_n^2 \sum_{i=1}^{\ul} \left(1\wedge \frac{\left\| \tilde \Delta_i^n f''\right\|^2}{c^2u_n^2}\right)^2\right]\\
 % \leq &\mathbb E\left[c^2u_n^2 \sum_{i=1}^{\ul} \left(1\wedge \frac{\left\| \tilde \Delta_i^n f''\right\|}{cu_n}\right)\right]\\
  \leq & K c^2t\Delta_n^{\rho}\phi_n.
\end{align*}
For the fourth summand we find for $1<q=(2-r)w/\rho$ if $r<1$ and $1<q$ arbitrary if $r=2$ and use once more Lemma A.1 from \cite{Schroers2024} to obtain a constant $K>0$ and a real sequence $(\phi_n)_{n\in \mathbb N}$ converging to $0$ such that
\begin{align*}
& c^2u_n^2 \sum_{i=1}^{\ul} \left(1\wedge \frac{\left\|\Pi_m \tilde \Delta_i^n f''\right\|^2}{c^2u_n^2}\right) \indicator_{c\|\Pi_m\tilde \Delta_i^n f\|\leq u_n<C\|\tilde \Delta_i^n f\|, \|\tilde \Delta_i^n f'\|> u_n}\\
\leq & c^2u_n^2 \sum_{i=1}^{\ul} \mathbb E\left[\left(1\wedge \frac{\left\|\Pi_m \tilde \Delta_i^n f''\right\|}{cu_n}\right)^q\right]^{\frac 1q}\mathbb P[|\tilde \Delta_i^n f'\|> u_n]^{\frac {q-1}q}\\
%\leq & K t c^2  \Delta_n^{\frac{(2-r)w}q}\phi_n\\
\leq & K t c^2  \Delta_n^{\rho}\phi_n.
\end{align*}
Summing up, we proved \eqref{discretization of the truncation rule does no harm-moment based}.

Let us now turn to the case that only Assumption \ref{As: H} holds. 
Assumption \ref{As: H} implies that  there is a localizing sequence of stopping times $(\rho_n)_{n\in\mathbb N}$ such that $\alpha_{t\wedge \rho_n}$ is bounded for each $n\in \mathbb N$. As $\sigma$ and $f$ are c{\`a}dl{\`a}g, the sequence of stopping times $\theta_n:=\inf\{s: \|f_s\|+\|\sigma_s\|_{\text{HS}}\geq n\}$ are localizing as well. If $(\tau_n)_{n\in\mathbb N}$ is the sequence of stopping times for the jump part as described in Assumption \ref{As: H}, we can define 
$\varphi_n:=\rho_n\wedge \theta_n\wedge \tau_n,\quad n\in\mathbb N.$
This defines a localizing sequence of stopping times, for which  the coefficients $\alpha_s \indicator_{s\leq \varphi_n}$, $\sigma_s \indicator_{s\leq \varphi_n}$ and $\gamma_s(z) \indicator_{s\leq \varphi_n}$ satisfy  Assumption \ref{In proof: Very very Weak localised integrability Assumption on the moments}(p,r) for $r\in (0,2)$ and all $p>0$. 

Now define 
\begin{align*}
Z_n(t):=&  \Delta_n^{-\rho}\sum_{i=1}^{\ul} \left\|\Pi_m \tilde \Delta_i^n f\right\|^2 (\indicator_{g_n(\Pi_m \tilde \Delta_i^n f)\leq  u_n}- \indicator_{g_n( \tilde \Delta_i^n f)\leq  u_n}). \\
\geq & \Delta_n^{-\rho}\left\|\sum_{i=1}^{\ul} (\Pi_m \tilde \Delta_i^n f)^{\otimes 2} \indicator_{g_n(\Pi_m \tilde \Delta_i^n f)\leq  u_n}-\sum_{i=1}^{\ul} (\Pi_m \tilde \Delta_i^n f)^{\otimes 2} \indicator_{g_n( \tilde \Delta_i^n f)\leq  u_n}\right\|.
\end{align*}
If $\varphi_n\geq t+1$, we have
\begin{align*}
Z_n(t\wedge \varphi_N)\leq &  \Delta_n^{-\rho}\sum_{i=1}^{\lfloor (t+1)/\Delta_n\rfloor} \left\|\Pi_m \tilde \Delta_i^n f_{\cdot\wedge\varphi_N}\right\|^2 (\indicator_{g_n(\Pi_m \tilde \Delta_i^n f_{\cdot\wedge\varphi_N})\leq  u_n}- \indicator_{g_n( \tilde \Delta_i^n f_{\cdot\wedge\varphi_N})\leq  u_n}).
%\Delta_n^{-\rho}\left\|\sum_{i=1}^{\ul} (\Pi_m \tilde \Delta_i^n f_{\cdot\wedge \varphi_N})^{\otimes 2} \indicator_{g_n(\Pi_m \tilde \Delta_i^n f_{\cdot\wedge \varphi_N})\leq  u_n}-\sum_{i=1}^{\ul} (\Pi_m \tilde \Delta_i^n f_{\cdot\wedge \varphi_N})^{\otimes 2} \indicator_{g_n( \tilde \Delta_i^n f_{\cdot\wedge \varphi_N})\leq  u_n}\right\|\\
\end{align*} 
%We can assume without loss of generality that $T\notin \mathbb N$
%otherwise we extend the process process and the coefficients up to time$T+\epsilon$ by setting $\alpha_{T+s})\alpha_T$ ..., such that the local boundedness assumption still holds
%and 
%Choose $N_n\in \mathbb N$ large such that $\mathbb P[\phi_{N_n}<t_n]<\delta$ for an arbitrary $\delta>0$.  
We obtain as $n\to \infty$
 \begin{align*}
 \lim_{n\to \infty}  \mathbb P\left[\sup_{t\in [0,T]}\mathcal Z_n^i(t)\geq  \epsilon\right]
   %= &  \lim_{n\to \infty}\mathbb P\left[\sup_{t\in [0,T]}\mathcal Z_n^i(t)\geq  \epsilon\right]\\
   \leq &\lim_{n\to \infty}  \mathbb P\left[\sup_{t\in [0,T]}\mathcal Z_n^i(t\wedge \varphi_N)
    \geq  \epsilon,  T< \varphi_N\right]+   \lim_{n\to \infty}\mathbb P\left[T \geq\varphi_N \right]=0
\end{align*}
where the convergence in the last line is due to  \eqref{discretization of the truncation rule does no harm-moment based} and Markov's inequality and since we know that $\Delta_n^{-\rho}\sum_{i=1}^{\lfloor (t+1)/\Delta_n\rfloor} \left\|\Pi_m \tilde \Delta_i^n f_{\cdot\wedge\varphi_N}\right\|^2 (\indicator_{g_n(\Pi_m \tilde \Delta_i^n f_{\cdot\wedge\varphi_N})\leq  u_n}- \indicator_{g_n( \tilde \Delta_i^n f_{\cdot\wedge\varphi_N})\leq  u_n})$ converges to $0$ uniformly on compacts by \eqref{discretization of the truncation rule does no harm-moment based}. This yields \eqref{discretization of the truncation rule does no harm}.
\end{proof}

Now we are able to prove Theorem \ref{T: Theorems for abstract semigroups} as a Corollary of the results in \cite{Schroers2024} and Lemma \ref{L: Bridging Lemma}.

\begin{proof}[Proof of Theorem \ref{T: Theorems for abstract semigroups}]
We start with assertion (i).  For that, we observe that
\begin{align*}
   \Pi_m SARCV_t^n \Pi_m-[\Pi X,\Pi X]_t= \left( \Pi_m SARCV_t^n\Pi_m-[\Pi_m X,\Pi_m X] \right)+\left([\Pi_m X,\Pi_m X]-[\Pi X,\Pi X]\right)
\end{align*}
For the first summand it is
\begin{align*}
    \left\| \Pi_m SARCV_t^n\Pi_m-[\Pi_m X,\Pi_m X]_t \right\|=   \left\| \Pi_m\left( SARCV_t^n-[ X,X]_t \right) \Pi_m\right\|\leq  \left\| SARCV_t^n-[ X,X]_t \right\|,
\end{align*}
which converges to $0$ as $n\to \infty$ uniformly on compacts in probability by Theorem 3.1 in  \cite{Schroers2024}. For the second summand, we have
\begin{align*}
 &  \sup_{t\in [0,T]} \left\|[\Pi_m X,\Pi_m X]-[\Pi X,\Pi X]\right\|\\
   \leq  & \int_0^T \|\Pi_m \Sigma_s \Pi_m-\Sigma_s\| ds+ \sum_{s\leq T} \|\Pi_m (X_s-X_{s-})^{\otimes 2}\Pi_m-(X_s-X_{s-})^{\otimes 2}\|.
\end{align*}
By dominated convergence, if we can prove that for all $s\geq 0$ it is as $m\to\infty$ and in probability that
\begin{align}
    \|\Pi_m \Sigma_s \Pi_m-\Sigma_s\|\to 0 \text{ and } \|\Pi_m (X_s-X_{s-})^{\otimes 2}\Pi_m-(X_s-X_{s-})^{\otimes 2}\|\to 0,
\end{align}
the proof follows. But  this holds true even as almost sure convergence, by Proposition 4 and Lemma 5 in \cite{Panaretos2019}.

Before we prove the remaining assertions, let us observe the subsequent error decomposition
\begin{align}
&SARCV(u_n,,-,m)_t^n-[\Pi X^C,[\Pi X^C]_t \notag\\
\leq &SARCV(u_n,,-,m)_t^n-\sum_{i=1}^{\ul} (\Pi_m\tilde \Delta_i^n f)^{\otimes 2}\indicator_{g_n(\tilde \Delta_i^n f)\leq u_n}\label{AUXEQ1}\\
& \qquad +\Pi_m\left(\sum_{i=1}^{\ul} (\tilde \Delta_i^n f)^{\otimes 2}\indicator_{g_n(\tilde \Delta_i^n f)\leq u_n}-[X^C,X^C]_t\right)\Pi_m\label{AUXEQ2}\\
& \qquad+  [\Pi_m X^C,\Pi_m X^C]_t-[\Pi X^C,\Pi X^C]_t\label{AUXEQ3}
\end{align}

We proceed with the proof of (ii).  By Lemma \ref{L: Bridging Lemma}, \eqref{AUXEQ1} converges to $0$.  The second summand \eqref{AUXEQ2} converges to $0$ by Theorem 3.2 in \cite{Schroers2024}. The last summand \eqref{AUXEQ3} is bounded by $b_m^T$,  which converges to $0$ as $m\to \infty$.

Let us now turn to the proof of (iii), which works analogous to the proof of (ii), by employing the decompositiion of the approximation error into \eqref{AUXEQ1}, \eqref{AUXEQ2} and \eqref{AUXEQ3}.  Indeed,  Lemma \ref{L: Bridging Lemma} yields that \eqref{AUXEQ1} is $o_p(\Delta_n^{\rho})$ with respect to the Hilbert-Schmidt norm, while Theorem 3.3 in \cite{Schroers2024} yields that the second summand is $\mathcal O_p(\Delta_n^{\min(\rho,\gamma)})$,with respect to the Hilbert-Schmidt norm, which shows (iii). 

Now let us prove the central limit theorem (iv).  Again employing the error decomposition into \eqref{AUXEQ1}, \eqref{AUXEQ2} and \eqref{AUXEQ3},  we find that, Lemma \ref{L: Bridging Lemma} yields that  \eqref{AUXEQ1} is $o_p(\Delta_n^{\rho})=o_p(\Delta_n^{1/2})$ and by Assumption, the same holds for  \eqref{AUXEQ3}, since it is bounded by .$b_m^T$.  Hence,  we find that under the Assumptions imposed in (iv), it is 
\begin{align*}
& \sqrt n\left(SARCV(u_n,,-,m)_t^n-[\Pi X^C,[\Pi X^C]_t\right)\\
= & \Pi_m \left(\sqrt n\left(\sum_{i=1}^{\ul} (\tilde \Delta_i^n f)^{\otimes 2}\indicator_{g_n(\tilde \Delta_i^n f)\leq u_n}-[X^C,X^C]_t\right)\right)+o_p(1)
\end{align*}
Now (iv) follows directly from Theorem 3.5 in \cite{Schroers2024}.

We conclude the proof by showing (v).  For that we introduce the decomposition
\begin{align}
&\frac 1TSARCV(u_n,,-,m)_T^n-\Pi \mathcal C \Pi \notag\\
\leq & \frac 1T SARCV(u_n,,-,m)_T^n-\frac 1T\sum_{i=1}^{\ulT} (\Pi_m\tilde \Delta_i^n f)^{\otimes 2}\indicator_{g_n(\tilde \Delta_i^n f)\leq u_n}\label{CAUXEQ1}\\
& \qquad +\Pi_m\left(\frac 1T\sum_{i=1}^{\ulT} (\tilde \Delta_i^n f)^{\otimes 2}\indicator_{g_n(\tilde \Delta_i^n f)\leq u_n}-\mathcal C\right)\Pi_m\label{CAUXEQ2}\\
& \qquad+\Pi_m \mathcal C \Pi_m-\Pi \mathcal C \Pi.\label{CAUXEQ3}
\end{align} 
By \ref{L: Bridging Lemma}, the first summand \eqref{CAUXEQ1} is $ o_p(\Delta_n^{\rho})$. The second summand \eqref{CAUXEQ2} converges to $0$ by  Theorem 3.6 in \cite{Schroers2024} and the third summand \eqref{CAUXEQ3} converges to $0$ as $m\to\infty$.  We obtain the rates of convergence also from Theorem 3.6 in \cite{Schroers2024} applied to \eqref{CAUXEQ2} and since $\rho$ can be chosen larger than $1/2$ if $r<1$ and  the last summand equals $\text{tr}((\Pi-\Pi_m)\mathcal C(\Pi-\Pi_m))$.
\end{proof}

\subsection{Formal proofs of Section \ref{Sec: term structure models theory}}\label{Formal proofs for term structure models}
We will now show how \ref{T: General discrete LLN}, Theorem \ref{T: General Limit discretized truncated LLN}, \ref{T: Rate of convergence for discretized estimator} and \ref{T: Long-time asymptotics for termstructure volatiltiy} can be deduced from Theorem \ref{T: Theorems for abstract semigroups}.
Let us begin with the general identifiability results.
\begin{proof}[Proof of Theorem \ref{T: General discrete LLN}]
    We have that the integral operator $\Delta_n^{-2}\mathcal T_{\hat q_t^{n}}$ corresponding to the piecewise constant kernel $\Delta_n^{-2}\hat q_t^{n}$, according to Remark \ref{rem: semigroup adjustments are difference returns} is given by
    $$\Delta_n^{-2}\mathcal T_{\hat q_t^{n}}= \Pi_{n,M}(SARCV_t^n)\Pi_{n,M}.$$
    where $\Pi_{n,M}$ is defined as in \eqref{Orthonomal Projection on Indicators}. Setting $\Pi_m= \Pi_{m,M}$ and $\Pi=I$ if $M=\infty$, or resp.$\Pi f(x)=\indicator_{[0,M]}(x)f(x)$ and $M<\infty$ for $f\in L^2(\mathbb R_+)$, the result follows immediately from Theorem \ref{T: Theorems for abstract semigroups}(i)
\end{proof}

\begin{proof}[Proof of Theorem \ref{T: General Limit discretized truncated LLN}]
Again using the notation of Remark \ref{rem: semigroup adjustments are difference returns}, we can observe that the integral operator $\Delta_n^{-2}\mathcal T_{\hat q_t^{n,M,-}}$ corresponding to the kernel $\Delta_n^{-2}q_t^{n,M,-}$ defined in Remark \ref{Rem: Reduction of the time to maturity is fine} is (with $\Pi_{n,M}$ as in Remark \ref{rem: semigroup adjustments are difference returns}) given by
    $$\Delta_n^{-2}\mathcal T_{\hat q_t^{n,M,-}}= (SARCV_t^n(u_n,-,n).$$
    Thus, setting $\Pi_m= \Pi_{m,M}$ and $\Pi=I$ if $M=\infty$, or resp.$\Pi f(x)=\indicator_{[0,M]}(x)f(x)$ and $M<\infty$ for $f\in L^2(\mathbb R_+)$, the result follows immediately from  \ref{T: Theorems for abstract semigroups}(ii).
\end{proof}

Before proving Theorems \ref{T: Rate of convergence for discretized estimator} and \ref{T: Long-time asymptotics for termstructure volatiltiy}, we observe that we can quantify the spatial discretization error now also in terms of the regularity of the semigroup.
 
\begin{theorem}\label{T: Uniform spatial error bound} %Let $q_s\in L^1([0,H]^2)$ denote the integral kernel  such that
%$$\Sigma_s f(x)= \int_0^H q_s(x,y)f(y) dy.$$
%      If \begin{align*}
%   \int_0^T \left(\int_0^H\int_0^H\int_0^H\int_0^H \frac {|q_s(x',y')-q_s(x,y)|^2}{\|(x',y')-(x,y)\|^{2+2\gamma}} dxdydx'dy'\right)^{\frac 12}ds=\mathcal O_p(1).
%\end{align*} 
%If Assumption \ref{As: spatial regularity}($\gamma$) for $\gamma\in (0,1/2]$ is valid, 
We have for all $\gamma >0$ that
\begin{align}\label{Error bound for spatial term structure discretization}
     \left\|\Pi_{n,M}\Sigma_s\Pi_{n,M}- \Sigma_s \right\|_{\text{HS}} \leq   2 \|\sigma_s\|_{\text{op}} \left\|\Pi_{n,M}\sigma_s- \sigma_s\right\|_{\text{HS}}\notag
 \leq     & 2\|\sigma_s\|_{\text{op}}\Delta_m^{\gamma} \sup_{r\leq \Delta_n}\frac{\left\|\left(\mathcal S(r)-I\right)\sigma_s\right\|_{\text{HS}} }{r^{\gamma}}.
\end{align}
Hence, if  Assumption \ref{As: spatial regularity}($\gamma$) for $\gamma\in (0,1/2]$ is valid, we find  (with $\Pi=I$ if $M=\infty$, or resp.$\Pi f(x)=\indicator_{[0,M]}(x)f(x)$ and $M<\infty$ for $f\in L^2(\mathbb R_+)$)
\begin{equation}
    \left\|\Pi_{n,M} \int_0^t \Sigma_s ds \Pi_{n,M}-\int_0^t \Pi\Sigma_s \Pi ds\right\|_{\text{HS}} =\mathcal O_p(\Delta_n^{\gamma})
\end{equation}
If even Assumption \ref{In proof II} holds, we find a constant $K$, which is independent of $T$ and $m$ such that
\begin{equation}
   \mathbb E\left[\sup_{t\in [0,T]} \left\|\Pi_{n,M} \int_0^t \Sigma_s ds \Pi_{n,M}-\int_0^t \Pi\Sigma_s\Pi ds\right\|_{\text{HS}} \right]\leq KT\Delta_m^{\gamma}.
\end{equation}
\end{theorem}
\begin{proof} Let $q_s^{\sigma}\in L^2(\mathbb R_+^2)$ denote the integral kernel  such that for all $f\in L^2(\mathbb R_+)$ and $x\geq 0$ it is
$$\sigma_s f(x)= \int_{\mathbb R_+} q_s^{\sigma}(x,y)f(y) dy.$$ 
Without loss of generality,  choose $q_s^{\sigma}$ to be symmetric.
Then for  $M=\infty$ and $M<\infty$ it is
   \begin{align*}
         \left\|(\Pi_{m,M}-\Pi)  \sigma_s   \right\|_{\text{HS}}^2
         %=&\int_{[0,M]} \sum_{j=1}^{\lfloor M/\Delta_m\rfloor}\int_{(j-1)\Delta_m}^{j\Delta_m} \left(\Delta_m^{-1}\int_{(j-1)\Delta_m}^{j\Delta_m} q_s^{\sigma}(x',y)dx'-q_s^{\sigma}(x,y) \right)^2 dx dy\\
         \leq &\int_{[0,M]} \sum_{j=1}^{\lfloor M/\Delta_m\rfloor}\int_{(j-1)\Delta_m}^{j\Delta_m} \Delta_m^{-1}\int_{(j-1)\Delta_m}^{j\Delta_m} \left(q_s^{\sigma}(x',y)-q_s^{\sigma}(x,y) \right)^2dx' dx dy\\
         \leq &2\Delta_m^{-1}\int_{[0,M]} \sum_{j=1}^{\lfloor M/\Delta_m\rfloor}\int_{(j-1)\Delta_m}^{j\Delta_m} \int_{x}^{j\Delta_m} \left(q_s^{\sigma}(x',y)-q_s^{\sigma}(x,y) \right)^2dx' dx dy\\
        %  = &2\Delta_m^{-1}\int_{[0,M]} \sum_{j=1}^{\lfloor M/\Delta_m\rfloor}\int_{(j-1)\Delta_m}^{j\Delta_m} \int_{0}^{j\Delta_m-x} \left(q_s^{\sigma}(x'+x,y)-q_s^{\sigma}(x,y) \right)^2dx' dx dy\\
%            \leq &2\Delta_m^{-1}\int_{[0,M]} \sum_{j=1}^{\lfloor M/\Delta_m\rfloor}\int_{(j-1)\Delta_m}^{j\Delta_m} \int_{0}^{\Delta_m} \left(q_s^{\sigma}(x'+x,y)-q_s^{\sigma}(x,y) \right)^2dx' dx dy\\
            = &2\Delta_m^{-1}\int_{0}^{\Delta_m}\int_{[0,M]} \sum_{j=1}^{\lfloor M/\Delta_m\rfloor}\int_{(j-1)\Delta_m}^{j\Delta_m}  \left(q_s^{\sigma}(x'+x,y)-q_s^{\sigma}(x,y) \right)^2 dx dydx'.      \end{align*}
      Hence, 
      \begin{align*}
       \left\|(\Pi_{m,M}-\Pi)  \sigma_s   \right\|_{\text{HS}}^2
           % \leq &2\Delta_m^{-1}\int_{0}^{\Delta_m}\int_{[0,M]} \sum_{j=1}^{\lfloor M/\Delta_m\rfloor}\int_{(j-1)\Delta_m}^{j\Delta_m}  \left(q_s^{\sigma}(x'+x,y)-q_s^{\sigma}(x,y) \right)^2 dx dydx'\\ 
      \leq & 2 \sup_{x\leq \Delta_m} \left(\Delta_m^{-1}\int_{[0,M]}\sum_{j=1}^{\lfloor M/\Delta_m\rfloor}\int_{(j-1)\Delta_m}^{j\Delta_m} \left(((\mathcal S(x)-I)q_s(\cdot, y'))(x'))\right)^2 dx'dy'\right)\\
         \leq & 2 \sup_{x\leq \Delta_m} \left(\int_{[0,M]}\int_{[0,M]} \left(\frac{((\mathcal S(x)-I)q_s(\cdot, y'))(x'))}x\right)^2 dx'dy'\right).
  %     \leq & 2 \Delta_m^{2\gamma}\sup_{x\leq \Delta_m} \left(\int_0^{M}\int_0^{M} \left(\frac{((\mathcal S(x)-I)q_s(\cdot, y'))(x'))}{x^{\gamma}}\right)^2 dx'dy'\right)\\
  %     \leq & 2 \Delta_m^{2\gamma}\sup_{x\leq \Delta_m} \left(\frac{\left\|\left(\mathcal S(x)-I\right)\sigma_s\right\|_{\text{HS}}}{x^{\gamma}}\right)^2
  %    \leq &2\Delta_m^2\int_{(j_1-1)\Delta_m}^{j_1\Delta_m}\int_{(j_1-1)\Delta_m}^{j_1\Delta_m} \int_{(j_1-1)\Delta_m}^{j_1\Delta_m}\int_{x'}^{j_2\Delta_m}  \left(q_s(x',y')-q(x,y'))\right)^2 dxdydx'dy' \\
    %   = &2\Delta_m^2\int_{(j_1-1)\Delta_m}^{j_1\Delta_m}\int_{(j_1-1)\Delta_m}^{j_1\Delta_m} \int_{(j_1-1)\Delta_m}^{j_1\Delta_m}\int_{0}^{j_2\Delta_m-x'}  \left(q_s(x',y')-q(x+x',y'))\right)^2 dxdydx'dy' \\
     %   = &2\Delta_m^2\int_{(j_1-1)\Delta_m}^{j_1\Delta_m}\int_{(j_1-1)\Delta_m}^{j_1\Delta_m} \int_{(j_1-1)\Delta_m}^{j_1\Delta_m}\int_{0}^{j_2\Delta_m-x'}  \left(((I-\mathcal S(x))q_s)(x',y')\right)^2 dxdydx'dy' \\
    %    = &2\Delta_m^2\int_{(j_1-1)\Delta_m}^{j_1\Delta_m}\int_{0}^{\Delta_m} \int_{(j_1-1)\Delta_m}^{j_1\Delta_m}\int_{(j_1-1)\Delta_m}^{j_1\Delta_m}  \left(((I-\mathcal S(x))q_s)(x',y')\right)^2 dx'dy'dxdy 
    \end{align*}
    %Analogously we obtain the same estimate for 
  %  \begin{align*}
  %     & \sum_{j_1,j_2}^{\lfloor H/\Delta_m\rfloor}\int_{(j_1-1)\Delta_m}^{j_1\Delta_m}\int_{(j_2-1)\Delta_m}^{j_2\Delta_m} \left(\int_{(j_1-1)\Delta_m}^{j_1\Delta_m}\int_{(j_2-1)\Delta_m}^{j_2\Delta_m} (q_s(x,y')-q(x,y)) dx'dy'\right)^2 dx dy\\
  %      \leq & 2 \Delta_m^{4+2\gamma}\sup_{x\leq \Delta_m} \left(\frac{\left\|\left(\mathcal S(x)-I\right)\Sigma_s\right\|_{L_{\text{HS}}(L^2(0,M))} }{x^{\gamma}}\right)^2
  %  \end{align*}
  %  Summing up, we obtain
  %  \begin{align*}
  %       &  \left\|\Pi_m  \Sigma_s  \Pi_m- \Sigma_s \right\|_{L_{\text{HS}}(L^2(0,M))} 
  %       \leq   2\Delta_m^{\gamma}\sup_{x\leq \Delta_m} \frac{\left\|\left(\mathcal S(x)-I\right)\Sigma_s\right\|_{L_{\text{HS}}(L^2(0,M))} }{x^{\gamma}}= 2\Delta_m^{\gamma} \|\Sigma_s\|_{F_{\gamma}^{\mathcal s}}
    %\end{align*}
    This proves the claim.

    \iffalse
    \begin{align*}
        &  \left\|\Pi_m  \Sigma_s  \Pi_m- \Sigma_s \right\|^2_{L_{\text{HS}}(H)}\\
          =& \int_0^H\int_0^H \left(\hat q_s^m(x,y)- q_s(x,y)\right)^2 dx dy\\
          = & \sum_{j_1,j_2}^{\lfloor H/\Delta_n\rfloor}\int_{(j_1-1)\Delta_n}^{j_1\Delta_n}\int_{(j_2-1)\Delta_n}^{j_2\Delta_n} \left(q(j_1,j_2)\Delta_n^{-2}-  q_s(x,y)\right)^2 dx dy\\
          = &\Delta_n^{-4} \sum_{j_1,j_2}^{\lfloor H/\Delta_n\rfloor}\int_{(j_1-1)\Delta_n}^{j_1\Delta_n}\int_{(j_2-1)\Delta_n}^{j_2\Delta_n} \left(\int_{(j_1-1)\Delta_n}^{j_1\Delta_n}\int_{(j_2-1)\Delta_n}^{j_2\Delta_n} (q_s(x',y')-q(x,y)) dx'dy'\right)^2 dx dy\\
          = &\Delta_n^{2\gamma} \sum_{j_1,j_2}^{\lfloor H/\Delta_n\rfloor}\int_{(j_1-1)\Delta_n}^{j_1\Delta_n}\int_{(j_2-1)\Delta_n}^{j_2\Delta_n} \int_{(j_1-1)\Delta_n}^{j_1\Delta_n}\int_{(j_2-1)\Delta_n}^{j_2\Delta_n} \frac{\left(q_s(x',y')-q(x,y)\right)^2}{\|(x,y)-(x',y')\|_{\mathbb R^2}^{2+2\gamma}} dx'dy' dx dy\\
          \leq & 2^{2\gamma}\Delta_n^{2\gamma} \int_0^H\int_0^H\int_0^H\int_0^H \frac{\left(q_s(x',y')-q(x,y)\right)^2}{\|(x,y)-(x',y')\|_{\mathbb R^2}^{2+2\gamma}} dx'dy' dx dy
    \end{align*}
    where the integral is finite by Assumption. The proves follows by dominated convergence.
    \fi
\end{proof}

It is left to show  Theorem \ref{T: Rate of convergence for discretized estimator},  Theorem \ref{T: CLT for truncated estimator} and Theorem \ref{T: Long-time asymptotics for termstructure volatiltiy}. We start with the
\begin{proof}[Proof of Theorem \ref{T: Rate of convergence for discretized estimator}]
Using again notation \eqref{Hilbert-Schmidt kernel equivalence} it is (with $\Pi_{n,M}$ as in Remark \ref{rem: semigroup adjustments are difference returns}) 
    $\Delta_n^{-2}\mathcal T_{\hat q_t^{n,M,-}}= SARCV_t^n(u_n,-,n)$  where $\hat q_t^{n,M,-}$ is defined in Remark \ref{Rem: Reduction of the time to maturity is fine} and we obtain from Theorem \ref{T: Theorems for abstract semigroups}(iii)
       that
\begin{equation*}
   \sup_{m\in \mathbb N\cup\{\infty\}}\sup_{t\in [0,T]}\left\|SARCV_t^n(u_n,-,n)-\int_0^t \Pi_{n,M}\Sigma_s \Pi_{n,M}ds\right\|_{\text{HS}}=\mathcal O_p\left(\Delta_n^{\min(\gamma,\rho)}\right).
\end{equation*}
Moreover,  due to Theorem \ref{T: Uniform spatial error bound} we obtain that with $\Pi=I$ if $M=\infty$, or resp.$\Pi f(x)=\indicator_{[0,M]}(x)f(x)$ and $M<\infty$ for $f\in L^2(\mathbb R_+)$ it is
\begin{equation*}
    \left\|\Pi_{n,M}\int_0^t \Sigma_s ds \Pi_{n,M}-\Pi\int_0^t \Sigma_s ds\Pi\right\|_{\text{HS}} = \mathcal O_p(\Delta_m^{\gamma}),
\end{equation*}
which proves the claim.
\end{proof}
We continue with the
\begin{proof}[Proof of Theorem \ref{T: CLT for truncated estimator}]
We first prove that \eqref{slightly stronger than favard 1/2 condition} implies Assumption \ref{As: spatial regularity}($1/2$). It is by H{\"o}lder's inequality and the basic inequality $=\|AB^*\|_{\text{HS}}\leq \|A\|_{\text{HS}}\|B\|_{\text{op}}$
for a Hilbert-Schmidt operator $A$ and a bounded linear operator $B$,
\begin{align*}
  \int_0^T \|q_s^C\|_{\mathfrak F_{\gamma}}ds%=  \int_0^T \|\Sigma_s\|_{\mathfrak F_{\gamma}^{\mathcal S}}ds
 % = &  \int_0^T \sup_{r>0}\frac{\|(I-\mathcal S(r))\Sigma_s\|_{L_{\text{HS}}(L^2(\mathbb R_+))}}{r^{\frac 12}}ds\\
%  \leq  &  \int_0^T \sup_{r>0}\frac{\|(I-\mathcal S(r))\sigma_s\|_{\text{op}}}{r^{\frac 12}}\|\sigma_s\|_{\text{HS}}ds\\
  \leq &\left(\int_0^T \sup_{r>0}\frac{\|(I-\mathcal S(r))\sigma_s\|_{\text{op}}^2}{r}ds\right)^{\frac 12}\left(\int_0^T \|\sigma_s\|_{\text{HS}}^2ds\right)^{\frac 12}.
\end{align*}
Now \eqref{slightly stronger than favard 1/2 condition} implies that the factor on the left is finite almost surely, whereas the factor on the right is finite almost surely, due to the stochastic integrability of the volatility. This implies that Assumption \ref{As: spatial regularity}($1/2$) is valid.

We now continue to derive Theorem \ref{T: CLT for truncated estimator} from Theorem \ref{T: Theorems for abstract semigroups}(iv). %One just has to observe that
We only have to show that $b_n^T=\int_0^T (\Pi-\Pi_{n,M})\Sigma_s ds=o_p(\Delta_n^{ 1/2})$. 
For that, observe that
\begin{align*}
 &   \Delta_n^{-\frac 12}\left\|\Pi_{n,M} \int_0^t \Sigma_s ds \Pi_{n,M}-\int_0^t \Pi\Sigma_s \Pi ds\right\|_{\text{HS}} \\
  \leq &  \Delta_n^{-\frac 12}\int_0^t\left\|(\Pi_{n,M}-\Pi) \Sigma_s \right\|_{\text{HS}}ds+ \Delta_n^{-\frac 12}\int_0^t\left\| \Sigma_s (\Pi_{n,M}-\Pi)\right\|_{\text{HS}}ds.
\end{align*}
We will prove convergence of the first summand to $0$ as $n\to \infty$, while for the second summand, the proof is analogous.
We  define the orthonormal basis $(e_j)_{j\in \mathbb N}\subset C_c(I)\subset L^2(I)$ where either $I=\mathbb R_+$ if $M=\infty$ and $I=[0,M]$ if $M<\infty$ and $C_c(I)$ is the set of compactly supported infinitely differentiable functions (which is dense in $L^2(I)$). Then, obviously for each $j$ we can find a constant $k_j$ such that $\sup_{|x-y|\leq \Delta_n}|e_j(x)-e_j(y)|\leq k_j \Delta_n$ and, thus, if $K_j\in \mathbb N$ such that $e_j(x)=0$ for all $x\geq K_j$
\begin{align*}
    \|(\Pi_{n,M}-\Pi)e_j\|^2=&\int_0^{K_j}\left(\sum_{i=1}^{\infty} n\int_{(i-1)\Delta_n}^{i\Delta_n} e_j(y)dy\indicator_{[(i-1)\Delta_n,i\Delta_n]}(x)-e_j(x)\right)^2dx
 %   =&\int_0^{K_j}\sum_{i=1}^{ K_j/\Delta_n}\left( n\int_{(i-1)\Delta_n}^{i\Delta_n} e_j(y)-e_j(x)dy\right)^2\indicator_{[(i-1)\Delta_n,i\Delta_n]}(x)dx\\
    %\leq &\int_0^{K_j}\sum_{i=1}^{ K_j/\Delta_n}\left( k_j \Delta_n\right)^2\indicator_{[(i-1)\Delta_n,i\Delta_n]}(x)dx\\
    \leq  k_j^2 \Delta_n^2 K_j.
\end{align*}
Let $P_N$ denote the orthonormal projection onto $span(e_i\otimes e_j:i,j=1,...,N)$. We can  decompose
%and define by $P_N$ the projecteion on the space $span(e_1,...,e_N)$. Then for the first summand above and using \eqref{Error bound for spatial term structure discretization} with $\Sigma_s$ exchanged respectively with $P_N\Sigma_s$ and  $(I-P_N)\Sigma_s$ where we denote the integral kernels of the latter two Hilbert-Schmidt operators we find 
\begin{align*}
  &   \Delta_n^{-\frac 12}\int_0^t\left\|(\Pi_{n,M}-\Pi) \Sigma_s \right\|_{\text{HS}}ds\\
    \leq & \Delta_n^{-\frac 12}\int_0^t\left\| P_N(\Pi_{n,M}-\Pi)\Sigma_s \right\|_{\text{HS}}ds+\Delta_n^{-\frac 12}\int_0^t\left\| (I-P_N)(\Pi_{n,M}-\Pi)\Sigma_s \right\|_{\text{HS}}ds
\end{align*}
It is simple to see that $\|P_N A\|_{\text{HS}}\leq \sum_{i,j=1}^N |\langle Ae_i,e_j\rangle|$ and, hence,
For the first part we find
\begin{align*}
   \Delta_n^{-\frac 12}\int_0^t\left\| P_N(\Pi_{n,M}-\Pi)\Sigma_s \right\|_{\text{HS}}ds
   %\leq     \Delta_n^{-\frac 12}\int_0^t \sum_{i,j=1}^N| \langle (\Pi_{n,M}-\Pi)\Sigma_s e_i,e_j\rangle| ds
%    \leq  &   \Delta_n^{-\frac 12}\int_0^t \|\Sigma_s \|_{\text{nuc}}\left(\sum_{j=1}^N\| (\Pi_{n,M}-\Pi)e_j\|\right)ds\\
       \leq & \Delta_n^{\frac 12}\int_0^t\| \Sigma_s \|_{\text{nuc}}ds\left( \sum_{j=1}^Nk_j \sqrt K_j\right).
\end{align*}
This converges to $0$ as $n\to\infty$ for all $N\in \mathbb N$.
For the second summand we observe that $I-P_N$ is the orthonormal projection onto $\overline{span(e_i\otimes e_j.i,j\geq N+1)}$  and hence can be written as $I-P_N=(I-p_N)(\cdot)(I-p_N)$ where $I-p_N=\sum_{i=N+1}^{\infty} e_i^{\otimes 2}$. 
We find by H{\"o}lder's inequality that
\begin{align*}
 &   \Delta_n^{-\frac 12}\int_0^t\left\| (I-P_N)(\Pi_{n,M}-\Pi)\Sigma_s \right\|_{\text{HS}}ds\\
  = &  \Delta_n^{-\frac 12}\int_0^t\left\| (I-p_N)(\Pi_{n,M}-\Pi)\Sigma_s (I-p_N)\right\|_{\text{HS}}ds
 % \leq  & \Delta_n^{-\frac 12}\int_0^t\left\| (I-p_N)(\Pi_n-I)\sigma_s\right\|_{L_{\text{HS}}(L^2(\mathbb R_+))}\| \sigma_s^* (I-p_N)\|_{\text{op}}ds\\
 % \leq & \Delta_n^{-\frac 12}\int_0^t\left\| (I-p_N)(\Pi_n-I)\Sigma_s (I-p_N)\right\|_{L_{\text{HS}}(L^2(\mathbb R_+))}ds
 \\
  \leq  & \Delta_n^{-\frac 12}\left(\int_0^t\left\| (I-p_N)(\Pi_{n,M}-\Pi)\sigma_s\right\|_{\text{op}}^2ds\right)^{\frac 12}\left(\int_0^t\|\sigma_s^*(I-p_N)\|_{\text{HS}}^2ds\right)^{\frac 12}
\end{align*}
The second factor converges to $0$ as $N\to\infty$ since
$$\| \sigma_s^* (I-p_N)\|_{\text{HS}}^2= \sum_{i=1}^{\infty}\| \sigma_s(I-p_N) e_j\|^2=\sum_{i=N+1}^{\infty}\| \sigma_s e_j\|^2$$
converges to $0$ as $N\to \infty$ and then the dominated convergence theorem applies. The first factor is bounded, since %by \eqref{Error bound for spatial term structure discretization} we have that
\begin{align*}
    \int_0^t\left\| (I-p_N)(\Pi_{n,M\Pi}-\Pi)\sigma_s\right\|_{\text{op}}^2ds \leq & \int_0^t\left\| (\Pi_{n,M}-\Pi)\sigma_s\right\|_{\text{op}}^2ds\\
    \leq &2\Delta_n \int_0^t \sup_{r\leq \Delta_n}\frac{\left\|\left(\mathcal S(r)-I\right)\sigma_s\right\|_{L_{\text{HS}}(L^2(\mathbb R_+))} }{r}ds.
\end{align*}
This is finite by Assumption and summing up we obtain that as $N\to \infty$
$$ \sup_{n\in \mathbb N}\Delta_n^{-\frac 12}\int_0^t\left\| (I-P_N)(\Pi_{n,M}-\Pi)\Sigma_s \right\|_{L_{\text{HS}}(L^2(\mathbb R_+))}ds\to 0.$$
\end{proof}

Let us now conclude with the
\begin{proof}[Proof of Theorem \ref{T: Long-time asymptotics for termstructure volatiltiy}]
We use that as before, the for integral operator $\Delta_n^{-2}\mathcal T_{\hat q_t^{n,M,-}}$ corresponding to the kernel $\Delta_n^{-2}q_t^{n,M,-}$ defined in Remark \ref{Rem: Reduction of the time to maturity is fine} it is  $\Delta_n^{-2}\mathcal T_{\hat q_t^{n,M,-}}= (SARCV_t^n(u_n,-,n)$(with $\Pi_{n,M}$ as in Remark \ref{rem: semigroup adjustments are difference returns}). We obtain under the Assumption of Theorem \ref{T: Long-time asymptotics for termstructure volatiltiy} that by Theorems \ref{T: Theorems for abstract semigroups}(v) and Theorem \ref{T: Uniform spatial error bound}
      there is a constant $K>0$, which is independent of $T$ and $n$ such that
\begin{equation*}
    \mathbb E\left[\sup_{m\in \mathbb N\cup\{\infty\}}\sup_{t\in [0,T]}\left\|SARCV_t^n(u_n,-,n)-\int_0^t \Pi_{n,M}\Sigma_s \Pi_{n,M} ds\right\|\right]\leq K T \Delta_n^{\gamma}.
\end{equation*}
and
\begin{equation*}
   \mathbb E\left[\sup_{t\in [0,T]} \left\|\Pi_{n,M} \int_0^t \Sigma_s ds \Pi_{n,M}-\int_0^t \Pi\Sigma_s \Pi ds\right\|_{L_{\text{HS}}(L^2(0,M))} \right]\leq KT\Delta_n^{\gamma}.
\end{equation*}
Moreover, by Assumption we have that as $T\to\infty$
$$\frac 1T \int_0^T \Pi\Sigma_s \Pi ds\overset{p}{\longrightarrow}\Pi\mathcal C\Pi.$$
Hence, the claim follows since we can decompose
\begin{align*}
    &\frac 1T SARCV_T^n(u_n,-,n)-\mathcal C\\
    =&\frac 1T \left(SARCV_T^n(u_n,-,n)-\int_0^t \Pi_{n,M}\Sigma_s \Pi_{n,M} ds\right)+\frac 1T \int_0^T  \Pi_{n,M}\Sigma_s \Pi_{n,M} -\Pi\Sigma_s  \Pi ds\\ &\quad+ \frac 1T \int_0^t \Pi\Sigma_s \Pi ds -\Pi\mathcal C\Pi.
\end{align*}
\end{proof}

\section{Further  Practical considerations}
We now make some considerations for the practical implementation of the estimator here.  Precisely, we discuss the effects of smoothing the data a posteriori in the cross-sectional dimension and showcase a possible rescaling procedure for the truncation rule described in Section \ref{Sec: truncation in practice}.
\subsection{Ex-post smoothing}
For term structure models,  we might have strong beliefs that forward curves are continuous or even differentiable. While such smoothness Assumptions are reflected by better rates of convergence, the estimator $\hat q^n$ is discontinuous and we might want to derive a smooth approximation instead. A possible way to achieve this is to smooth the estimators a posteriori. This can also serve the purpose of an ex-post regularization to obtain more pleasing visual results or can favor the computational tractability of the estimator (a difference return curve with a daily resolution and 10 years maximally considered maturity needs to store approximately 2500 data points).
Hence, we might want to reduce the number of data points in the maturity direction in the sense of functional data analysis. That is, let $P_m$ be an orthonormal projection onto a finite-dimensional subspace of $L^2(0,M)$ which is spanned by the orthonormal vectors $e_1,....,e_m$. For instance, we could consider a spline basis, Fourier bases or just a lower resolution than daily (e.g. monthly) and let $P_m$ be the projection onto $\{\indicator_{[(j-1)\Delta_m,j\Delta_m]}/\sqrt{\Delta_{m}}:j=1,...,\lfloor M/\Delta_m\rfloor\}$ for $m=n*l$ for some $l\in \mathbb N$. In general, if $P_m$ is a continuous linear projection, we have
    $$\sup_{t\in [0,T]}\| \Delta_n^{-2} P_m \mathcal T_{\hat q_t^{n,-}} P_m-\int_0^t\Sigma_s ds\|\leq \sup_{t\in [0,T]}\|\Delta_n^{-2}\mathcal T_{\hat q_t^{n,-}}-\int_0^t\Sigma_s ds\|+\int_0^T\| P_m\Sigma_s P_m-\Sigma_s\| ds,$$
    so the additional error is quantified by the second summand on the left.
    As long as $\Pi_m \to I$ strongly, this converges to $0$ as $m\to \infty$ by  Proposition 4 and Lemma 5 in \cite{Panaretos2019}.
    The exact rate of convergence depends on the particular projection as well as the regularity of the volatility operator. It can be quantified by imposing further regularity assumptions on $(\Sigma_t)_{t\geq 0}$.
    An example is given next.
\begin{example}[Forward curves in reproducing kernel Hilbert spaces]
    Assume that $\Sigma_s$ maps into a reproducing kernel Hilbert space $H_k= \mathcal  T_k^{\frac 12} L^2(\mathbb R_+^2)\subset L^2(\mathbb R_+^2)$ where $k\in L^2(\mathbb R_+^2)$ is a kernel and $\mathcal  T_k$ is the corresponding positive definite integral operator with kernel $k$. 
    The space $H_k$ can be equipped with the norm
    $\|f\|_{H_k}=\|\mathcal T_k^{-\frac 12}f\|_{L^2(\mathbb R_+)}$. For instance, we might assume that $k(x,y)=\frac 1{a}(1+e^{-a\min(x,y)})$ for some $a>0$ corresponding to the forward curve space introduced by \cite{Filipovic2000}, which is also the space in which the nonparametrically smoothed yield curve data from \cite{FPY2022} are taken that we use for our empirical analysis in Section \ref{Sec: Empirical Study} . Such a kernel has a Mercer decomposition 
    $k(x,y)=\sum_{i=1}^{\infty}\lambda_i e_i(s)e_i(t)$ for an orthonormal basis $(e_i)_{i\in \mathbb N}$ of $L^2(\mathbb R_+)$ and corresponding positive eigenvalues $(\lambda_i)_{i\in \mathbb N}$. We might specify $P_m=\sum_{i=1}^m e_i^{\otimes 2}$ to be the orthonormal projection onto these basis functions.

    If we even have that $\Sigma_s\in L_{\text{HS}}(L^2(\mathbb R_+),H_k)$, and $\int_0^T \|\Sigma_s \|_{L_{\text{HS}}(L^2(\mathbb R_+),H_k)} ds<\infty$ almost surely, we obtain
\begin{align*}
    \int_0^T\| P_m\Sigma_s P_m-\Sigma_s\|_{L_{\text{HS}(L^2(\mathbb R_+))}} ds%\leq  2 \int_0^T  \|(I-P_m)\Sigma_s\|_{L_{\text{HS}(L^2(\mathbb R_+)}}ds
    %\leq & 2 \|(I-P_m)K^{\frac 12}\|_{L(L^2(\mathbb R_+))}\int_0^T  \|\Sigma_s\|_{L_{\text{HS}(L^2(\mathbb R_+),H_k)}}ds
   \leq  \lambda_{m+1}^{\frac 12}\int_0^T  \|\Sigma_s\|_{L_{\text{HS}(L^2(\mathbb R_+),H_k)}}ds, 
\end{align*}
which yields an additional $\mathcal O_p(\lambda_{m+1}^{\frac 12})$-error.
\end{example}

\subsection{Remarks on the scaling factor for preliminary estimators of the quadratic variation}\label{Sec: bias adjustment preliminary estimator}

In Section \ref{Sec: truncation in practice} we adjusted the truncated estimator $q_t^n(-)$ in the preliminary step by some $\rho^*>0$. 
As we do not know $\Sigma$, this correct scaling can be conducted in several ways. One reasonable possibility is to choose $\rho^*$ in such a way that the scaled truncated estimator coincides with another robust variance estimate for the data projected onto a particular linear functional. In the simple framework without drift and jumps and where $\Sigma$ is constant and independent of the driving Wiener process and $\Delta_n$ small, we have that $ \Delta_n^{-2}\sum_{i=1}^{\lfloor M/\Delta_n\rfloor} \tilde \Delta_{i\Delta_n} d(j\Delta_n)\indicator_{[(j-1)\Delta_n,j\Delta_n]}\approx \tilde \Delta_i^n f\overset{approx.}{\sim} N(0,\frac 1T \int_0^T \Sigma_s ds)$.  
Hence, we choose
$$\rho^*=\frac {\left(q_{.75}-q_{.25}\right)^2}{4\Phi^{-1}(0.75)^2\Delta_n\hat\lambda_1}$$
%$$\rho^*=\frac {q_{.75}-q_{.25}}{2\Phi^{-1}(0.75)\int_0^M\int_0^M \hat q^{n,-}_t(x,y) dxdy}$$
where $q_{.75}$, and resp. the $q_{.25}$, is the $0.75$-quantile  and resp.  the $0.25$-quantile, of the data $\sum_{i=1}^{\lfloor M/\Delta_n\rfloor} \tilde \Delta_{i\Delta_n} d(j\Delta_n)\langle \indicator_{[(j-1)\Delta_n,j\Delta_n]}, \hat e_1\rangle $, $i=1,...,\ulT$, $\hat \lambda_1$ and $\hat e_1$ are respectively the first eigenvalue and the first eigenvector of the preliminary estimator $\hat q_t^n$ and $\Phi^{-1}(0.75)$ is the $.75$ quantile of the standard normal distribution. %$\Delta_n\sum_{i=1}^{\lfloor M/\Delta_n\rfloor}d^n(i,\cdot)\approx \langle \indicator_{[0,M]},\tilde \Delta_i^n f\rangle, i=1,..., \ulT$.
In this way, 
%Indeed, in this simple framework and if the largest 25 percent of of the data as measured by the norm are also roughly the largest 25 percent of the data as measured by the level components $\left|\Delta_n\sum_{i=1}^{\lfloor M/\Delta_n\rfloor} d^n(i,\cdot)\right|, i=1,...,\ulT$, 
%we have that
%$\int_0^M\int_0^M \hat q^{n,-}_t(x,y) dxdy< 0.75\int_0^M\int_0^M \hat q_T^n(x,y) dxdy$ and the latter converges to $0.75 \int_0^M\int_0^Mq_T(x,y)dxdy$ in probability. 
 the rescaled estimator $\rho^* \hat q_t^n (-)$ projected onto $\hat e_1^{\otimes 2}$ corresponds to the interquartile estimator of the variance of the factor loadings of the first eigenvector $\hat e_1$, that is, $\hat \lambda_1$
 %the level component $\indicator_{[0,M]}^{\otimes 2}$
 %$\rho^*\int_0^M\int_0^M \hat q^{n,-}_t(x,y)dxdy/T$
 corresponds to the normalized interquartile range estimator
$$NIQR^2=\left(\frac {(q_{.75}-q_{.25})}{2\sqrt\Delta_n\Phi^{-1}(0.75)}\right)^2.$$
%which is consistent for $\langle \frac 1T \int_0^T \Sigma_t dt  \hat e_1,\hat e_1\rangle$ %$\langle \frac 1T \int_0^T \Sigma_t dt \indicator_{[0,M]},\indicator_{[0,M]}\rangle$
%in the simple setting of a constant and deterministic $\Sigma$ and the absence of drift and jumps. 

\section{Remarks on the simulation scheme}\label{Sec: remarks on the simulation scheme}
We here describe how to sample local averages
$F_{i,j}:=\langle \indicator_{[(j-1)\Delta_n, j\Delta_n]}, f_{i\Delta_n}\rangle_{L^2([0,10]}$
for $n=100$, $i=1,...,100$ and $j=1,...,1000$
of the forward curve process 
%\begin{align*}
%    f_t= \int_0^t \mathcal S(t-s) \alpha ds + \mathcal S(t-s) \sigma_s dW_s+ \int_0^t \mathcal S(t-s) dJ_s
%\end{align*}
described in section \ref{Sec: Simulation Study}.
% Due to their independence, we can simulate compomentwise. That is, we simulate the random vector of coefficients of the process $\Pi_{10} f_t=\sum_{j=1}^{\lfloor 10/\Delta\rfloor} \langle \indicator_{[(j-1)\Delta_n, j\Delta_n]}, f_{i\Delta_n}\rangle_{L^2([0,10]}\indicator_{[(j-1)\Delta_n, j\Delta_n]}$ where 
We use that for $i\geq 1$ it is
\begin{align*}
    F_{i,\cdot}\equiv & \Pi_{10} f_{i\Delta_n}=  \Pi_{10} \mathcal S(\Delta_n)f_{(i-1)\Delta_n}+\Pi_{10} \int_{(i-1)\Delta_n}^{i\Delta_n} \mathcal S(i\Delta_n-s)  dX_s= F_{i-1,\cdot+1}+\Pi_{10} \tilde \Delta_i^n f
    \end{align*}
    where $\tilde \Delta_i^n f= f_{i\Delta_n}-\mathcal S(\Delta_n)f_{(i-1)\Delta_n}$ as before. Conditional on the Ornstein Uhlenbeck process $x$, the adjusted increments are independent and we can simulate $F_{i,\cdot}, i=1,..., n$ iteratively by simulating $x$ and the increments $\Pi_{10} \tilde \Delta_i^n f$. For the latter we have in distribution (conditional on $x$)
       \begin{align*}
   \Pi_{10} \tilde \Delta_i^n f %&   \overset{d}{\equiv} %\Pi_{10}N\left(0, \int_{(i-1)\Delta_n}^{i\Delta_n}\mathcal S(i\Delta_n-s) \Sigma_s\mathcal S(i\Delta_n-s)^*ds\right) + \Pi_{10}\int_{(i-1)\Delta_n}^{i\Delta_n} \mathcal S(i\Delta_n-s) dJ_s    \\
   %  \overset{d}{=} & \Pi_{10}N\left(0,  \int_{(i-1)\Delta_n}^{i\Delta_n} x^2(s)ds Q_a\right) + \Pi_{10}\int_{(i-1)\Delta_n}^{i\Delta_n} \mathcal S(i\Delta_n-s) dJ_s    \\
  % = &   F_{i-1, \cdot+1}+\left(\Pi_{10} e^{-a_0 \cdot}\right)\int_0^{\Delta_n} e^{-a_0 (\Delta_n-s)} ds+\sqrt{\int_{(i-1)\Delta_n}^{i\Delta_n} x^2(s)ds}N\left(0,   \Pi_{10}Q_1\Pi_{10}\right)\\\
  % &+  N\left(0, \int_{0}^{\Delta_n} (\Pi_{10}e^{-a_1 (\cdot+\Delta_n-s)}|\cdot+\Delta_n-s|^{-\rho})^{\otimes 2}  ds\right)\\
   \overset{d}{=} &  \sqrt{\int_{(i-1)\Delta_n}^{i\Delta_n} x^2(s)ds}N\left(0,   \Pi_{10}Q_a\Pi_{10}\right)+ \Pi_{10}\int_{(i-1)\Delta_n}^{i\Delta_n} \mathcal S(i\Delta_n-s) dJ_s
\end{align*}
    \iffalse
    \begin{align*}
   \Pi_{10} \tilde \Delta_i^n f \overset{d}{\equiv} &   \Pi_{10}N\left(\int_{(i-1)\Delta_n}^{i\Delta_n}\mathcal S(i\Delta_n-s) \alpha ds, \int_{(i-1)\Delta_n}^{i\Delta_n}\mathcal S(i\Delta_n-s) \Sigma_s\mathcal S(i\Delta_n-s)^*ds\right)
    \\
    & \qquad + \Pi_{10}\int_{(i-1)\Delta_n}^{i\Delta_n} \mathcal S(i\Delta_n-s) dJ_s
    \\
     \overset{d}{=} & \Pi_{10}N\left(\int_{(i-1)\Delta_n}^{i\Delta_n} e^{-a_0 (\cdot+i\Delta_n-s)} ds,  \int_{(i-1)\Delta_n}^{i\Delta_n} x^2(s)ds Q_1+   \int_{(i-1)\Delta_n}^{i\Delta_n} (k_{\gamma}(\cdot+i\Delta_n-s)^{\otimes 2} 
   ds\right) \\
     & \qquad + \Pi_{10}\int_{(i-1)\Delta_n}^{i\Delta_n} \mathcal S(i\Delta_n-s) dJ_s
    \\
  % = &   F_{i-1, \cdot+1}+\left(\Pi_{10} e^{-a_0 \cdot}\right)\int_0^{\Delta_n} e^{-a_0 (\Delta_n-s)} ds+\sqrt{\int_{(i-1)\Delta_n}^{i\Delta_n} x^2(s)ds}N\left(0,   \Pi_{10}Q_1\Pi_{10}\right)\\\
  % &+  N\left(0, \int_{0}^{\Delta_n} (\Pi_{10}e^{-a_1 (\cdot+\Delta_n-s)}|\cdot+\Delta_n-s|^{-\rho})^{\otimes 2}  ds\right)\\
   \overset{d}{=} &   \left(\Pi_{10} e^{-a_0 \cdot}\right)\frac{1-e^{-a_0 \Delta_n}}{a_0}+\sqrt{\int_{(i-1)\Delta_n}^{i\Delta_n} x^2(s)ds}N\left(0,   \Pi_{10}Q_1\Pi_{10}\right)\\\
   &+  N\left(0, \int_{0}^{\Delta_n} (\Pi_{10}k_{\gamma}(\cdot+\Delta_n-s)^{\otimes 2} 
   ds\right) + \Pi_{10}\int_{(i-1)\Delta_n}^{i\Delta_n} \mathcal S(i\Delta_n-s) dJ_s
\end{align*}
\fi
where we used that $\mathcal S(t)Q_a\mathcal S(t)^*=Q_a$ for all $t\geq 0$. Moreover, we can identify the covariance $\Pi_{10}Q\Pi_{10}$ with covariance matrix
\begin{equation}\label{Discretized Gaussian covariance kernel}
    \Pi_{10} Q_1 \Pi_{10}\equiv \left[\int_{(j-1-1)\Delta_n}^{j_1\Delta_n}\int_{(j-2-1)\Delta_n}^{j_2\Delta_n} e^{-a(x-y)^2} dxdy \right]_{j_1,j_2=1,...,1000}.
\end{equation}%and
%$$\int_{0}^{\Delta_n} (\Pi_{10}k_{\gamma}(\cdot+\Delta_n-s)^{\otimes 2} 
 %  ds\equiv \varsigma_2:=\left[\int_0^{\Delta_n}\int_{(j-1-1)\Delta_n}^{j_1\Delta_n}\int_{(j-2-1)\Delta_n}^{j_2\Delta_n} k_{\gamma}(x+\Delta-s)k_{\gamma}(y+\Delta-s) dxdy ds\right]_{j_1,j_2=1,...,1000}.$$
 %  which is derived by numerical integration.
  % Moreover, for the drift part we have that
  % $$\Pi_{10} e^{-a_0 \cdot}\equiv\vartheta:=%\left[ \int_{(j-1)\Delta_n}^{j\Delta_n} e^{-a_0 y} dy\right]_{j=1,...,1000}=
  % \left[ \frac{e^{-a_0(j-1)\Delta_n}\left(1-e^{-a_0 \Delta_n}\right)}{a_0}\right]_{j=1,...,1000}.$$
   To have a good approximation of the integrals $\int_{(i-1)\Delta_n}^{i\Delta_n} x^2(s)ds$ we simulate the square root-process $x$ on a resolution of $10000$, allowing us to make an approximation of the integrals of $x$ with a Riemann sum of length $100$. 

For the jump part, we have $J=J_1+J_2$ where $J_1,J_2$ are two $L^2([0,10])$-valued compound Poisson processes, that is
   $J_t^i= \sum_{l=1}^{N^i_t} \chi^i_l,$ for $ i=1,2,$ and $t \geq 0$
   where $N^i$ are compound Poisson processes with intensities $\lambda_i$ and jumps $\chi_i\sim N(0,Q^{jump}_i)$, where
   $Q^{jump}_1= Q_{0.01}$ and  $Q^{jump}_2= K$ as described in section \ref{Sec: Simulation Study}. It is then, again since $\mathcal S(t)Q_1\mathcal S(t)^*=Q_1$ %  The corresponding kernel is chosen to produce a discontinuity at maturity $x=5$ and Therefore is corresponding to the covariance
   %$Q^{jump}_2= \indicator_{[0,5]}^{\otimes 2}$. Since $\Pi_{10}\mathcal S(t) Q_1^{jump}\mathcal S(t)^*\Pi_{10}=\Pi_{10}Q_2\Pi_{10}$ and $\mathcal S(t) Q_2^{jump}\mathcal S(t)^*=\indicator_{[0,5-t]}^{\otimes 2}$ for $t\leq \Delta$ we find that
    $$\Pi_{10}\int_{(i-1)\Delta_n}^{i\Delta_n} \mathcal S(i\Delta_n-s) dJ_s^1\overset{d}{=} \sum_{i=1}^{N_{\Delta_n}^1} \Pi_{10}\chi_i^1$$
    and
   $$\Pi_{10}\int_{(i-1)\Delta_n}^{i\Delta_n} \mathcal S(i\Delta_n-s) dJ_s^2\overset{d}{=} \sum_{i=1}^{N_{\Delta_n}^2} \Pi_{10}\chi_i^2(\cdot+ \Delta_n-\tau_i)=\sum_{i=1}^{N_{\Delta_n}^2} e^{-(\Delta-\tau_i)}\Pi_{10}\chi_i^2$$
where $\Pi_{10}\chi_i^1\sim N(0,\Pi_{10}Q_{0.01}\Pi_{10})$ and $\Pi_{10}\chi_i^2\sim N(0,\Pi_{10}K\Pi_{10})$ and where $\tau_i$ are the jump times at which $N_{\tau_i}-N_{\tau_i-}>0$.  As $\Pi_{10}Q_{0.01}\Pi_{10}$ can be identified with a matrix analogously to \eqref{Discretized Gaussian covariance kernel} and $\Pi_{10}K\Pi_{10}$, disregarding a normalization constant,  with the matrix
\begin{align*}
 \Pi_{10}K\Pi_{10}\equiv &\left[\int_{(j_1-1)\Delta_n}^{j_1\Delta_n}\int_{(j_2-1)\Delta_n}^{j_2\Delta_n} e^{-(x+y)} dxdy \right]_{j_1,j_2}
=\left[\left(\frac{1-e^{-10 \Delta_n}}{10}\right)^2 e^{-10(i+j-2)\Delta_n}\right]_{j_1,j_2}
\end{align*}
Hence, the adjusted increments $\Pi_{10} (f_{i\Delta_n}-\mathcal S(\Delta_n)f_{i\Delta_n})$ can be simulated exactly. %Observe, however, that we have to simulate our process on $L^2[0,11]$ and not $L^2[0,11]$, since we need to take into account that the convolutions in the deterministic and stochastic integrals take into account semigroup arguments of size up to $1$, such that the forward curve at $x=10$ will be affected by the forward curves between $0$ and $11$.

\section{Detailed results for the empirical analysis }\label{Sec: Detailed results}
We here provide the detailed results for the empirical study of Section \ref{Sec: Empirical Study} in Tables \ref{Tab: JUMPSTUDY} and \ref{Tab: Dimensionalitiesfull}.

\begin{table}\caption{
Columns $2$ to $4$ report the number of jumps detected by the estimator and the ratio of the norms of the truncated estimator $\hat q_i^{*,-}$ to the quadratic variation estimator $\hat q_i^*$, which indicates how large the impact of jumps was on the quadratic variation in each year. The number is bold, if at least one jump was detected. The norms of the quadratic variation estimators are reported in column $5$.
}
\label{Tab: JUMPSTUDY} 
\begin{tabular}{ c ccc c }
 \toprule
 Year &\multicolumn{3}{c}{Trunc. increments,$\frac{\|\hat q_i^{*,-}\|_{L^2}}{\|\hat q_i^*\|_{L^2}}$} & $\|\hat q_i^*\|_{L^2}$ \\
\cmidrule(lr{1em}){2-4}
&$l=3$& $l=4$& $l=5$ &\\
        \midrule  
 1990      & {\bf 1, 0.88} & 0, 1.00 & 0, 1.00 & {\bf 0.00096}\\
       1991       & 0, 1.00 & 0, 1.00 & 0, 1.00 & 0.00061\\
       1992   & 0, 1.00 & 0, 1.00 & 0, 1.00 & 0.00073 \\
      1993    & 0, 1.00 & 0, 1.00 & 0, 1.00 & 0.00067 \\
      1994     & {\bf 3, 0.85} & {\bf 2, 0.96} & {\bf 2, 0.96} & {\bf 0.00123} \\
      1995 & {\bf 1, 0.97} & 0, 1.00 & 0, 1.00 & {\bf 0.00074}\\
      1996      & {\bf 2, 0.89} & 0, 1.00 & 0, 1.00&{\bf 0.00108}\\
      1997      & {\bf 2, 0.98} & {\bf 2, 0.98} & 0, 1.00& {\bf 0.00066}\\
      1998      & {\bf 15, 0.55} & {\bf 9, 0.65} & {\bf6, 0.69}& {\bf 0.00110}\\
      1999     & 0, 1.00 & 0, 1.00 & 0, 1.00&0.00091\\
      2000      & 0, 1.00 &0, 1.00 & 0, 1.00& 0.00069\\
      2001       & {\bf 2, 0.94} & {\bf 2, 0.94} &{\bf 2, 0.94}&{\bf 0.00116}\\
      2002      & {\bf 2, 0.98} & 0, 1.00 & 0, 1.00&{\bf 0.00118}\\
      2003     & 0, 1.00 & 0, 1.00 & 0, 1.00&  0.00129\\
      2004      & 0, 1.00 & 0, 1.00 & 0, 1.00&0.00087\\
      2005     & 0, 1.00 & 0, 1.00 & 0, 1.00&0.00059\\
      2006    & {\bf 2, 0.97} & {\bf 2, 0.97} & {\bf 2, 0.97}&{\bf 0.00038}\\
      2007    & {\bf 4, 0.99} & 0, 1.00 & 0, 1.00&{\bf 0.00074}\\
      2008    & {\bf 4, 0.91} & 0, 1.00 & 0, 1.00&{\bf 0.00230}\\
      2009   & {\bf 1, 0.88} & 0, 1.00 & 0, 1.00&{\bf 0.00208}\\
      2010    & 0, 1.00 & 0, 1.00 & 0, 1.00& 0.00133\\
      2011     & 0, 1.00 & 0, 1.00 & 0, 1.00& 0.00157\\
      2012    & 0, 1.00 & 0, 1.00 & 0, 1.00&0.00071\\
      2013   & {\bf 2, 0.87} & 0, 1.00 & 0, 1.00&{\bf 0.00079}\\
      2014     & 0, 1.00 & 0, 1.00 & 0, 1.00&0.00047\\
      2015    & 0, 1.00 & 0, 1.00 & 0, 1.00&0.00085\\
      2016  & 0, 1.00 & 0, 1.00 & 0, 1.00&0.00059\\
      2017    & 0, 1.00 & 0, 1.00 & 0, 1.00&0.00037\\
      2018   & 0, 1.00 & 0, 1.00 & 0, 1.00&0.00033\\
      2019     & 0, 1.00 & 0, 1.00 & 0, 1.00&0.00050\\
      2020     & {\bf 9, 0.49} & {\bf 3, 0.73} & 0, 1.00&{\bf 0.00109}\\
        2021     & {\bf 2, 0.94} & 0, 1.00 & 0, 1.00&{\bf0.00056}\\
        2022     & 0, 1.00 & 0, 1.00 & 0, 1.00& 0.001561\\
      \bottomrule
\end{tabular}
\end{table}

\begin{table}\caption{Columns $2$ to $5$ report the numbers $D_{C}^{\hat e^{*,i}}(p)$ for $C=\mathcal T_{q_i^{*,-}},$ 
defined in \eqref{Dimensionality measure} of linear factors needed in each year to explain $p=85\%, 90\%, 95\%, 99\%$ of the variation of difference returns as measured by the truncated variation estimators $\hat q_i^{*,-}$ where the truncation rule was conducted with $l=3$ and $\hat e^{*,i}=(\hat e_1^{*,i},\hat e_2^{*,i},...)$ is the basis of eigenfunctions corresponding to the kernel $\hat q_i^*$. Columns $6$ to $9$ report $D_{C}^{\hat e^{long}}(p)$ for $C=\mathcal T_{q_i^{*,-}}$, which explain how many leading eigenvectors of the static estimator $\hat q_{long}^*$ are needed as approximating factors to explain the variation in all years separately. 
}
\label{Tab: Dimensionalitiesfull} 
\begin{tabular}{ c cccccccc }
 \toprule
 Year & \multicolumn{4}{c}{  $D_{\mathcal T_{\hat q_i^{*,-}}}^{\hat e^{*,i}}(p)$}& \multicolumn{4}{c}{$D_{\mathcal T_{\hat q_i^{*,-}}}^{\hat e^{long}}(p)$} \\
 \cmidrule(lr{1em}){2-5} \cmidrule(lr{1em}){6-9}
 &$0.85$&$0.90$&$0.95$& $0.99$ &$0.85$&$0.90$&$0.95$& $0.99$ \\
        \midrule  
 1990      &4& 6  & 9 & 15&$5$&$7$&$10$& $15$ \\
       1991      &5& 6  & 9 & 15&$5$&$7$&$10$& $15$ \\
       1992      &5& 6  & 8 & 14&$6$&$7$&$9$& $15$  \\
      1993      &4& 5  & 7 & 14&$4$&$5$&$9$& $14$ \\
      1994      &3& 5  & 8 & 14&$3$&$5$&$9$& $15$  \\
      1995      &4& 5  & 8 & 14&$4$&$6$&$9$& $15$ \\
      1996      &4& 5  & 8 & 13&$4$&$6$&$8$& $15$ \\
      1997      &3& 4  & 7 & 13&$3$&$4$&$9$& $14$\\
      1998      &5& 6  & 8 & 14&$5$&$6$&$10$& $15$ \\
      1999      &3& 5  & 8 & 14&$4$&$6$&$10$& $16$ \\
      2000      &4& 5  & 8 & 13&$4$&$6$&$10$& $14$\\
      2001      &4& 5  & 8 & 13&$5$&$6$&$9$& $13$ \\
      2002      &3& 5  & 7 & 12&$4$&$6$&$8$& $14$\\
      2003      &2& 4  & 6 & 10&$2$&$4$&$6$& $12$ \\
      2004      &2& 3  & 6 & 11&$2$&$4$&$7$& $13$\\
      2005      &2& 3  & 5 & 11&$2$&$3$&$8$& $12$ \\
      2006      &2& 2  & 4 & 10&$2$&$2$&$7$& $11$ \\
      2007      &2& 3  & 6 & 10&$3$&$6$&$10$& $13$ \\
      2008      & 3& 4  & 7 & 12&$3$&$5$&$9$& $13$ \\
      2009      &3& 4  & 6 & 11&$4$&$5$&$7$& $13$\\
      2010      &2& 3  & 6 & 11&$3$&$4$&$7$& $13$ \\
      2011      &2& 3  & 5 & 10&$2$&$4$&$6$& $12$\\
      2012      &1& 2  & 3 & 8 &$2$&$3$&$4$& $10$\\
      2013      &2& 2  & 3 & 8 &$2$&$3$&$5$& $10$\\
      2014      &2& 2  & 4 & 10 &$2$&$3$&$6$& $11$\\
      2015      & 2& 2  & 3 & 10 &$2$&$2$&$5$& $11$\\
      2016      & 2& 2  & 4 & 11 &$2$&$2$&$5$& $12$\\
      2017      &2& 2  & 5 & 11 &$2$&$3$&$6$& $12$\\
      2018      & 2& 2  & 5 & 12 &$2$&$3$&$7$& $13$\\
      2019      &2& 2  & 5 & 11 &$2$&$2$&$7$& $12$\\
      2020      &2& 3  & 6 & 12 &$2$&$4$&$8$& $13$\\
        2021      &2& 2  & 4 & 10 &$2$&$3$&$6$& $12$\\
        2022      &2& 2  & 4 & 8&$2$&$2$&$5$& $11$ \\
      \bottomrule
\end{tabular}
\end{table}
\end{appendix}

\end{document}